\documentclass[12pt,draftcls, onecolumn]{IEEEtran}
\usepackage[utf8]{inputenc}
\usepackage{amsmath,amsfonts,amssymb}
\usepackage{multicol}
\usepackage{comment}
\usepackage{booktabs}
\usepackage{float}
\usepackage{tikz}
\usepackage{amsthm}
\tikzset{>=stealth}
\usepackage{braket}
\usepackage{pgfplots}
\pgfplotsset{compat=1.18}
\usepgfplotslibrary{fillbetween}
\usetikzlibrary{spy}
\usepackage{physics}
\usepackage{cases}
\usepackage{bbm}
\usetikzlibrary{patterns}
\tikzstyle{vertex}=[auto=left,circle,fill=black!25,minimum size=20pt,inner sep=0pt]
\usetikzlibrary{decorations.pathreplacing,calligraphy}
\usepackage[linesnumbered,ruled,vlined]{algorithm2e}
\usepackage[short]{optidef}
\SetKwRepeat{Do}{do}{while}
\renewcommand{\vec}[1]{\ensuremath{\boldsymbol{#1}}}
\usepackage[thinc]{esdiff}
\newcommand{\velb}{\varepsilon}

\newcommand{\vedisc}{\omega}
\newcommand{\EE}{\mathbb{E}}
\newtheorem{remark}{Remark}
\newtheorem{definition}{Definition}
\allowdisplaybreaks
\title{Bandit-Based Rate Adaptation for a Single-Server Queue}
\author{Mevan Wijewardena, Kamiar Asgari, Michael J. Neely
\thanks{The authors are with the Electrical Engineering department at the University of Southern California.}
}
\date{December 2025}
\usepackage{pgfplots}
\newtheorem{theorem}{Theorem}
\newtheorem{lemma}{Lemma}
\newtheorem{corollary}{Corollary}[theorem]
\begin{document}
\maketitle
%%%%%Congestion Games
\begin{abstract}
This paper considers the problem of obtaining bounded time-average expected queue sizes in a single-queue system with a partial-feedback structure. Time is slotted; in slot $t$ the transmitter chooses a rate $V(t)$ from a continuous interval. Transmission succeeds if and only if $V(t)\le C(t)$, where channel capacities $\{C(t)\}$ and arrivals are i.i.d. draws from fixed but unknown distributions. The transmitter observes only binary acknowledgments (ACK/NACK) indicating success or failure. Let $\varepsilon>0$ denote a sufficiently small lower bound on the slack between the arrival rate and the capacity region.  We propose a \emph{phased} algorithm that progressively refines a discretization of the uncountable infinite rate space and, without knowledge of $\varepsilon$, achieves a $\mathcal{O}\!\big(\log^{3.5}(1/\varepsilon)/\varepsilon^{3}\big)$ time-average expected queue size uniformly over the horizon. We also prove a converse result showing that for any rate-selection algorithm, regardless of whether $\varepsilon$ is known, there exists an environment in which the worst-case time-average expected queue size is $\Omega(1/\varepsilon^{2})$. Thus, while a gap remains in the setting without knowledge of $\varepsilon$, we show that if $\varepsilon$ is known, a simple single-stage UCB type policy with a fixed discretization of the rate space achieves $\mathcal{O}\!\big(\log(1/\varepsilon)/\varepsilon^{2}\big)$, matching the converse up to logarithmic factors.
\end{abstract}

\section{Introduction}
Research on controlling queues in stochastic environments has received widespread attention in both networking and online learning communities. The classical work on this front assumed full feedback on the network conditions~\cite{Neely2010coml}. However, in many real-world systems, the network conditions are unknown and have to be estimated using partial feedback signals (e.g., ACK/NACK). This motivated the line of work known as \emph{queueing bandits}~\cite{Krishnasamy2016} that combines bandit learning with queue stability.

In the full feedback setting, the work on queue-stability focuses on achieving different forms of stability, such as mean-rate stability and strong stability~\cite{Neely2010coml}. Queueing bandit models extend these results to partial-feedback scenarios, where the system must balance learning unknown service characteristics with maintaining stable queues. The prior work on queueing bandits focused on finite action spaces. However, finite-armed bandit formulations are computationally expensive in certain applications such as rate selection in IEEE 802.11 systems~\cite{combes_optimal_2013} due to the sheer size of the action space. This motivates our continuum-armed formulation, where in each time slot the transmitter selects a rate $V(t)$ from the continuous rate space $[0,1]$ to serve the backlog of queued arrivals. The data has to be transmitted through a channel with unknown time-varying capacity, and the transmitter only receives binary (ACK/NACK) feedback indicating whether the transmission was a success. The goal is to achieve a uniformly bounded time-average expected queue size. 

Due to the continuous rate space, we cannot directly apply the techniques developed for classical queueing bandits in our setting. The line of work on \emph{continuum-armed bandits} extends the classical multi-armed bandit problem to handle continuous action spaces~\cite{Agrawal1995}. In continuum-armed bandits, the set of arms is indexed by a (possibly uncountable) subset of the real line, and each arm’s mean reward is a continuous function of its index.  To the best of our knowledge, our work is the first to integrate queueing with continuum-armed bandits. The continuum arm bandit problems are typically solved by picking a finite set of arms, where at least one of the picked arms guarantees a good reward~\cite{Kleinberg2003}. However, in our setting, the unknown arrival and service rates makes it impossible to fix any finite set of arms that guarantees queue stability---it is possible that none of the initially chosen arms stabilizes the system. Hence, the algorithm must adaptively refine the set of picked arms using the information learned on the arrival and service rates.

\noindent
\textbf{Contributions:} Below we list our major contributions. 
\begin{enumerate}
    \item We consider a novel formulation of the rate adaptation problem as a continuum-armed queueing bandit, where the transmitter chooses transmission rates from the continuous rate space $[0,1]$ and receives only binary feedback (ACK/NACK) indicating transmission success. In each time slot, arriving data are queued and the transmitter chooses a rate $V(t)$ to serve the backlog over a time-varying channel with unknown capacity. The objective is to ensure a uniformly bounded time-average expected queue size over the horizon.
    \item We design a phased UCB scheme that iteratively refines a discretization of the rate space [0,1] across phases. In each phase, we run an adaptation of the UCB1 algorithm from~\cite{Auer2002FinitetimeAO} on the current discretization. Let $\varepsilon > 0$ denote a lower bound on the gap between the arrival rate and the channel capacity. Without requiring prior knowledge of $\varepsilon$, the proposed algorithm guarantees a time-average expected queue size of order $\mathcal{O}(\log^{3.5}(1/\varepsilon)/\varepsilon^3)$ uniformly over the horizon.
    \item We establish a converse result showing that for any algorithm that chooses transmission rates, \textbf{whether or not it knows $\varepsilon$}, there exists an environment such that the worst-case time-average expected queue size is of the order $\Omega(1/\varepsilon^2)$. Thus, while our current algorithm achieves a time-average expected queue-size bound polynomial in $(1/\varepsilon)$, there remains a gap between the upper and lower bounds when $\varepsilon$ is unknown.
    \item  We establish that when the transmitter knows $\varepsilon$, adopting the UCB1 algorithm from~\cite{Auer2002FinitetimeAO} yields a time-average expected queue size of order $\mathcal{O}(\log(1/\varepsilon)/\varepsilon^2)$, uniformly over the horizon. This matches the converse bound, and hence the algorithm is optimal up to logarithmic factors when $\varepsilon$ is known. 
\end{enumerate}
\subsection{Related Work}
Network scheduling in stochastic environments has received widespread attention over the past few decades. This includes scheduling for vehicular networks~\cite{Cai2024}, unmanned aerial vehicle networks~\cite{Kong2021}, wireless networks~\cite{Jung2007,Tarzjani2025}, and computer networks~\cite{Maguluri2012StochasticMO}. The main goal of these works is to schedule to minimize power consumption~\cite{Cruz2003}, maximize utility~\cite{Palomar2006}, ensure fairness~\cite{Songwu1999}, and ensure queue-stability~\cite{Neely2010coml}. The above problems have additional challenges in the partial feedback settings~\cite{lyu2017optimal,li2020multi,Huang2024}. In the partial feedback setting, the above problems can be more generally captured under stochastic control problems with partial information~\cite{KAELBLING199899}. In addition to scheduling, these problems have applications in finance and pricing~\cite{Kleinberg2003,Chichportich2023}, resource allocation~\cite{zuo2021combinatorial}, smart grid~\cite{Shardin02012017}, trajectory planning~\cite{Sun2021}, and neuroscience~\cite{Velentzas117598}.

One of the most common partial feedback models is the multi-armed bandit (MAB) problem~\cite{Lai1985,Auer2002FinitetimeAO}. In its basic form, an agent repeatedly chooses from a finite set of arms, each associated with an unknown reward distribution. Upon selecting the arm, the agent observes a random reward drawn from the corresponding distribution. The agent's goal is to learn, over time, to identify and select the arm with the highest mean reward. This problem has the classic exploration vs. exploitation trade-off, where if the agent does not explore to learn the best arm, she may end up persistently choosing a suboptimal arm. However, exploration comes at a cost since the agent has to choose suboptimal arms when exploring. Hence, any suitable algorithm for the MAB problem must achieve a balance between the two~\cite{Cesa-Bianchi_Lugosi_2006, Bubeck2012}. Upper confidence bound-based algorithms are designed to handle the aforementioned exploration vs. exploitation tradeoff~\cite{Bubeck2012,lattimore_szepesvári_2020}. Beyond the classic stochastic model, numerous extensions of the MAB framework have been studied, including adversarial bandits~\cite{auer1995gambling,pmlr-v75-wei18a}, linear bandits, combinatorial bandits, and contextual bandits~\cite{Bubeck2012}. Multi-armed bandit problems are also extended to handle possibly uncountable infinite, continuous action spaces through the line of work known as continuum-armed bandits~\cite{Agrawal1995,Kleinberg2008,pmlr-v134-podimata21a}.

Queueing bandits that combines queueing with multi-armed bandits is also extensively studied in the past decade~\cite{Krishnasamy2020,Fu2023,Freund2023,Huang2024,Nguyen,Yang2023}. The work~\cite{Krishnasamy2016} studies a time-slotted multiple-server system in which arriving jobs are queued for service, and the number of arrivals in each time slot is independent and identically distributed (i.i.d.). In each time slot, the job at the head of the queue has to be assigned to one of the servers. If the service is successful, the job leaves the queue at the end of the time slot. The service distribution of each server is unknown, and in every time slot, the service outcome is drawn independently and identically from this distribution. The goal is to design an algorithm to minimize \emph{queue regret}, defined as the difference between the queue lengths under the considered algorithm and those under an oracle policy that knows the true service distributions. It was established in~\cite{Krishnasamy2016} that the \emph{queue regret} scales as $\tilde{\mathcal{O}}(1/t)$ with respect to time $t$, where $\tilde{\mathcal{O}}$ hides polylogarithmic factors. In terms of the traffic slackness $\varepsilon$, their analysis implies a time-average expected queue size of at least $\mathcal{O}(1/\varepsilon^2)$. The work of~\cite{Huang2024} relaxes the i.i.d. arrival and service assumptions by considering a dynamic environment, where arrival and service rates may vary subject to constraints. Meanwhile,~\cite{Fu2023} introduces a different model in which the incoming jobs are assigned to servers that maintain separate queues for the assigned jobs. Table~\ref{tab:comparison} provides a comparison of the worst-case time-average expected queue size of recent work on queueing bandits.

\begin{table}[t]
\centering
\caption{Comparison of the worst-case time-average expected queue size achieved by recent queueing bandit algorithms. Here, $\varepsilon$ denotes the traffic slackness, i.e., the gap between the arrival rate and the capacity region.}
\label{tab:comparison}
\small
\begin{tabular}{lccccc}
\toprule
\textbf{Work} & \textbf{Action Space} & \textbf{Environment} & \textbf{Upper Bound}&\textbf{Lower Bound}\\
\midrule
\textbf{This Paper} & Continuous   & Stochastic & $\mathcal{O}\left(\frac{\log^{3.5}(1/\varepsilon)}{\varepsilon^3}\right)$&$\Omega\left(\frac{1}{\varepsilon^2}\right)$ \\[3pt]
&  ($[0,1]$)& (Unknown $\varepsilon$)& &\\[3pt]
\hline
\textbf{This Paper} & Continuous    & Stochastic & $\mathcal{O}\left(\frac{\log(1/\varepsilon)}{\varepsilon^2}\right)$ &$\Omega\left(\frac{1}{\varepsilon^2}\right)$\\[3pt]
& ($[0,1]$)& (Known $\varepsilon$)& &\\[3pt]
\hline
Krishnasamy et al~\cite{Krishnasamy2020}& Discrete  &Stochastic& $\mathcal{O}\!\left(\frac{1}{\varepsilon^2}\right)$ &$\Omega\left(\frac{1}{\varepsilon}\right)$\\[3pt]
\hline
Yang et al~\cite{Yang2023}& Discrete & Stochastic & $\mathcal{O}\!\left(\frac{1}{\varepsilon^3}\right)$ &- \\[3pt]
\hline
Freund et al~\cite{Freund2023} & Discrete  & Stochastic &$\mathcal{O}\left(\frac{\log(1/\varepsilon)}{\varepsilon}\right)$ &$\Omega\left(\frac{1}{\varepsilon}\right)$\\[3pt]
\hline
Huang et al~\cite{Huang2024} & Discrete  & Adversarial & $\mathcal{O}\!\left(\frac{1}{\varepsilon^2}\right)$  &- \\
\bottomrule
\end{tabular}
\end{table}

Rate selection and adaptation has become one of the most important problems in communications, particularly in wireless systems such as IEEE 802.11~\cite{combes_optimal_2013,combes_optimal_2019,tang_joint_2021,cho_use_2023,tong_rate_2023,le_multi-armed_2024}. In each time interval, the transmitter selects a combination of parameters: module scheme, coding rate, guard interval, channel width, and number of spatial streams that jointly determine the attempted transmission rate for that slot. Given an attempted rate $r$, the transmission succeeds if and only if $r$ is no greater than the unknown instantaneous time-varying channel capacity. Let $R^{\text{max}}$ denote the maximum transmission rate that can be attempted. A possible approach is to model the rate selection problem as a finite armed bandit problem with action space $\{0,1,2,\dots,R^{\text{max}}\}$, and learn the unknown channel capacity. However, in practical schemes $R^{\text{max}}$ can be very large (typically between $10^7$ and $10^{10}$ in IEEE 802.11 schemes), which makes the finite armed bandit formulation above computationally expensive due to the sheer size of the decision space. This motivates our continuum formulation, where we choose $V(t)$ as an arbitrary real number in $[0,1]$. Here, the rates are normalized to the interval $[0,1]$ for analytical tractability, where 1 corresponds to the maximum achievable transmission rate $R^{\text{max}}$. With this approach, we avoid the need to exhaustively consider all possible discrete transmission rates in the vast set $\{0,1,2,\dots,R^{\text{max}}\}$ in each time slot.

\subsection{Notation}
For integers $n$ and $m$, we denote by $[n:m]$ the set of integers between $n$ and $m$ inclusive. If $m<n$, $[n:m]$ is the empty set. We use calligraphic letters to denote sets. Vectors and matrices are denoted in boldface characters. For a vector $\vec{x} \in \mathbb{R}^n$, and $k \in [1:n]$, $x_k$ denotes the $k$-th entry of $\vec{x}$. Likewise, for a matrix $\vec{M} \in \mathbb{R}^{ n \times m}$, $k \in [1:n]$, and $l \in [1:m]$, $M_{k,l}$ denotes the entry at the intersection of $k$-th row and $l$-th column of $\vec{M}$. For $\vec{x} \in \mathbb{R}^n$, define $[\vec{x}]_+$ to be the projection of $\vec{x}$ onto the nonnegative orthant. In particular, \([\vec{x}]_+ = \max\{\vec{x},\vec{0}\}\), where the max is taken entry-wise.
\section{System Model}\label{sec:sys_mod}
We consider a system with a single transmitter attempting to transmit over a single channel in discrete time slots $t \in \{1,2,\dots\}$. In time slot $t$, the transmitter receives $A(t)$ data units to be transmitted through the channel, and the channel has a time-varying capacity $C(t)$ supported in $[0,1]$. The transmitter chooses a rate $V(t)$, without knowing $C(t)$. Transmission is successful if only if $V(t) \leq C(t)$. If the transmission is successful, the transmitter transmits $V(t)$ units of data. The transmitter only gets feedback on whether the transmission is successful or not (i.e., $\mathbbm{1}\{V(t) \leq C(t)\}$). The data to be transmitted are queued on the transmitter's side. The queue evolves according to the following rule:
\begin{align}\label{eqn:basic_queeing equation}
    &Q(1) = 0, \text{ and } \ Q(t+1) = \Bigl[Q(t) + A(t) - V(t)\mathbbm{1}\{V(t) \leq C(t)\}\Bigr]_+ \text{ for all } t\geq 1. 
\end{align}
Our objective is to ensure a finite time-average expected queue size. Specifically, when the arrival process lies strictly within the system's capacity region, we aim to establish a constant \(G\)—which depends only on the fixed parameters of the problem—such that
\begin{align}\label{eqn:Main_goal}
    \frac{1}{H} \sum_{t=1}^H \mathbb{E}\{Q(t)\} \leq G
\end{align}
holds for all time horizons \(H \in \{1,2,3,\dots\}\).

We make the following assumptions:

\begin{enumerate}
    \item[\textbf{A1}] In each time slot \(t\), the random variables \(A(t)\) and \(C(t)\), both taking values in \([0,1]\), are drawn independently from distributions that are unknown to the transmitter. We define the average arrival rate as \(\lambda = \mathbb{E}\{A(t)\}\).

    \item[\textbf{A2}] There exists a maximizer \(r^* \in [0,1]\) of the function \(g : [0,1] \to [0,1]\) defined by
    \begin{align}\label{eqn:g_def}
        g(r) = r \, \mathbb{P}\{C(1) \geq r\}.
    \end{align}
    Furthermore, we assume that
    \begin{align*}
        g(r^*) - \varepsilon \geq \lambda
    \end{align*}
    for some \(\varepsilon > 0\).
\end{enumerate}
We begin with the following three lemmas.
\begin{lemma}\label{eqn:det_queue_bound}
    We have $Q(t) \leq t-1$ for all $t \in \mathbb{N}$. 
    \begin{proof}
        Notice that from the queueing equation~\eqref{eqn:basic_queeing equation}, we have for all $t \geq 1$,
        \begin{align}
             Q(t+1) = \Bigl[Q(t) + A(t) - V(t)\mathbbm{1}\{V(t) \leq C(t)\}\Bigr]_+ \leq \Bigl[Q(t) + A(t)\Bigr]_+ = Q(t) + A(t) \leq Q(t)+1. \nonumber
        \end{align}
        Combining the above with the fact that $Q(1) = 0$, we have the lemma.
    \end{proof}
\end{lemma}
\begin{lemma}[one-sided Lipschitz continuity]\label{lem:lip}
The function $g(r)$ satisfies the following one-sided 1-Lipschitz continuity property: For any $0\leq r_2\leq r_1\leq 1$ we have $g(r_1) - g(r_2) \leq r_1 - r_2$.
\end{lemma}
\begin{proof}
 Since $ r_2\leq r_1$, we have $\mathbb{P}\{C(1) \geq r_1\}\leq\mathbb{P}\{C(1) \geq r_2\}$. Thus, 
 \begin{align}
 g(r_1) - g(r_2) = r_1\mathbb{P}\{C(1) \geq r_1\}- r_2\mathbb{P}\{C(1) \geq r_2\}  \leq (r_1-r_2)\mathbb{P}\{C(1) \geq r_2\} \leq r_1 - r_2 \nonumber
 \end{align}
\end{proof}
\begin{lemma}\label{eqn:key_lemma_1}
Consider $d \in \mathbb{N}$ such that $d \geq 1/\varepsilon$. There exists $k^* \in [1:d]$ such that 
\begin{align}
g\left(k^*/d\right) -\lambda \geq \varepsilon - \frac{1}{d}, \nonumber   
\end{align}
where $g$ and $\varepsilon$ are defined in \eqref{eqn:g_def}, and $\lambda$ is defined in Assumption~\textbf{A1}.
\begin{proof}
Let $k_{\mathrm{low}} = \max\bigl\{\,k\in [1:d] : k/d \le r^{*}\,\bigr\}$, where $r^*$ is defined in Assumption~\textbf{A1}.
Since \(1/{d} \le \velb \le  g(r^*) = r^*\mathbb{P}\{C(1) \leq r^*\} \le r^{*}\), the index \(k_{\mathrm{low}}\) is well‑defined.

We first prove that
\begin{align}\label{eqn:int_res}
    g(r^*) - g\left(\frac{k_{\mathrm{low}}}{d}\right) \leq \frac{1}{d}. 
\end{align}
From the definition of $k_{\mathrm{low}}$, we have \(k_{\mathrm{low}}/d \le r^{*}\). Hence, if $k_{\mathrm{low}} = d$, we have $r^* = 1$, in which case \eqref{eqn:int_res} holds trivially. Therefore, we assume $k_{\mathrm{low}}  \in [1:d-1]$. By the Lipschitz property (Lemma~\ref{lem:lip}),
\begin{align}
 g(r^*)- g\left(\frac{k_{\mathrm{low}}}{d}\right) 
\leq
\left(r^{*} - \frac{k_{\text{low}}}{d}\right)
<_{(a)}\;
\frac{1}{{d}}. \nonumber
\end{align}
where (a) follows from the definition of $k_{\text{low}}$, since $(k_{\text{low}} +1)/d > r^*$. Hence, \eqref{eqn:int_res} holds. Now we complete the proof. Notice that
\begin{align}
g\left(\frac{k_{\text{low}}}{d}\right)
\geq_{(a)}
g(r^*) -
\frac{1}{{d}}
\geq_{(b)}
\lambda + \velb-
\frac{1}{d}, \nonumber
\end{align}
where (a) holds from \eqref{eqn:int_res}, and (b) follows from Assumption~\textbf{A2}. Hence, we are done.
\end{proof}
\end{lemma}

Now, we have the following corollary as a result of the above lemma.
\begin{corollary}\label{corr:main_corr}
    Fix $\gamma > 1$, and let $d \in \mathbb{N}$ such that $d \geq \gamma/\varepsilon$.  There exists $k^* \in [1:d]$ such that 
    \begin{align}
        g\!\left(\frac{k^*}{d}\right)-\lambda \geq \frac{\gamma-1}{\gamma}\,\varepsilon \nonumber
  \;>\;0. 
    \end{align}
\end{corollary}
In the paper, we consider two settings; when $\varepsilon$ is known (Section~\ref{sec:known}) and when $\varepsilon$ (Section~\ref{sec:unknown}) is unknown. Below, we briefly describe these two settings.

\subsection{Known $\varepsilon$.}
Fix \(\gamma>1\) and choose \(d=\lceil \gamma/\varepsilon\rceil\). We restrict rates to the grid \(\{1/d,2/d,\dots,1\}\); selecting \(V(t)=k/d\) induces a service process with mean \(g(k/d)\) (by \eqref{eqn:g_def}). By Corollary~\ref{corr:main_corr}, there exists \(k^*\in[1:d]\) such that $g\!\left(\frac{k^*}{d}\right)-\lambda>0$.
Hence, repeatedly using the rate \(k^*/d\) yields service strictly exceeding the arrival rate in expectation, implying a bounded time‐average queue size. Note that \(r^*\) need not equal \(k^*/d\). Define the \emph{rate levels}  $\mathcal{K}=\{1,2,\dots,d\}$, where level \(k\in\mathcal{K}\) corresponds to rate \(k/d\). The learning goal is then to identify $k^* \in \arg\max_{k\in\mathcal{K}} g(k/d)$. We carefully choose $\gamma$ to obtain the best bounds. We achieve this via the classical UCB algorithm; our main contribution in this setting is the technically rigorous analysis that yields tight bounds on the time-average expected queue size. This is addressed in Section~\ref{sec:known}.

\subsection{Unknown $\varepsilon$}\label{eqn:unknown_para}
When \(\varepsilon\) is unknown, \(d\) above cannot be chosen as a function of \(\varepsilon\). We therefore partition time into phases. In phase \(\ell\) (of length \(T_\ell\)), we consider the set of \emph{rate levels} $\mathcal{K}_\ell=\{1,2,\dots,d_\ell\}$
and restrict \(V(t)\in\{k/d_\ell: k\in\mathcal{K}_\ell\}\); in particular, the \(k\)-th level corresponds to rate \(k/d_\ell\). The idea then is to use an adaptation of the UCB1 algorithm from~\cite{Auer2002FinitetimeAO} with $\mathcal{K}_l$ as the set of arms (recall that by \eqref{eqn:g_def}, arm $k \in \mathcal{K}_l$ induces a service process with mean \(g(k/d)\)). We choose nondecreasing, unbounded sequences \(\{T_\ell\}_{\ell\ge1}\) and \(\{d_\ell\}_{\ell\ge1}\). Given $\gamma > 1$, for sufficiently large \(i\), we have \(d_i\ge \gamma/\varepsilon\), so if the algorithm selects near-optimal levels sufficiently often within each phase, the queue becomes stable from phase \(i\) onward. However, since the number of \emph{rate levels} increases for each phase, the exploration time required to learn near-optimal levels also increases. If $d_l$ grows too quickly, this exploration burden can lead to instability. Hence, to obtain sharp bounds, the sequences \(\{T_\ell\}\) and \(\{d_\ell\}\) must be chosen carefully; this is addressed in the Section~\ref{sec:unknown}. 

\noindent
\textbf{Organization.}
Section~\ref{sec:unknown} treats the unknown-\(\varepsilon\) case, Section~\ref{sec:converse} presents a converse result, and Section~\ref{sec:known} treats the known-\(\varepsilon\) case.
\section{Unknown $\varepsilon$}\label{sec:unknown}
In this section, we focus on developing the algorithm for the unknown $\varepsilon$ case. The algorithm takes in two tunable parameters $C \in (0,1)$ and $\delta \in (0,1/2)$. The algorithm proceeds in phases, where the $l$-th phase ($l \in \{1,2,\dots\}$) lasts for 
\begin{align}\label{eqn:t_l}
    T_l = 2^{l+2} 
\end{align}
time slots, and during the $l$-th phase we choose rates $V(t) \in \{k/d_l: k \in \mathcal{K}_l\}$, where 
\begin{align}\label{eqn:rate_levels}
    \mathcal{K}_l = \{1,2\dots,d_l\}
\end{align}
is the set of \emph{rate levels} in phase $l$, and 
\begin{align}\label{eqn:K_l}
    d_l= \left\lceil C T_l^{\left(\frac{1}{2}-\delta\right)}\right\rceil.
\end{align}
Now, we describe the motivation for the choice of $T_l,d_l$. The choice of $T_l$ follows from the standard doubling trick argument used in classic multi armed bandit algorithms. For the sequence $\{d_l\}$, the key requirement is that the number of \emph{rate levels} at time $t$ must grow slower than $\sqrt{t}$ in order to ensure stability. This condition will become evident during our analysis.

First, we define some notation. Let us denote by $T^{\text{sum}}_l$ the last time slot of the $(l-1)$-th phase. Hence,
\begin{align}\label{eqn:t_l_sum}
    T^{\text{sum}}_l = \sum_{i = 1}^{l-1}T_i
\end{align}
with the convention that $T^{\text{sum}}_1 = 0$. For $t \in \mathbb{N}$, let $a(t)$ denote the phase to which time slot $t$ belongs. In particular,
\begin{align}\label{eqn:a_t_def}
    a(t) = \min\{ l \in \mathbb{N}: T^{\text{sum}}_l \geq t\} - 1.
\end{align}
Throughout the analysis,  whenever we refer to a time slot using the number of the time slot in the phase, we will use the letter $u$ (i.e. $u$-th time slot of phase $l$). When we refer to the time slot using the number of the time slot in the overall time frame, we will use $t$. Hence, $u$-th time slot of the $l$-th phase is the $(T^{\text{sum}}_l + u)$-th time slot in the overall time frame, and $t$-th time slot of the overall time frame is the $(t-T_{a(t)})$-th time slot of the $a(t)$-th phase.
Let $K_l(u) \in \mathcal{K}_l$ denote the \emph{rate level} used during the $u$-th time slot of the $l$-th phase, where $\mathcal{K}_l = \{1,2,\dots,d_l\}$. Also, let $S_{l,k}(u) \in [0,1]$ denote the service received during the $u$-th time slot of the $l$-th phase if \emph{rate level} $k$ is used. In particular, 
\begin{align}\label{eqn:slt_klt_eqn}
S_{l,k}(u) = \frac{k}{d_l}\mathbbm{1}\left\{\frac{k}{d_l} \leq C(T^{\text{sum}}_l+u)\right\}.
\end{align}
for each $k \in \mathcal{K}_l$. Notice that $\mathbb{E}\{S_{l,k}(u)\} = g(k/d_l)$ (see the definition of function $g$ in \eqref{eqn:g_def}). For $l \geq 1$, and $k \in \mathcal{K}_l$, let us define
\begin{align}\label{eqn:mu_lk}
    \mu_{l,k} =  g\left(\frac{k}{d_l}\right)
\end{align}
for notational convenience. With this notation, the queueing equation can be written as
%change Q
\begin{align}\label{eqn:queing_equation}
    Q(T^{\text{sum}}_l+u+1) = \left[Q(T^{\text{sum}}_l+u) + A(T^{\text{sum}}_l+u) - S_{l,K_l(u)}(u)\right]_+.
\end{align}
For phase $l \in \{1,2,\dots\}$, \emph{rate level} $k \in \mathcal{K}_l$, and $u \in [0:T_l]$, we define the following. Let $N_{l,k}(u)$ denote the number of times the \emph{rate level} $k$ is chosen on or before the $u$-th time slot of phase $l$. In particular,
\begin{align}
    N_{l,k}(u) = \sum_{\tau = 1}^u \mathbbm{1}\{K_l(\tau) = k \}. \nonumber
\end{align}
Hence, $N_{l,k}(0) = 0$. Let $\bar{\mu}_{l,k}(u)$ denote the empirical mean of the $k$-th \emph{rate level} at time slot $u$ during the $l$-th phase. In particular,
\begin{align}
    \bar{\mu}_{l,k}(u) = \begin{cases}
        0 & \text{ if }  N_{l,k}(u) = 0\\
        \frac{\sum_{\tau = 1}^u \mathbbm{1}\{K_l(\tau) = k \}S_l(\tau)}{N_{l,k}(u)} & \text{ otherwise}
    \end{cases} \nonumber
\end{align}
In addition, we define
\begin{equation} \label{eq:ucb}
    \text{UCB}_{l,k}(u) = \bar{\mu}_{l,k}(u) + \sqrt{\frac{(7-2\delta)
\log\Bigl(T_l\Bigr)}{4\max\{1, N_{l,k}(u)\}}},
\end{equation}
which is an upper confidence bound of $\mu_{l,k}$ at the $u$-th time slot of the $l$-th phase. The constant $(7-2\delta)$ above is carefully chosen to obtain the best constants in the queue bound.
Now we are ready to introduce the algorithm. Algorithm~\ref{algo:2} summarizes the steps.

\begin{algorithm}[H]
\label{algo:2}
\caption{UCB for a Single-Queue Uniform Mesh Rate (Parameters $C,\delta$)}\label{algo:UCB_mesh}
\DontPrintSemicolon
\For{each phase $l \in \{1,2,\dots\}$}{
    \textbf{Initialization:}\\
    \quad For each $k \in \mathcal{K}_l$ ($\mathcal{K}_l$ is defined in \eqref{eqn:rate_levels}), set:
    \begin{itemize}
        \item $\bar{\mu}_{l,k}(0) \gets 0$,
        \item $N_{l,k}(0) \gets 0$, and
        \item $\text{UCB}_{l,k}(0) \gets \sqrt{\frac{(7-2\delta)\log(T_l)}{4}}$.
    \end{itemize}
    \For{each timeslot $u \in [1:T_l]$}{
        Set 
        \begin{align}\label{eqn:decision}
        K_l(u) \gets \arg\max_{k \in \mathcal{K}_l} \text{UCB}_{l,k}(u-1)
        \end{align}
        and run the \textsc{Update Subroutine}$(l,u)$.
    }
}
\end{algorithm}
\begin{algorithm}[H]
\label{alg:Update}
\SetAlgoLined
\DontPrintSemicolon
Update the number of samples for each arm $k \in [1:d_l]$:
\begin{align}
N_{l,k}(u) = N_{l,k}(t-1) + \mathbbm{1}\{K_l(u) = k\}. \nonumber
\end{align}\\

Update the sample mean for each arm $k \in [1:d_l]$:
\begin{align}
\bar{\mu}_{l,k}(u) \gets \frac{N_{l,k}(u-1)\,\bar{\mu}_{l,k}(u-1) + \mathbbm{1}\{K_l(u) = k\}S_{l}(u)}{N_{l,k}(u)}. \nonumber
\end{align}\\
Update $\text{UCB}_{l,k}(u)$ for each arm $k \in [1:d_l]$ according to \eqref{eq:ucb}.
\caption{Update Subroutine$(l,u)$}
\end{algorithm}

\subsection{Performance Bounds of the Algorithm}
The main goal of this section is to prove Theorem~\ref{thm:main_thm_bef_opt}, which establishes a time-average expected queue-size bound expressed in terms of the algorithm parameters $C$ and $\delta$, and the auxiliary parameters $p$, $q$, and $\gamma$. The result holds uniformly over all $\gamma > 1$, $q \in (1,2)$, and $p$ satisfying $1/p + 1/q = 1$. Corollary~\ref{thm:main_thm} then refines this bound by optimizing over these parameters. In addition, to the finite-time result, Theorem~\ref{thm:main_thm_bef_opt} also establishes that the limiting time-average expected queue size is of the order $\mathcal{O}(1/\varepsilon)$. 
\begin{theorem}\label{thm:main_thm_bef_opt}
     Running Algorithm~\ref{algo:UCB_mesh} with parameters $\delta \in (0,1/2)$ and $C \in (0,1)$ we have the following.
    \begin{enumerate}
        \item Consider a time horizon $H \in \mathbb{N}$. We have that,
    \begin{align}\label{eqn:llate_stage_2}
         &\frac{\sum_{t = 1}^H \mathbb{E}\{Q(t)\}}{H}\nonumber\\&\leq \frac{65\times2^{\frac{2}{p-1}}\gamma }{(\gamma-1) \varepsilon}+\frac{\left(2^{\frac{p+1}{p-1}}+2\right)\gamma^{\frac{2}{1-2\delta}}}{\varepsilon^{\frac{2}{1-2\delta}}C^{^{\frac{2}{1-2\delta}}}}+1 +\frac{2^{\frac{5q}{2}-\delta q+3}C^q\gamma^{2q}(7-2\delta)^q\log^{q+2}(2H)H^{1-\frac{q}{2}-\delta q}}{(\gamma-1)^{2q}\varepsilon^{2q}(1-(q/2))^2}  \nonumber\\&\ \ \ \ + \frac{2^{2q +3}\gamma^{2q}(7-2\delta)^q\log^{q+2}(2H)H^{1-q}}{(\gamma-1)^{2q}\varepsilon^{2q}(1-(q/2))^2}
    \end{align}
    for all $p,q,\gamma$ such that $q \in (1,2)$, $1/p+1/q = 1$, and $\gamma > 1$.
    \item We have
    \begin{align}\label{eqn:limit}
         \lim\sup_{H \to \infty}\frac{1}{H}\sum_{t = 1}^H \mathbb{E}\{Q(t)\} \leq  \frac{65\times 2^{\frac{2-4\delta}{1+2\delta}}}{\varepsilon}.
    \end{align}
   
    \end{enumerate}   
\end{theorem}
In Section~\ref{sec:analysis}, we focus on proving the preceding theorem. Examining the finite-time bound in the first part of the theorem, we observe that first three terms (including the 1) do not depend on the horizon $H$ and therefore remain uniformly bounded. The last term vanishes as $H \to \infty$ since $q \in (1,2)$. The fourth term can be made to vanish as $H \to \infty$ by choosing $q \in \left(\frac{1}{1/2+\delta},2\right)$, which is always feasible because $\delta \in (0,1/2)$. Combining these observations yields
\begin{align}\label{eqn:triv_bound}
    \frac{1}{H}\sum_{t = 1}^H \mathbb{E}\{Q(t)\} = \mathcal{O}\left(\frac{\log^{q+2}\left(\frac{1}{\varepsilon}\right)}{\varepsilon^{\max\left\{2q,\frac{2}{1-2\delta}\right\}}}\right)
\end{align}
which holds uniformly over the horizon provided that $1-\frac{q}{2}-\delta q < 0$. Optimizing \eqref{eqn:triv_bound} over $q$, $\delta$ under this constraint gives $\frac{1}{H}\sum_{t = 1}^H \mathbb{E}\{Q(t)\} = \mathcal{O}\left(\frac{\log^{3.5+\alpha}\left(\frac{1}{\varepsilon}\right)}{\varepsilon^{3+2\alpha}}\right)$ for any $\alpha >0$. However, when obtaining \eqref{eqn:triv_bound}, we neglected the fact that the last two terms of \eqref{eqn:llate_stage_2} vanish as $H \rightarrow \infty$ under $1-\frac{q}{2}-\delta q < 0$. By leveraging this fact and combining the bounds obtained from Theorem~\ref{thm:main_thm_bef_opt}-1 for two different values of $q$, we obtain a tighter scaling $\frac{1}{H}\sum_{t = 1}^H \mathbb{E}\{Q(t)\} = \mathcal{O}\left(\frac{\log^{3.5}\left(\frac{1}{\varepsilon}\right)}{\varepsilon^{3}}\right)$ that holds uniformly over the horizon. Corollary~\ref{thm:main_thm} summarizes this optimized result.
\begin{corollary}\label{thm:main_thm}
     Using $C = 0.04$, and $\delta = 1/6$, for all $H \in \mathbb{N}$, Algorithm~\ref{algo:UCB_mesh} satisfies
        \begin{align}\label{eqn:part__1}
            \frac{1}{H}\sum_{t=1}^{H} \mathbb{E}\{Q(t)\} \leq1 + \frac{267}{\varepsilon} +\frac{16846843}{\varepsilon^3} + \frac{2675\log^{3.5}\left(\frac{1}{\varepsilon}\right)}{\varepsilon^3}. 
        \end{align} 
        and
        \begin{align}\label{eqn:part_2}
              \lim\sup_{H \to \infty}\frac{1}{H}\sum_{t = 1}^H \mathbb{E}\{Q(t)\} \leq  \frac{130}{\varepsilon}.
        \end{align}
    \begin{proof}
        The limiting time-average result in \eqref{eqn:part_2} simply follows by plugging $\delta = 1/6$ in Theorem~\ref{thm:main_thm_bef_opt}-2. We prove \eqref{eqn:part__1} in Appendix~\ref{app:main_thm}. 
    \end{proof}
\end{corollary}

\subsection{Proof of Theorem~\ref{thm:main_thm_bef_opt}}\label{sec:analysis}

The goal of this section is to prove Theorem~\ref{thm:main_thm_bef_opt}. First, fix $\gamma,p,q \in \mathbb{R}$ such that $\gamma > 1$, $q \in (1,2)$, and $1/p+1/q = 1$. These are the variables appearing in \eqref{eqn:llate_stage_2}. Define
\begin{align}\label{eqn:b_def}
    b = \min\left\{l \in \mathbb{N}: d_l \geq \gamma/\varepsilon \right\}.
\end{align}
where $d_l$ defined in \eqref{eqn:K_l} is the number of \emph{rate levels} in phase $l$. We have the following lemma that bounds the number of time slots to reach phase $b$.
\begin{lemma}\label{lemma:to_b}
We have
\begin{align}
    T^{\text{sum}}_b <  2\left(\frac{\gamma}{\varepsilon C}\right)^{\frac{2}{1-2\delta}}, \nonumber
\end{align}
where $T_l^{\text{sum}}$ is defined in \eqref{eqn:t_l_sum}.
\begin{proof}
Notice that we can assume $b>1$ ($b = 1$ is trivial since $T^{\text{sum}}_1 = 0$). 
We have that
\begin{align}
    CT_{b-1}^{\left(\frac{1}{2}-\delta\right)} \leq d_{b-1} < \frac{\gamma}{\varepsilon} \nonumber
\end{align}
where the first inequality follows from the definition of $d_l$ in \eqref{eqn:K_l}, and the second inequality follows from the definition of $b$ in \eqref{eqn:b_def}. This gives
\begin{align}
    T_{b-1} \leq \left(\frac{\gamma}{\varepsilon C}\right)^{\frac{2}{1-2\delta}}. \nonumber
\end{align}
Hence,
\begin{align}
    T^{\text{sum}}_b =  \sum_{\tau = 1}^{b-1}T_{\tau}= \sum_{\tau = 1}^{b-1}2^{\tau+2} \leq 2^{b+2}  = 2T_{b-1} <  2\left(\frac{\gamma}{\varepsilon C}\right)^{\frac{2}{1-2\delta}}. \nonumber
\end{align}
Hence, we are done.
    \end{proof}
\end{lemma}
The following lemma serves as the building block in proving both parts of Theorem~\ref{thm:main_thm_bef_opt}. We first state the lemma. In Section~\ref{sec:main_thm_1}, we prove Theorem~\ref{thm:main_thm_bef_opt}-1 using the lemma. In Section~\ref{sec:main_thm_2}, we prove Theorem~\ref{thm:main_thm_bef_opt}-2 using the lemma. Finally, in Section~\ref{sec:main_lemma}, we prove the lemma.
\begin{lemma}\label{lemma:unifying_lemma}
    Consider $I \in \mathbb{N}$. Running Algorithm~\ref{algo:UCB_mesh} with parameters $\delta \in (0,1/2)$ and $C \in (0,1)$ we have
    \begin{align}
        &\sum_{t=1}^{I} \mathbb{E}\{Q(t)\}   \nonumber\\&\leq \frac{65\times 2^{\frac{2}{p-1}}\gamma I}{(\gamma-1) \varepsilon}+2^{\frac{2}{p-1}}(T_b^{\text{sum}})^2+\frac{2^{\frac{2}{p-1}}\gamma}{(\gamma-1)\varepsilon}\left[\mathbb{E}\{Q^2(T^{\text{sum}}_b+1)\}-\mathbb{E}\{Q^2( I+1)\}\right]_+ \nonumber\\&\ \ \ \  +\frac{2^{\frac{5q}{2}-\delta q+3}C^q\gamma^{2q}(7-2\delta)^q \log^{q+2}(2I)I^{2-\frac{q}{2}-\delta q}}{(\gamma-1)^{2q}\varepsilon^{2q}(1-(q/2))^2} +\frac{2^{2q+3}\gamma^{2q}(7-2\delta)^q \log^{q+2}(2I)I^{2-q}}{(\gamma-1)^{2q}\varepsilon^{2q}(1-(q/2))^2} \nonumber
    \end{align}
    for any $I, p,q,\gamma$ satisfying $I \in \mathbb{N}$, $q \in (1,2)$, $1/p+1/q = 1$, and $\gamma >1$, where $T_l^{\text{sum}}$ is defined in \eqref{eqn:t_l_sum}
\end{lemma}
\subsubsection{Proof of Theorem~\ref{thm:main_thm_bef_opt}-1}\label{sec:main_thm_1}
First, notice that if $H \leq T^{\text{sum}}_b$, then
\begin{align}
    \sum_{t=1}^H Q(t) \leq \sum_{t=1}^H(t-1) \leq H^2 \leq HT^{\text{sum}}_b \leq  2H\left(\frac{\gamma}{\varepsilon C}\right)^{\frac{2}{1-2\delta}} \nonumber
\end{align}
where the first inequality follows from Lemma~\ref{eqn:det_queue_bound} and the last inequality follows from Lemma~\ref{lemma:to_b}.
Hence, Theorem~\ref{thm:main_thm_bef_opt}-1 trivially holds. Therefore, for the rest of this section, let us assume $H \geq T^{\text{sum}}_b+1$. Let us define 
\begin{align}\label{eqn:til_H_def}
    \tilde{H} = \max\{h \in [T^{\text{sum}}_b+1:H]:\mathbb{E}\{Q^2(h)\} \geq \mathbb{E}\{Q^2(T^{\text{sum}}_b+1)\}   \}-1.
\end{align}
Notice that the above definition is valid since we assumed $H \geq T^{\text{sum}}_b+1$ and $T^{\text{sum}}_b+1 \in \{h \in [T^{\text{sum}}_b+1:H]:\mathbb{E}\{Q^2(h)\} \geq \mathbb{E}\{Q^2(T^{\text{sum}}_b+1)\} \}$.  

From the definition of $\tilde{H}$, for all $h \in [\tilde{H}+2,H]$,  we have
\begin{align}\label{eqn:q_h}
    \mathbb{E}\{Q(h)\} \leq \sqrt{\mathbb{E}\{Q^2(h)\}} <_{(a)} \sqrt{\mathbb{E}\{Q^2(T^{\text{sum}}_b+1)\}} \leq T^{\text{sum}}_b
\end{align}
where (a) follows from the definition of $\tilde{H}$ in \eqref{eqn:til_H_def}, and the last inequality follows from Lemma~\ref{eqn:det_queue_bound}. This gives
\begin{align}\label{eqn:mp_eq}
    \sum_{t=\tilde{H}+1}^H \mathbb{E}\{Q(h)\} &=  \mathbb{E}\{Q(\tilde{H}+1)\}+\sum_{t=\tilde{H}+2}^H \mathbb{E}\{Q(h)\}  \leq_{(a)} \tilde{H}+HT^{\text{sum}}_b \leq H(T^{\text{sum}}_b+1) \nonumber\\&\leq 2H\left(\frac{\gamma}{\varepsilon C}\right)^{\frac{2}{1-2\delta}} +H
\end{align}
where (a) follows from Lemma~\ref{eqn:det_queue_bound} and \eqref{eqn:q_h} and the last inequality follows from Lemma~\ref{lemma:to_b}.

Notice that from the definition of $\tilde{H}$, we have $\tilde{H} \in [0:H-1]$. If $\tilde{H} = 0$, from \eqref{eqn:mp_eq}, we trivially have Theorem~\ref{thm:main_thm_bef_opt}-1. Hence, we assume $\tilde{H} \geq 1$. Using $I = \tilde{H}$ in Lemma~\ref{lemma:unifying_lemma}, we have
\begin{align}\label{eqn:added_new}
      &\sum_{t=1}^{\tilde{H}} \mathbb{E}\{Q(t)\} \nonumber\\ &\leq \frac{65\times 2^{\frac{2}{p-1}}\gamma \tilde{H}}{(\gamma-1) \varepsilon}+2^{\frac{2}{p-1}}(T_b^{\text{sum}})^2+\frac{2^{\frac{2}{p-1}}\gamma}{(\gamma-1)\varepsilon}\left[\mathbb{E}\{Q^2(T^{\text{sum}}_b+1)\}-\mathbb{E}\{Q^2( \tilde{H}+1)\}\right]_+ \nonumber\\&\ \   +\frac{2^{\frac{5q}{2}-\delta q+3}C^q\gamma^{2q}(7-2\delta)^q \log^{q+2}(2\tilde{H})\tilde{H}^{2-\frac{q}{2}-\delta q}}{(\gamma-1)^{2q}\varepsilon^{2q}(1-(q/2))^2} +\frac{2^{2q+3}\gamma^{2q}(7-2\delta)^q \log^{q+2}(2\tilde{H})\tilde{H}^{2-q}}{(\gamma-1)^{2q}\varepsilon^{2q}(1-(q/2))^2}\nonumber\\& =_{(a)} \frac{65\times 2^{\frac{2}{p-1}}\gamma \tilde{H}}{(\gamma-1) \varepsilon}+2^{\frac{2}{p-1}}(T_b^{\text{sum}})^2 +\frac{2^{\frac{5q}{2}-\delta q+3}C^q\gamma^{2q}(7-2\delta)^q \log^{q+2}(2\tilde{H})\tilde{H}^{2-\frac{q}{2}-\delta q}}{(\gamma-1)^{2q}\varepsilon^{2q}(1-(q/2))^2} \nonumber\\&\ \ +\frac{2^{2q+3}\gamma^{2q}(7-2\delta)^q \log^{q+2}(2\tilde{H})\tilde{H}^{2-q}}{(\gamma-1)^{2q}\varepsilon^{2q}(1-(q/2))^2}\nonumber\\& \leq_{(b)} \frac{65\times 2^{\frac{2}{p-1}}\gamma H}{(\gamma-1) \varepsilon}+2^{\frac{2}{p-1}}HT_b^{\text{sum}} +\frac{2^{\frac{5q}{2}-\delta q+3}C^q\gamma^{2q}(7-2\delta)^q \log^{q+2}(2H)H^{2-\frac{q}{2}-\delta q}}{(\gamma-1)^{2q}\varepsilon^{2q}(1-(q/2))^2} \nonumber\\&\ \ +\frac{2^{2q+3}\gamma^{2q}(7-2\delta)^q \log^{q+2}(2H)H^{2-q}}{(\gamma-1)^{2q}\varepsilon^{2q}(1-(q/2))^2}
\end{align}
where (a) follows since from the definition of $\tilde{H}$, we have $\mathbb{E}\{Q^2(T^{\text{sum}}_b+1)\}\leq \mathbb{E}\{Q^2( \tilde{H}+1)\}$, and (b) follows since $H \geq \tilde{H} \geq T^{\text{sum}}_b$.
Summing \eqref{eqn:added_new} with \eqref{eqn:mp_eq}, we have
\begin{align}\label{eqn:before_main_thm}
&\sum_{t=1}^{H} \mathbb{E}\{Q(t)\} \nonumber\\&\leq_{(a)} \frac{65\times 2^{\frac{2}{p-1}}\gamma H}{(\gamma-1) \varepsilon}+2^{\frac{2}{p-1}}HT_b^{\text{sum}}+ 2H\left(\frac{\gamma}{\varepsilon C}\right)^{\frac{2}{1-2\delta}}+H \nonumber\\&\ \ \ \ +\frac{2^{\frac{5q}{2}-\delta q+3}C^q\gamma^{2q}(7-2\delta)^q \log^{q+2}(2H)H^{2-\frac{q}{2}-\delta q}}{(\gamma-1)^{2q}\varepsilon^{2q}(1-(q/2))^2} +\frac{2^{2q+3}\gamma^{2q}(7-2\delta)^q \log^{q+2}(2H)H^{2-q}}{(\gamma-1)^{2q}\varepsilon^{2q}(1-(q/2))^2} \nonumber\\&\leq_{(b)} \frac{65\times 2^{\frac{2}{p-1}}\gamma H}{(\gamma-1) \varepsilon}+\frac{\left(2^{\frac{p+1}{p-1}}+2\right)\gamma^{\frac{2}{1-2\delta}}H}{\varepsilon^{\frac{2}{1-2\delta}}C^{\frac{2}{1-2\delta}}}+H +\frac{2^{\frac{5q}{2}-\delta q+3}C^q\gamma^{2q}(7-2\delta)^q \log^{q+2}(2H)H^{2-\frac{q}{2}-\delta q}}{(\gamma-1)^{2q}\varepsilon^{2q}(1-(q/2))^2} \nonumber\\&\ \ \ \ +\frac{2^{2q+3}\gamma^{2q}(7-2\delta)^q \log^{q+2}(2H)H^{2-q}}{(\gamma-1)^{2q}\varepsilon^{2q}(1-(q/2))^2} 
\end{align}
where (a) follows by adding~\eqref{eqn:added_new} with \eqref{eqn:mp_eq}, and (b) follows from Lemma~\ref{lemma:to_b}. Dividing both sides by $H$, we get Theorem~\ref{thm:main_thm_bef_opt}-1.
\subsubsection{Proof of Theorem~\ref{thm:main_thm_bef_opt}-2}\label{sec:main_thm_2}
Using $I = H$ in in Lemma~\ref{lemma:unifying_lemma}, and dividing both sides by $H$ we have
\begin{align}
      &\frac{1}{H}\sum_{t=1}^{H} \mathbb{E}\{Q(t)\}  \nonumber\\&\leq  \frac{65\times 2^{\frac{2}{p-1}}\gamma}{(\gamma-1) \varepsilon}+\frac{2^{\frac{2}{p-1}}(T_b^{\text{sum}})^2}{H}+\frac{2^{\frac{2}{p-1}}\gamma}{(\gamma-1)\varepsilon H}\left[\mathbb{E}\{Q^2(T^{\text{sum}}_b+1)\}-\mathbb{E}\{Q^2( I+1)\}\right]_+\nonumber\\&\ \ \ \  +\frac{2^{\frac{5q}{2}-\delta q+3}C^q\gamma^{2q}(7-2\delta)^q \log^{q+2}(2H)H^{1-\frac{q}{2}-\delta q}}{(\gamma-1)^{2q}\varepsilon^{2q}(1-(q/2))^2} +\frac{2^{2q+3}\gamma^{2q}(7-2\delta)^q \log^{q+2}(2H)H^{1-q}}{(\gamma-1)^{2q}\varepsilon^{2q}(1-(q/2))^2} \nonumber\\& \leq  \frac{65\times 2^{\frac{2}{p-1}}\gamma}{(\gamma-1) \varepsilon}+\frac{2^{\frac{2}{p-1}}(T_b^{\text{sum}})^2}{H}+\frac{2^{\frac{2}{p-1}}\gamma\mathbb{E}\{Q^2(T^{\text{sum}}_b+1)\}}{(\gamma-1)\varepsilon H}\nonumber\\&\ \ \ \  +\frac{2^{\frac{5q}{2}-\delta q+3}C^q\gamma^{2q}(7-2\delta)^q \log^{q+2}(2H)H^{1-\frac{q}{2}-\delta q}}{(\gamma-1)^{2q}\varepsilon^{2q}(1-(q/2))^2} +\frac{2^{2q+3}\gamma^{2q}(7-2\delta)^q \log^{q+2}(2H)H^{1-q}}{(\gamma-1)^{2q}\varepsilon^{2q}(1-(q/2))^2}\nonumber\\& \leq_{(a)}  \frac{65\times 2^{\frac{2}{p-1}}\gamma}{(\gamma-1) \varepsilon}+\frac{2^{\frac{2}{p-1}}(T_b^{\text{sum}})^2}{H}+\frac{2^{\frac{2}{p-1}}\gamma(T^{\text{sum}}_b)^2}{(\gamma-1)\varepsilon H}\nonumber\\&\ \ \ \  +\frac{2^{\frac{5q}{2}-\delta q+3}C^q\gamma^{2q}(7-2\delta)^q \log^{q+2}(2H)H^{1-\frac{q}{2}-\delta q}}{(\gamma-1)^{2q}\varepsilon^{2q}(1-(q/2))^2} +\frac{2^{2q+3}\gamma^{2q}(7-2\delta)^q \log^{q+2}(2H)H^{1-q}}{(\gamma-1)^{2q}\varepsilon^{2q}(1-(q/2))^2}\nonumber\\& \leq_{(b)}  \frac{65\times 2^{\frac{2}{p-1}}\gamma}{(\gamma-1) \varepsilon}+\frac{4}{H}\left(2^{\frac{2}{p-1}}+\frac{2^{\frac{2}{p-1}}\gamma}{(\gamma-1)\varepsilon}\right)\left(\frac{\gamma}{\varepsilon C}\right)^{\frac{4}{1-2\delta}}\nonumber\\&\ \ \ \  +\frac{2^{\frac{5q}{2}-\delta q+3}C^q\gamma^{2q}(7-2\delta)^q \log^{q+2}(2H)H^{1-\frac{q}{2}-\delta q}}{(\gamma-1)^{2q}\varepsilon^{2q}(1-(q/2))^2} +\frac{2^{2q+3}\gamma^{2q}(7-2\delta)^q \log^{q+2}(2H)H^{1-q}}{(\gamma-1)^{2q}\varepsilon^{2q}(1-(q/2))^2}\nonumber
\end{align}
where (a) follows from Lemma~\ref{eqn:det_queue_bound},  and (b) follows from Lemma~\ref{lemma:to_b}. Now, assume $q \in (1,2)$ is chosen such that $q> 1/(1/2+\delta)$ (recall that $\delta \in (0,1/2)$, so this choice is possible). Hence, we have $1-\frac{q}{2}-\delta q < 0$ and $1-q<0$. Hence, as $H \to \infty$, the last three terms of the above bound go to 0. This gives
\begin{align}
      &\lim_{H \to \infty}\frac{1}{H}\sum_{t=1}^{H} \mathbb{E}\{Q(t)\}   \leq  \frac{65\times 2^{2(q-1)}\gamma}{(\gamma-1) \varepsilon}.\nonumber
\end{align}
where we have used  $p = q/(q-1)$ (because $1/p+1/q = 1$).
Since the above holds for all $q \in \left(\frac{1}{1/2+\delta},2\right)$, and $\gamma > 1$, we have Theorem~\ref{thm:main_thm_bef_opt}-2.
\subsubsection{Proof of Lemma~\ref{lemma:unifying_lemma}}\label{sec:main_lemma}
% Consider $l \geq b$. Recall that phase $l$ will have a rate level that when used throughout the phase will make the queues stable (see the description of the unknown $\varepsilon$ case in Section~\ref{sec:sys_mod}). The idea is to learn this rate level in each phase and prove that we can achieve a bounded time-average queue size.
First, notice that if $I \leq T^{\text{sum}}_b$, then
\begin{align}
    \sum_{t = 1} ^I Q(t) \leq \sum_{t = 1} ^I (t-1) \leq I^2 \leq (T_b^{\text{sum}})^2, \nonumber
\end{align} 
where the first inequality follows from Lemma~\ref{eqn:det_queue_bound}.
Hence, Lemma~\ref{lemma:unifying_lemma} holds in this case.  Furthermore, if $I \leq 8$, Lemma~\ref{lemma:unifying_lemma} trivially holds since $Q(t) \leq 7$ for all $t \in [I]$. Hence, we assume that $I \geq \max\{T^{\text{sum}}_b+1,9\}$.

We begin with the following lemma with respect to the specific values $T_l$ and $d_l$ defined in \eqref{eqn:t_l} and \eqref{eqn:K_l}.

\begin{lemma}\label{lemma:init_lemma}
   Consider $I \geq 9$. For the $T_l$ and $d_l$ defined in \eqref{eqn:t_l} and \eqref{eqn:K_l}, we have the following
    \begin{enumerate}
        \item $T_{a(I)} \leq 2I$
        \item $a(I) \leq \log_2(I)$
        \item $\sum_{n=1}^{a(I)} T^{\text{sum}}_n \leq 4I$
    \end{enumerate}
    where $a(t)$ defined in \eqref{eqn:a_t_def} is the phase to which time slot $t$ belongs, $T_l^{\text{sum}}$ is defined in \eqref{eqn:t_l_sum}.
    \begin{proof}
      See Appendix~\ref{app:init_lemma}
    \end{proof}
\end{lemma}

Fix a phase $l \geq b$. For each $u \in [0:T_l-1]$, define the \emph{good event} $\mathcal{G}_{l}(u)$ as 
\begin{align}\label{eqn:good_event}
    \mathcal{G}_{l}(u) = \left\{ \mu_{l,k} \in \left[ \text{UCB}_{l,k}(u)- 2\sqrt{\frac{(7-2\delta)\log(T_l)}{4(1 \vee N_{l,k}(u))}},\text{UCB}_{l,k}(u)\right] \forall k \in \mathcal{K}_l\right\}
\end{align}
where $\text{UCB}_{l,k}(u)$ is defined in \eqref{eq:ucb}, and $\mu_{l,k}$ is defined in \eqref{eqn:mu_lk}.

We have the following lemma.
\begin{lemma}\label{lemma:good_event_lemma_app}
   Recall the definition of $b$ in \eqref{eqn:b_def}. Consider a phase $l \geq b$. For each $u \in [0:T_l - 1]$, we have that the event $\mathcal{G}_l(u)$ is independent of the history before phase $l$, and $\mathbb{P}\{\mathcal{G}_l(u)^c\} \leq \frac{4}{T_l}$.
\begin{proof}
See Appendix~\ref{app:good_event}
\end{proof}
\end{lemma}
From the queueing equation~\eqref{eqn:queing_equation}, we have for any $u \in [1:T_l]$. 
\begin{align}
    &Q(T^{\text{sum}}_l+u+1)^2 \nonumber\\&\leq \left[Q(T^{\text{sum}}_l+u) + A(T^{\text{sum}}_l+u) - S_{l,K_l(u)}(u)\right]^2 \nonumber\\& \leq Q(T^{\text{sum}}_l+u)^2 + [A(T^{\text{sum}}_l+u) - S_{l,K_l(u)}(u)]^2   + 2Q(T^{\text{sum}}_l+u)[A(T^{\text{sum}}_l+u)- S_{l,K_l(t)}(u)] \nonumber \\& \leq Q(T^{\text{sum}}_l+u)^2+ 1 + 2Q(T^{\text{sum}}_l+u)[A(T^{\text{sum}}_l+u) -S_{l,K_l(t)}(u)] \nonumber
\end{align}
where the last inequality follows since $A(T^{\text{sum}}_l+u)- S_{l,K_l(u)}(u) \in [-1,1]$. Define $\Delta_l(u) = \frac{1}{2}\mathbb{E}\{Q(T^{\text{sum}}_l+u+1)^2\} - \frac{1}{2}\mathbb{E}\{Q(T^{\text{sum}}_l+u)^2\}$. Taking the expectations of the above, we have
\begin{align}\label{eqn:1919191}
    \Delta_l(u) &\leq 
    \frac{1}{2}+ \mathbb{E}\{Q(T^{\text{sum}}_l+u)[\lambda - \mu_{l,K_l(u)}]\} \nonumber
    \\& =
    \frac{1}{2}+ \underbrace{\mathbb{E}\{Q(T^{\text{sum}}_l+u)[\lambda - \mu_{l,K_l(u)}]|\mathcal{G}_l(u-1)\}\mathbb{P}\{\mathcal{G}_{l}(u-1)\}}_{\text{Term 1}}
    \nonumber\\&\ \ \ \ +\underbrace{\mathbb{E}\{Q(T^{\text{sum}}_l+u)[\lambda - \mu_{l,K_l(u)}]|\mathcal{G}_l^c(u-1)\}\mathbb{P}\{\mathcal{G}_l^c(u-1)\}}_{\text{Term 2}}
\end{align}
Now, in the following two lemmas, we analyze term 1 and term 2 of the above inequality separately.
\begin{lemma}[Term 2 of \eqref{eqn:1919191}]\label{lemma:term_2}
    For any $l \geq b$ and $u \in [1:T_l]$, we have that
    \begin{align}
         \mathbb{E}\{Q(T^{\text{sum}}_l+u)[\lambda - \mu_{l,K_l(u)}]|\mathcal{G}^c(u-1)\}\mathbb{P}\{\mathcal{G}_l^c(u-1)\}  \leq  \frac{4T^{\text{sum}}_{l+1}}{T_l} \nonumber
    \end{align}
    \begin{proof}
        Notice that
        \begin{align}
            &\mathbb{E}\{Q(T^{\text{sum}}_l+u)[\lambda - \mu_{l,K_l(u)}]|\mathcal{G}^c(u-1)\}\mathbb{P}\{\mathcal{G}_l^c(u-1)\} \nonumber\\& \leq_{(a)} \mathbb{E}\{Q(T^{\text{sum}}_l+u)|\mathcal{G}^c(u-1)\}\mathbb{P}\{\mathcal{G}_l^c(u-1)\} \leq_{(b)} \frac{4(T^{\text{sum}}_l+u)}{T_l}\leq \frac{4T^{\text{sum}}_{l+1}}{T_l} \nonumber
        \end{align}
        where for (a) we have used $\lambda \leq 1$, for (b) we have used Lemma~\ref{eqn:det_queue_bound} and Lemma~\ref{lemma:good_event_lemma_app}, and the last inequality follows since $u \in [1:T_l]$.
    \end{proof}
\end{lemma}
\begin{lemma}[Term 1 of \eqref{eqn:1919191}]\label{lemma:temp_1}
   For any $l \geq b$ and $u \in [1:T_l]$, we have that
    \begin{align}
       & \mathbb{E}\{Q(T^{\text{sum}}_l+u)[\lambda - \mu_{l,K_l(u)}]|\mathcal{G}_l(u-1)\}\mathbb{P}\{\mathcal{G}_{l}(u-1)\} \nonumber\\& \leq    -\frac{(\gamma-1)\varepsilon}{\gamma}\mathbb{E}\{Q(u+T^{\text{sum}}_l)\}+\frac{4T^{\text{sum}}_{l+1}}{ T_l}+  \sqrt{7-2\delta}\mathbb{E}\left\{Q(T^{\text{sum}}_l+u)\sqrt{\frac{\log\left(T_l\right)}{(1 \vee N_{l,K_l(u)}(u-1))}}\right\} \nonumber
    \end{align}
    \begin{proof}
The main idea behind the proof is to use the definition of the good event $\mathcal{G}_l(u-1)$ in \eqref{eqn:good_event} to bound $\mu_{l,K_l(u)}$, and then use Corollary~\ref{corr:main_corr}. We defer the full proof to Appendix~\ref{app:temp_1}.
    \end{proof}
\end{lemma}
Using the above two lemmas in \eqref{eqn:1919191}, we have
\begin{align}\label{eqn:basic_eqn}
     &\frac{1}{2}\mathbb{E}\{Q^2(u+T^{\text{sum}}_l+1)\} - \frac{1}{2}\mathbb{E}\{Q^2(u+T^{\text{sum}}_l)\} 
    \\&\leq 
    \frac{1}{2} -\frac{(\gamma -1)\varepsilon}{\gamma}\mathbb{E}\{Q(u+T^{\text{sum}}_l)\} +  \sqrt{7-2\delta}\mathbb{E}\left\{Q(T^{\text{sum}}_l+u)\sqrt{\frac{\log\left(T_l\right)}{(1 \vee N_{l,K_l(u)}(u-1))}}\right\} +\frac{8T^{\text{sum}}_{l+1}}{T_l} \nonumber
\end{align}
For each $l \in [b:a(I)]$, we define
\begin{align}\label{eqn:til_t}
    \tilde{T}_l = \begin{cases}
        T_l & \text{ if } l \in [b:a(I)-1]\\
        I-T^{\text{sum}}_{a(I)} & l = a(I).
    \end{cases}
\end{align}
Hence, $\tilde{T}_l $ denotes the number of time slots belonging to phase $l$ within the first $I$ time slots.  

The following lemma is a consequence of summing \eqref{eqn:basic_eqn} over time slots and performing simple algerbraic manipulations. We defer the proof to the appendix.
\begin{lemma}\label{lemma:proof_1_to_til_H}
We have that
    \begin{align}
        &\sum_{t=1}^{I}\mathbb{E}\{Q(t)\} \nonumber\\&\leq \frac{\gamma I}{2(\gamma-1)\varepsilon}+ \frac{(T_b^{\text{sum}})^2}{2}  +\frac{\gamma\sqrt{7-2\delta}}{(\gamma-1)\varepsilon} \mathbb{E}\left\{\sum_{l=b}^{a( I)}\sum_{u=1}^{\tilde{T}_l} Q(T^{\text{sum}}_l+u)\sqrt{\frac{\log\left(T_l\right)}{(1 \vee N_{l,K_l(u)}(u-1))}}\right\} \nonumber\\&\ \ \ \  + \frac{8\gamma}{(\gamma-1)\varepsilon}\sum_{l=b}^{a( I)}T^{\text{sum}}_{l+1}  + \frac{\gamma}{2(\gamma-1)\varepsilon}\left[\mathbb{E}\{Q^2(T^{\text{sum}}_b+1)\}-\mathbb{E}\{Q^2( I+1)\}\right]_+ \nonumber
    \end{align}
\begin{proof}
See~Appendix~\ref{app:proof_1_to_til_H}
\end{proof}
\end{lemma}
To get the bound of Lemma~\ref{lemma:unifying_lemma} from Lemma~\ref{lemma:proof_1_to_til_H}, we require bounding the term
\begin{align}
    \frac{\gamma\sqrt{7-2\delta}}{(\gamma-1)\varepsilon} \mathbb{E}\left\{\sum_{l=b}^{a( I)}\sum_{u=1}^{\tilde{T}_l} Q(T^{\text{sum}}_l+u)\sqrt{\frac{\log\left(T_l\right)}{(1 \vee N_{l,K_l(u)}(u-1))}}\right\}.\nonumber
\end{align}
We begin this process with two lemmas. The following lemma is adapted from~\cite{Huang2024}.
 \begin{lemma}\label{lemma:q_pow_bnd}
Consider nonnegative real numbers $x_1,x_2,\dots,x_n$ such that $x_1 = 0$, $|x_i - x_{i+1}| \leq 1$ for all $i \in [1:n-1]$. Let $S = \sum_{t=1}^n x_t, \text{ and } D^p = \sum_{t=1}^n x_t^p$
for $p \geq 2$. We have $D \leq 2^{\frac{p-1}{2p}}S^{\frac{p+1}{2p}}$.
\begin{proof}
See Appendix~\ref{app:q_pow_bnd}
\end{proof}
\end{lemma}
\begin{lemma}\label{lemma:bounding_ucb_slack_holder}
For each $l \in [b:a( I)]$, $q \in (1,2)$, and $\tilde{T}_l$ defined in \eqref{eqn:til_t}, we have that
\begin{align}         
\sum_{u=1}^{\tilde{T}_l} \left(\frac{\log\left(T_l\right)}{(1 \vee N_{l,K_l(u)} (u-1))}\right)^{\frac{q}{2}} \leq  \log^{q/2}(T_l)\frac{d_l^{q/2}T_l^{1-(q/2)}}{1-(q/2)} \nonumber
\end{align}
\begin{proof}
See Appendix~\ref{app:bounding_ucb_slack_holder}
\end{proof}
\end{lemma}
Fix $p,q$ such that $q \in (1,2)$ satisfying $1/p + 1/q = 1$. From the H{\"o}lder inequality, we have that
\begin{align}\label{eqn:to_expect}
    &\sum_{l=b}^{a( I)}\sum_{u=1}^{\tilde{T}_l} Q(T^{\text{sum}}_l+u)\sqrt{\frac{\log\left(T_l\right)}{(1 \vee N_{l,K_l(u)}(u-1))}} \nonumber\\&\leq_{(a)} \left(\sum_{l=b}^{a( I)}\sum_{u=1}^{\tilde{T}_l} \left(\frac{\log\left(T_l\right)}{(1 \vee N_{l,K_l(u)}(u-1))}\right)^{\frac{q}{2}}\right)^{\frac{1}{q}}\left(\sum_{l=b}^{a( I)}\sum_{u=1}^{\tilde{T}_l} Q^p(T^{\text{sum}}_l+u)\right)^{\frac{1}{p}} \nonumber\\& = \left(\sum_{l=b}^{a( I)}\sum_{u=1}^{\tilde{T}_l} \left(\frac{\log\left(T_l\right)}{(1 \vee N_{l,K_l(u)}(u-1))}\right)^{\frac{q}{2}}\right)^{\frac{1}{q}}\left(\sum_{t = T^{\text{sum}}_b+1}^{I} Q^p(t)\right)^{\frac{1}{p}}\nonumber\\&\leq_{(b)} \left(\sum_{l=b}^{a( I)}\left[ \log^{q/2}(T_l)\frac{d_l^{q/2}T_l^{1-(q/2)}}{1-(q/2)} \right] \right)^{\frac{1}{q}}\left(\sum_{t = 1}^{I} Q^p(t)\right)^{\frac{1}{p}} \nonumber\\& \leq_{(c)} \left(a( I) \log^{q/2}(T_{a( I)})\frac{d_{a( I)}^{q/2}T_{a(I)}^{1-(q/2)}}{1-(q/2)}\right)^{\frac{1}{q}} \left(\sum_{t = 1}^{I} Q^p(t)\right)^{\frac{1}{p}}  \nonumber\\&\leq_{(d)} 2^{\frac{p-1}{2p}}\left(a( I) \log^{q/2}(T_{a( I)})\frac{d_{a( I)}^{q/2}T_{a(I)}^{1-(q/2)}}{1-(q/2)}\right)^{\frac{1}{q}} \left(\sum_{t=1}^{ I}Q(t)\right)^{\frac{p+1}{2p}} 
\end{align}
where (a) follows from the H{\"o}lder inequality, (b) follows from Lemma~\ref{lemma:bounding_ucb_slack_holder}, (c) follows since the sequences $T_1,T_2,\dots,$ and $d_1,d_2,\dots,$ are nondecreasing, and (d) follows by applying Lemma~\ref{lemma:q_pow_bnd} to the sequence $Q(1), Q(2),\dots, Q(I)$. Taking expectations of \eqref{eqn:to_expect} and using the Jensen's inequality, we have
\begin{align}
    &\mathbb{E}\left\{\sum_{l=b}^{a( I)}\sum_{u=1}^{\tilde{T}_l} Q(T^{\text{sum}}_l+u)\sqrt{\frac{\log\left(T_l\right)}{(1 \vee N_{l,K_l(u)}(u-1))}} \right\}  \nonumber\\&\leq 2^{\frac{p-1}{2p}}\left(a( I) \log^{q/2}(T_{a( I)})\frac{d_{a( I)}^{q/2}T_l^{1-(q/2)}}{1-(q/2)}\right)^{\frac{1}{q}} \left(\sum_{t=1}^{ I}\mathbb{E}\{Q(t)\}\right)^{\frac{p+1}{2p}}  \nonumber
\end{align}
where the last inequality follows from Jensen's inequality, since $\frac{p+1}{2p} <1 $ (recall that $p>1$).
Combining the above with Lemma~\ref{lemma:proof_1_to_til_H}, we have
\begin{align}\label{eqn:bef_bounding}
\sum_{t=1}^{I}\mathbb{E}\{Q(t)\} &\leq \frac{\gamma I}{2(\gamma-1)\varepsilon}+ \frac{(T_b^{\text{sum}})^2}{2}+ \frac{8\gamma}{(\gamma-1)\varepsilon}\sum_{l=b}^{a( I)}T^{\text{sum}}_{l+1} \nonumber\\&\ \ \ \ +\frac{2^{\frac{p-1}{2p}}\gamma\sqrt{(7-2\delta)}\left(a( I) \log^{q/2}(T_{a( I)})\frac{d_{a( I)}^{q/2}T_l^{1-(q/2)}}{1-(q/2)}\right)^{\frac{1}{q}}}{(\gamma-1)\varepsilon}  \left(\sum_{t=1}^{ I}\mathbb{E}\{Q(t)\}\right)^{\frac{p+1}{2p}} \nonumber\\&\ \ \ \  + \frac{\gamma}{2(\gamma-1)\varepsilon}\left[\mathbb{E}\{Q^2(T^{\text{sum}}_b+1)\}-\mathbb{E}\{Q^2( I+1)\}\right]_+
\end{align}
Next, we have the following lemma adapted from~\cite{Huang2024}.
\begin{lemma}\label{lemma:abX_lemma}
Consider nonnegative real numbers $a,b,X$ and $d \geq 1$. such that $X^d \leq a+ bX^{d-1}$. We have that $X^d \leq \left(a^{\frac{1}{d}} + b\right)^{d} \leq 2^{d-1}a+ 2^{d-1}b^d$.
\begin{proof}
    See Appendix~\ref{app:abX_lemma}.
\end{proof}
\end{lemma}
Using Lemma~\ref{lemma:abX_lemma} in \eqref{eqn:bef_bounding} with $X = \left(\sum_{t=1}^{I} \mathbb{E}\{Q(t)\}\right)^{(p-1)/2p}$ and $d=2p/(p-1)$, we have
\begin{align}\label{eqn:last_holder}
   &\sum_{t=1}^{I} \mathbb{E}\{Q(t)\}  \nonumber\\&\leq 2^{\frac{p+1}{p-1}}\left(\frac{\gamma I}{2(\gamma-1) \varepsilon}+\frac{(T_b^{\text{sum}})^2}{2}+\frac{8\gamma}{(\gamma-1)\varepsilon}\sum_{l=b}^{a(I)}T^{\text{sum}}_{l+1}  +\frac{\gamma\left[\mathbb{E}\{Q^2(T^{\text{sum}}_b+1)\}-\mathbb{E}\{Q^2( I+1)\}\right]_+}{2(\gamma-1)\varepsilon}\right) \nonumber\\&\ \ \ \  +  2^{\frac{p+1}{p-1}}\left[\frac{2^{\frac{p-1}{2p}}\gamma\sqrt{7-2\delta}\left[a(I)\log^{q/2}(T_{a(I)})\left(\frac{d_{a(I)}^{q/2}T_{a(I)}^{1-(q/2)}}{1-(q/2)}\right) \right]^{\frac{1}{q}}}{(\gamma-1)\varepsilon}\right]^{\frac{2p}{p-1}}\nonumber\\&\leq_{(a)} 2^{\frac{p+1}{p-1}}\left(\frac{\gamma I}{2(\gamma-1) \varepsilon}+\frac{(T_b^{\text{sum}})^2}{2}+\frac{8\gamma}{(\gamma-1)\varepsilon}\sum_{l=b}^{a(I)}T^{\text{sum}}_{l+1} +\frac{\gamma\left[\mathbb{E}\{Q^2(T^{\text{sum}}_b+1)\}-\mathbb{E}\{Q^2( I+1)\}\right]_+}{2(\gamma-1)\varepsilon}\right) \nonumber\\&\ \ \ \  +\frac{2^{2q}\gamma^{2q}(7-2\delta)^qa^2(I) \log^{q}(T_{a(I)})d^q_{a(I)}T_{a(I)}^{2-q}}{(\gamma-1)^{2q}\varepsilon^{2q}(1-(q/2))^2} \nonumber\\&\leq_{(b)} 2^{\frac{p+1}{p-1}}\left(\frac{\gamma I}{2(\gamma-1) \varepsilon}+\frac{(T_b^{\text{sum}})^2}{2}+\frac{32\gamma I}{(\gamma-1)\varepsilon} +\frac{\gamma\left[\mathbb{E}\{Q^2(T^{\text{sum}}_b+1)\}-\mathbb{E}\{Q^2( I+1)\}\right]_+}{2(\gamma-1)\varepsilon}\right) \nonumber\\&\ \ \ \  +\frac{2^{2q}\gamma^{2q}(7-2\delta)^qa^2(I) \log^{q}(2I)d^q_{a(I)}(2I)^{2-q}}{(\gamma-1)^{2q}\varepsilon^{2q}(1-(q/2))^2} \nonumber\\&\leq_{(c)} \frac{65\times 2^{\frac{2}{p-1}}\gamma I}{(\gamma-1) \varepsilon}+2^{\frac{2}{p-1}}(T_b^{\text{sum}})^2+\frac{2^{\frac{2}{p-1}}\gamma\left[\mathbb{E}\{Q^2(T^{\text{sum}}_b+1)\}-\mathbb{E}\{Q^2( I+1)\}\right]_+}{(\gamma-1)\varepsilon}\nonumber\\&\ \ \ \  +\frac{2^{q+4}\gamma^{2q}(7-2\delta)^q \log^{q+2}(2I)d^q_{a(I)}I^{2-q}}{(\gamma-1)^{2q}\varepsilon^{2q}(1-(q/2))^2} 
\end{align}
where for (a) we have used $1/p + 1/q = 1$ which gives $p/(p-1) = q$, (b) follows from Lemma~\ref{lemma:init_lemma}-1,3, and (c) follows from Lemma~\ref{lemma:init_lemma}-2 since $a(I) \leq \log_2(I) \leq 2\log(I) \leq 2\log(2I)$.

Now, notice that due to the definition of $d_l$ in \eqref{eqn:K_l}, we have
\begin{align}\label{eqn:da_h__}
    d_{a(I)} \leq  CT_{a(I)}^{\left(\frac{1}{2}-\delta\right)}+1 \leq 2^{\left(\frac{1}{2}-\delta\right)} CI^{\left(\frac{1}{2}-\delta\right)}+1
\end{align}
where the last inequality follows from Lemma~\ref{lemma:init_lemma}-1.
Hence,
\begin{align}\label{eqn:daq}
    d^q_{a(I)} \leq \left(2^{\frac{1}{2}-\delta
    }CI^{\frac{1}{2}-\delta
    }+1\right)^q  \leq 2^{\frac{3q}{2}-\delta q-1}C^qI^{\frac{q}{2}-\delta q}+2^{q-1} 
\end{align}
where the first inequality follows from~\eqref{eqn:da_h__}, and the second inequality follows from $(a+b)^x \leq 2^{x-1}(a^x+b^x)$ for nonnegative real numbers $a,b$ and $x \geq 1$ (recall that $q>1$). Using \eqref{eqn:daq} on \eqref{eqn:last_holder}, we get Lemma~\ref{lemma:unifying_lemma}.
\section{Converse Result}\label{sec:converse}
In this section we focus on proving our converse result. In particular, given $0 < \varepsilon \leq 1/144$,  we construct a finite set of environments satisfying assumptions \textbf{A1, A2}, and prove that there exists $T \in \mathbb{N}$ such that $\frac{1}{T}\sum_{t=1}^T \mathbb{E}\{Q(t)\} \geq \frac{6 \times 10^{-7}}{\varepsilon^2}$ in at least one of the environments.  Before defining the environments, we do some useful constructions.
\subsection{Preliminary Constructions}\label{sec:prel_cons}
Fix $\varepsilon$ such that $0 < \varepsilon \leq 1/144$. Define the sequence of real numbers $x_1,x_2,\dots$ such that
\begin{align}\label{eqn:x_def}
    &x_1 = \frac{7}{12}, \text{ and }x_{k+1} = x_k\left(1+ \frac{2\varepsilon}{\frac{1}{2}-\varepsilon}\right) \text{ for } k \geq 1.
\end{align}
Notice that the above is a strictly increasing sequence. Define the sequence of intervals $\mathcal{I}_1,\mathcal{I}_2,\dots$ by
\begin{align}\label{eqn:int_x_k}
    \mathcal{I}_k = (x_k,x_{k+1}]
\end{align}
We have the following claim

\noindent
\textbf{Claim 1:} For each $k \in \mathbb{N}$, we have $|\mathcal{I}_k| > 2\varepsilon$. So $x_k \rightarrow \infty$.
\begin{proof}
    Notice that,
       \begin{align}
           |\mathcal{I}_k| = \frac{2x_k\varepsilon}{\frac{1}{2}-\varepsilon} > \frac{\varepsilon}{\frac{1}{2}-\varepsilon} > 2\varepsilon, \nonumber
       \end{align}
       where the first inequality follows since $x_k \geq x_1 = 7/12 > 1/2$.
\end{proof}
Define
\begin{align}\label{eqn:K_def}
    K = \min\{k:x_{k+1} \geq 2/3\} 
\end{align}
Notice that such a $K$ exists due to claim 1. Hence, we have $2/3 \in \mathcal{I}_K$. Next, we have the following claim.

\noindent
\textbf{Claim 2:} For each $k \in [1:K]$, we have $|\mathcal{I}_k| < 3\varepsilon $.
\begin{proof}
Notice that
\begin{align}
   |\mathcal{I}_k| = \frac{2x_k\varepsilon}{\frac{1}{2}-\varepsilon} \leq \frac{4}{3}\left(\frac{\varepsilon}{\frac{1}{2}-\varepsilon}\right) < \frac{4}{3}\left(\frac{\varepsilon}{\frac{4}{9}}\right) = 3\varepsilon \nonumber
\end{align}
where the first inequality follows since $x_k \leq 2/3$ for all $k \leq K$, and the last inequality follows since $\varepsilon \leq 1/144 <  1/18$. 
\end{proof}
We now state the following lemma, which follows from Claims~1 and~2 (see Appendix~\ref{app:ref_def_lemma} for the proof).
\begin{lemma}\label{lemma:ref_def_lemma}
    We have the following.
    \begin{enumerate}
        \item For each $k \in [1:K]$, we have $[x_k,x_{k+1}] \subset [7/12,1)$.
        \item $K \geq 1/(36\varepsilon)$ and $K \geq 5$.  
    \end{enumerate}
    \begin{proof}
       See Appendix~\ref{app:ref_def_lemma}
    \end{proof}
\end{lemma}
\subsection{Environment Construction}
Now, we are ready to define the environments. In particular, we construct $K$ environments satisfying Assumptions~\textbf{A1, A2}, where $K$ is defined in \eqref{eqn:K_def}.  In all the environments, arrivals $A(t)$ are independent samples of a $\text{Bernoulli}(1/2)$ distribution (hence $\lambda = 1/2$). In the $k$-th  environment ($k \in [1:K]$), the capacities $C(t)$ are i.i.d. samples of a random variable $X_k$ with a CDF $F_{X_k}$, where
\begin{align}
    F_{X_k}(x) = \begin{cases}
        0 & \text{ if } x < \frac{1}{2} - \varepsilon\\
        1 - \frac{\frac{1}{2}-\varepsilon}{x} & \text{ if }x \in \left[\frac{1}{2} - \varepsilon ,x_k\right)\cup [x_{k+1},1)\\
         1 - \frac{\frac{1}{2}-\varepsilon}{x_k}& \text{ if }x \in [x_k,x_{k+1})\\
         1 & \text{ if } x \geq 1 \nonumber.
    \end{cases}
\end{align}
Notice that the definition of $F_{X_k}$ above is valid due to Lemma~\ref{lemma:ref_def_lemma}-1 and $x_k < x_{k+1}$. Additionally, we observe that $F_{X_1},F_{X_2},\dots,F_{X_K}$ are nonnegative, nondecreasing, right continuous functions satisfying $F_{X_k}(x) = 0$ for $x \leq 0$ and $F_{X_k}(x) = 1$ for $x \geq 1$, and hence are valid CDFs of random variables supported in $[0,1]$. Let us define the functions $g_k:[0,1] \to [0,1]$ for each $k \in [1:K]$, where $g_k(x) = x\mathbb{P}\{X_k \geq x \}$. A simple calculation shows
\begin{align}\label{eqn:g_k_def}
    g_k(x) = \begin{cases}
        x & \text{ for } x \in \left[0,\frac{1}{2}-\varepsilon\right]\\
        \frac{1}{2} - \varepsilon & \text{ for } x \in \left(\frac{1}{2} - \varepsilon , 1\right]\setminus \mathcal{I}_k\\
         \frac{x\left(\frac{1}{2}- \varepsilon\right)}{x_k} & \text{ for } x \in \mathcal{I}_k.
    \end{cases}
\end{align}
We have the following claim on $g_k$, which follows directly by the definition of $g_k$ in \eqref{eqn:g_k_def}.

\noindent
\textbf{Claim 3}: For each $k \in [1:K]$, we have that the function $g_k$ defined in  \eqref{eqn:g_k_def} is maximized in $[0,1]$ at $x_{k+1}$ and $g_k(x_{k+1}) = \frac{1}{2}+\varepsilon$.

Since $\lambda = 1/2$, Claim 3 ensures for each $k \in [1:K]$ that $\max_{x \in [0,1]} g_k(x) -\lambda = \varepsilon$. Hence, the functions $g_1,g_2,\dots,g_K$ satisfy the conditions of Assumption~\textbf{A2}. 
Figure~\ref{fig_sim_1} denotes the plots of the above CDFs, and the functions $g_k$.
\begin{figure}
\centering
{\includegraphics[width=0.4\linewidth]{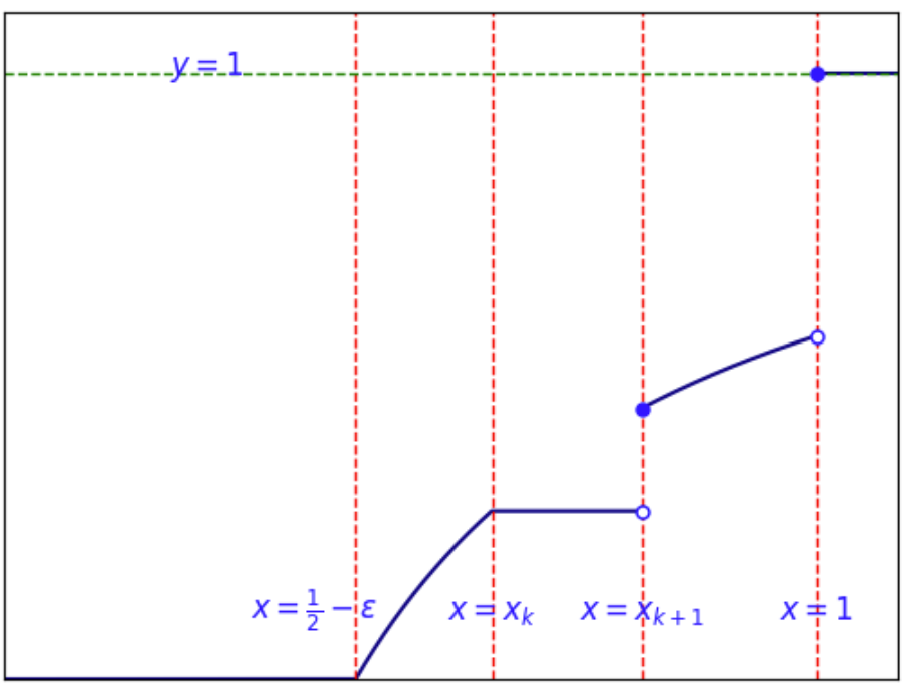}%
\label{Fig:42}}
\hfil
\centering
{\includegraphics[width=0.4\linewidth]{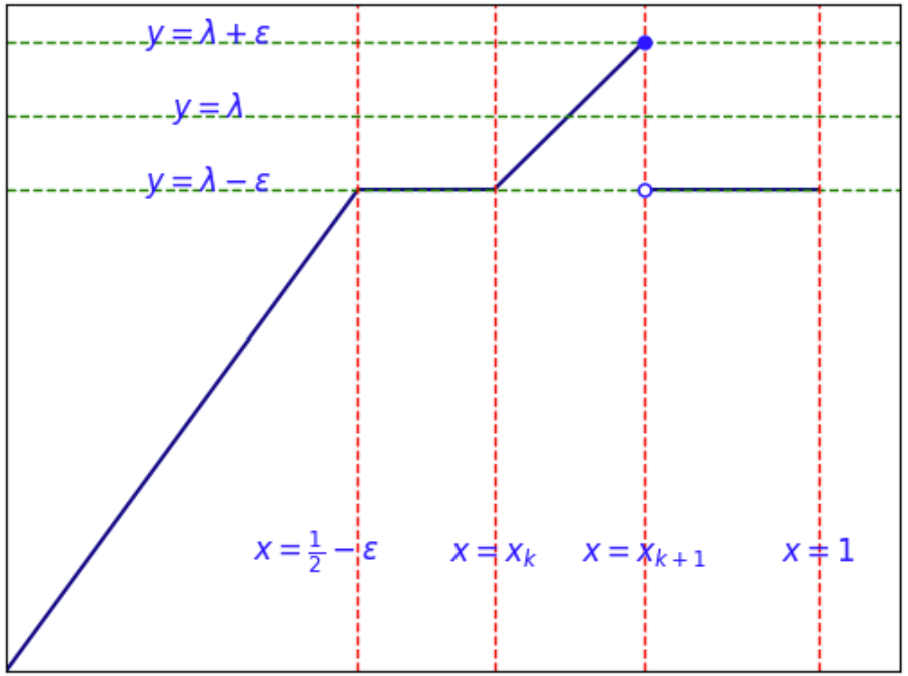}%
\label{Fig:47}}
\caption{Plot of the CDF, $F_{X_k}$ for some $k \in [1:K]$, and the corresponding function $g_k$ \textbf{Left:} Plot of $F_{X_k}$. \textbf{Right:} Plot of $g_k$. }
\label{fig_sim_1}
\end{figure}
\subsection{Converse Bound}
Now, we are ready to introduce the lemma that establishes the converse bound. The proof of the lemma has a similar structure to the proof of the converse result in~\cite{Freund2023}. 
\begin{theorem}\label{thm:lower bound}
    Consider any $0 < \varepsilon \leq 1/144$. Given an algorithm to choose the rates $V(t)$, there exists an environment $k^{'} \in [1:K]$ and $T \in \mathbb{N}$ such that in the Environment $k^{'}$, we have 
    \begin{align}
        \frac{1}{T}\sum_{t=1}^T \mathbb{E}\{Q(t)\} \geq \frac{6 \times 10^{-7}}{\varepsilon^2}. \nonumber
    \end{align}
    \begin{proof}
        See Appendix~\ref{app:lower_bound}.
    \end{proof}
\end{theorem}
The main idea of the proof is to construct an Environment 0 with Bernoulli$(1/2)$ arrivals, where the channel capacities $C(t)$ are i.i.d. samples of a random variable $X_0$ whose CDF $F_{X_0}$ satisfies $\max_{ x \in [0,1]} x\mathbb{P}\{X_0 \geq x \} = 1/2-\varepsilon$. In this environment, it is impossible to stabilize the queue. It is possible to construct $X_0$ such that, for any $k \in [1:K]$, the CDF functions $F_{X_0}$ and $F_{X_k}$ are the same outside of the small interval $\mathcal{I}_k$ defined in \eqref{eqn:int_x_k}. This ensures that the Kullback–Leibler divergence between the distributions of $X_0$ and $X_k$ remains small, implying that any algorithm is expected to behave similarly in the two environments. Consequently, the queue backlogs in the Environment~$k$ can be lower bounded using the fact that the queue cannot be stabilized in Environment~0.

\subsection{Empirical Behavior of Algorithm~\ref{algo:UCB_mesh} in Environment 1}
To illustrate the qualitative behavior of our algorithm (Algorithm~\ref{algo:UCB_mesh}), Figure~\ref{Fig:22} shows the simulated time-average queue size in Environment~1 for $\varepsilon = 1/144$. The algorithm exhibited similar performance across all five environments ($K = 5$, when $\varepsilon = 1/144$). Hence, we only plot results for Environment~1. Understanding the algorithm’s behavior in these environments is important, since they correspond to the worst-case instances that form the basis of the converse result.
\begin{figure}
\centering
\includegraphics[width=0.7\linewidth]{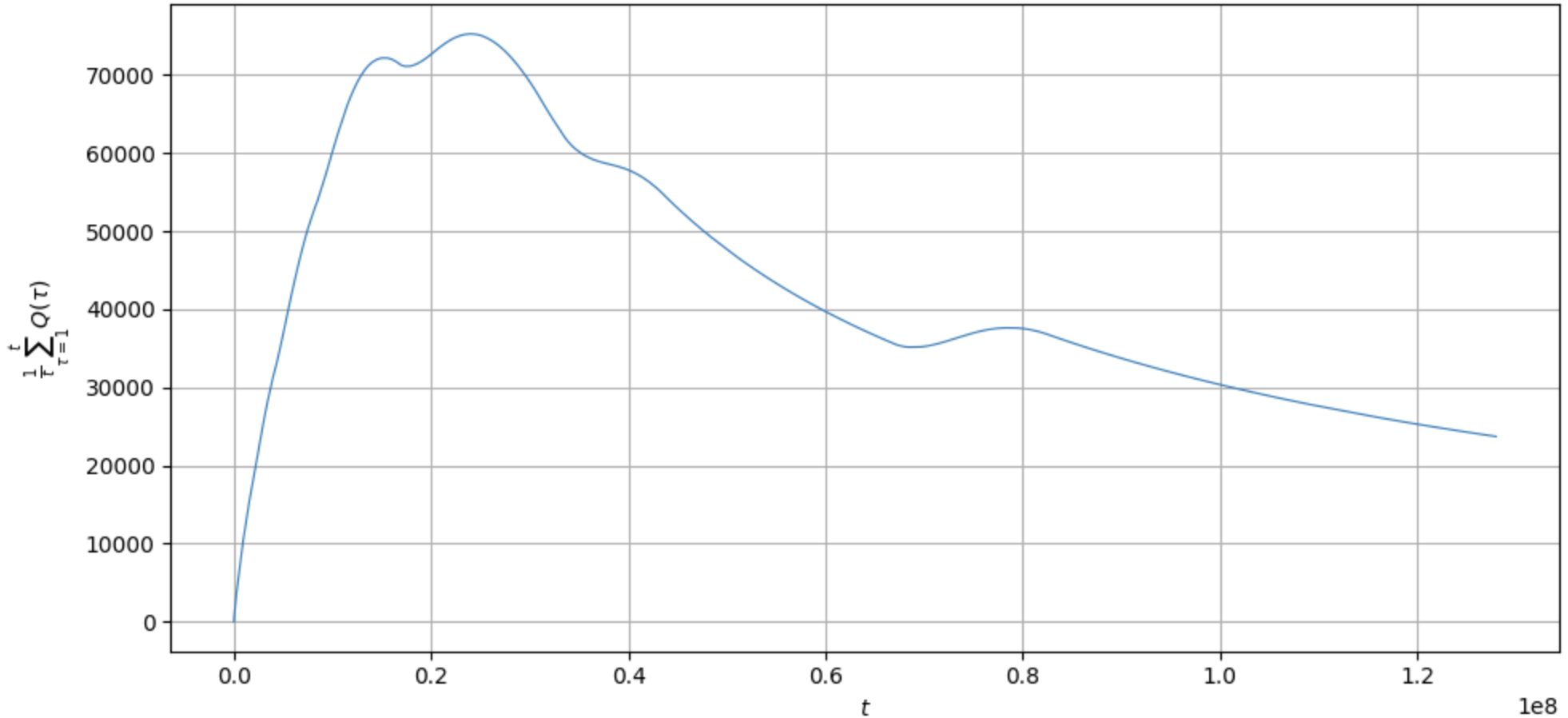}%
\caption{Plot of $\frac{1}{t}\sum_{\tau=1}^t Q(\tau)$ vs. $t$ for Environment 1 when $\varepsilon = 1/144$}\label{Fig:22}
\end{figure}
\section{Known $\varepsilon$}\label{sec:known}

We begin with a simpler discrete model that we then use to address the main (continuous) setup. Consider a queueing system whose service is controlled by a \emph{discrete} multi-armed bandit with arm set \(\mathcal{K}\) (let \({d} \triangleq |\mathcal{K}|\)). In each slot \(t\) and for each arm \(k\in\mathcal{K}\), a service rate \(S_k(t)\in[0,1]\) is realized. 
The vector of service rates at time $t$, denoted by $\{S_k(t)\}_{k\in\mathcal{K}}$, is i.i.d.\ over time with an unknown distribution.
% For any fixed arm \(k\), the sequence \(\{S_k(t)\}_{t\ge1}\) is i.i.d.\ across \(t\); however, for a given \(t\), the vector \(\{S_k(t)\}_{k\in\mathcal{K}}\) may be arbitrarily correlated across arms.
Let \(\mu_k \triangleq \EE[S_k(1)]\) denote the (unknown) mean service rate of arm \(k\). When the controller selects arm \(K(t)\in\mathcal{K}\), the queue evolves as
\[
Q(t+1) \;=\; \bigl[\,Q(t) + A(t) - S_{K(t)}(t)\,\bigr]_+ .
\]
Let \(k^*\in\mathcal{K}\) be an arm with maximal mean, \(\mu^*\triangleq \mu_{k^*}\), and define \(\vedisc \triangleq \mu^* - \lambda>0\) as the discrete analogue of \(\varepsilon\) for this section.

To bound the average queue length under UCB1~\cite{Auer2002FinitetimeAO}, we follow the two-stage approach of~\cite{Freund2023}, which is of independent interest. \emph{Stage I (learning):} Standard regret arguments bound the time needed to identify (up to estimation error) the arm with the largest mean. \emph{Stage II (control):} Once the estimate is sufficiently accurate, a Lyapunov drift analysis~\cite{Neely2010coml} characterizes the regime in which the system operates near the optimal rate.

The next two lemmas state the main results for each stage and are proved in Appendices~\ref{sec:stageI} and~\ref{sec:stageII}.

\begin{lemma}[Stage I]\label{thm:log(H)}
For any integer \(H>{d}\) and any \(\tilde\Delta\in(0,\vedisc)\),
\[
\frac{1}{H}\sum_{t=1}^H \EE\bigl[Q(t)\bigr]
\;\le\;
\frac{2}{\vedisc-\tilde\Delta}
\;+\;
\frac{8\,{d}\,\log H}{\tilde\Delta}.
\]
Optimizing over \(\tilde\Delta\) further gives
\[
\frac{1}{H}\sum_{t=1}^H\EE\bigl[Q(t)\bigr]
\;\le\;
\frac{2}{\vedisc}\Bigl(1+2\sqrt{\,{d}\,\log H}\Bigr)^2
\;\le\;
\frac{4}{\vedisc}\Bigl(1+4 {d}\,\log H\Bigr).
\]
\end{lemma}

\begin{lemma}[Stage II]\label{thm:stage2}
% Consider the slotted single-server queue under the UCB1 rate-selection policy (Algorithm~\ref{alg:ucb-queue}) with mesh size \({d}\). Let the capacity gap satisfy \(\vedisc = \mu^* - \lambda > 0\). Then for all \(H\ge1\), 
For any integer \(H\geq{1}\), the time-average expected queue size obeys
\[
\frac{1}{H}\,\EE\Bigl[\sum_{t=1}^H Q(t)\Bigr]
\;\le\;
\frac{4}{\vedisc}
\;+\;
\frac{32}{H\,\vedisc}
\;+\;
16\,\frac{{d}^2}{H\,\vedisc^4}
\Bigl(
\tfrac13
+8\bigl(1+\log H\bigr)\,\log H
\Bigr)^2.
\]
\end{lemma}

Next, we discretize the continuous domain \([0,1]\) into a finite grid and treat it as a discrete-armed bandit. Running UCB1 and applying the same analysis yields an upper bound that matches the converse in Section~\ref{sec:converse} up to polylogarithmic factors.

\subsection{Discretization and Combining the Two Bounds}
Set the mesh parameter \(d \triangleq \lceil 3/\velb \rceil\). We divide the interval \([0,1]\) into \({d}\) equally spaced nonzero points.\footnote{Unlike \cite{Kleinberg2003}, which chooses \({d}\) as a function of a time horizon, here \({d}\) depends solely on the known capacity slack \(\velb\).} Define
\begin{align}\label{eqn:defs__}
\mathcal{K} \;=\; \{1,2,\dots,{d}\}, 
\qquad
r_k \;=\; \frac{k}{{d}}\quad(k\in\mathcal{K}),
\qquad
\mathcal{R} \;=\; \{r_k : k\in\mathcal{K}\}.
\end{align}
We treat each discretized rate as an arm in a stochastic multi-armed bandit. The service of arm \(k\) at time \(t\) is $S_k(t) \;=\; r_k \,\mathbbm{1}\{\,r_k \le C(t)\,\}$.
Let \(\mu_k \triangleq \EE[S_k(t)]\) and \(\mu^* \triangleq \max_{k\in\mathcal{K}} \mu_k\). With the discrete capacity gap $\vedisc \;\triangleq\; \mu^* - \lambda$, it follows from Corollary~\ref{corr:main_corr} that $\vedisc \;\ge\; \frac{2}{3}\,\velb$. The following theorem states the main result of this section and is proved in Appendix~\ref{app:Proof-thm:main_known}.

\begin{theorem}\label{thm:main_known}
Running the UCB1 algorithm~\cite{Auer2002FinitetimeAO} on the arm set \(\mathcal{K}\) defined in~\eqref{eqn:defs__}, with \(d = \lceil 3/\velb \rceil\), yields the following bound for any horizon \(H \in \mathbb{N}\):
\[
    \frac{1}{H}\sum_{t=1}^H \EE\bigl[Q(t)\bigr]
    \;\le\;
    \begin{cases}
        \dfrac{1767\,\log\!\bigl(1/\velb\bigr)}{\velb^{2}}, & \text{if } \varepsilon \le e^{-3},\\[1.25ex]
        \dfrac{12378}{\velb^{2}}, & \text{if } \varepsilon \ge e^{-3}.
    \end{cases}
\]
\end{theorem}

\section{Conclusion}
In this paper, we studied the problem of achieving a bounded time-average queue size in a single-queue, single-server problem with a special partial feedback structure and a continuous rate space. When the arrival rate has a distance bounded above by $\varepsilon > 0 $ to the capacity region, and when $\varepsilon$ is known, we achieved $\mathcal{O}(\log(1/\varepsilon)/\varepsilon^2)$ worst-case time-average expected queue size with a simple UCB-based algorithm. The simple UCB algorithm was extended to an algorithm that runs in phases to handle the case when $\varepsilon$ is not known. This algorithm yields $\mathcal{O}(\log^{3.5}(1/\varepsilon)/\varepsilon^3)$ worst-case time-average expected queue size.  We also established a converse result that states, for any algorithm, regardless of whether the algorithm knows $\varepsilon$, there exists an environment that yields a worst-case time-average expected queue size of the order $\Omega(1/\varepsilon^2)$. We conjecture that when $\varepsilon$ is unknown, an algorithm achieving a time-average expected queue size of order $\mathcal{O}(\log^{\alpha}(1/\varepsilon)/\varepsilon^2)$ for some $\alpha > 0$ is possible. Designing such an algorithm would close the gap between the lower and upper bounds in this setting and is left as future work. Another interesting future direction is to extend the continuum-armed queueing framework to multi-queue or networked settings.
\bibliographystyle{IEEEtran}
\bibliography{main}

@String{Computing = "Computing" }

@String{Computer = "{IEEE} Computer" }

@String{Academic = "Academic Press" }

@String{Springer = "Springer-Verlag" }

@ArtifactSoftware{R,
    title = {R: A Language and Environment for Statistical Computing},
    author = {{R Core Team}},
    organization = {R Foundation for Statistical Computing},
    address = {Vienna, Austria},
    year = {2019},
    url = {https://www.R-project.org/},
}

@Book{Neely2010coml,
  author={M. J. {Neely}},
  title={Stochastic Network Optimization with Application to Communication and Queueing Systems}, 
  year={2010},
  publisher = {Morgan \& Claypool}
  }

@INPROCEEDINGS{Kleinberg2003,
  author={Kleinberg, R. and Leighton, T.},
  booktitle={44th Annual IEEE Symposium on Foundations of Computer Science, 2003. Proceedings.}, 
  title={The value of knowing a demand curve: bounds on regret for online posted-price auctions}, 
  year={2003},
  volume={},
  number={},
  pages={594-605},
  doi={10.1109/SFCS.2003.1238232}
}

@InProceedings{10.1007/978-3-540-72927-3_33,
author="Auer, Peter
and Ortner, Ronald
and Szepesv{\'a}ri, Csaba",
editor="Bshouty, Nader H.
and Gentile, Claudio",
title="Improved Rates for the Stochastic Continuum-Armed Bandit Problem",
booktitle="Learning Theory",
year="2007",
publisher="Springer Berlin Heidelberg",
address="Berlin, Heidelberg",
pages="454--468",
isbn="978-3-540-72927-3"
}

@inproceedings{Freund2023,
 author = {Freund, Daniel and Lykouris, Thodoris and Weng, Wentao},
 booktitle = {Advances in Neural Information Processing Systems},
 editor = {A. Oh and T. Naumann and A. Globerson and K. Saenko and M. Hardt and S. Levine},
 pages = {6532--6544},
 publisher = {Curran Associates, Inc.},
 title = {Quantifying the Cost of Learning in Queueing Systems},
 volume = {36},
 year = {2023}
}

@book{lattimore_szepesvári_2020, place={Cambridge}, title={Bandit Algorithms}, DOI={10.1017/9781108571401}, publisher={Cambridge University Press}, author={Lattimore, Tor and Szepesvári, Csaba}, year={2020}}

@article{Lai1985,
author = {Lai, T.L and Robbins, Herbert},
title = {Asymptotically efficient adaptive allocation rules},
year = {1985},
issue_date = {March, 1985},
publisher = {Academic Press, Inc.},
address = {USA},
month ={Mar.},
volume = {6},
number = {1},
issn = {0196-8858},
doi = {10.1016/0196-8858(85)90002-8},
journal = {Adv. Appl. Math.},
pages = {4–22},
numpages = {19}
}

@article{Auer2002FinitetimeAO,
  title={Finite-time Analysis of the Multiarmed Bandit Problem},
  author={Peter Auer and Nicol{\`o} Cesa-Bianchi and Paul Fischer},
  journal={Machine Learning},
  year={2002},
  month = {May.},
  volume={47},
  pages={235-256},
}

@inproceedings{auer1995gambling,
  title={Gambling in a rigged casino: The adversarial multi-armed bandit problem},
  author={Auer, Peter and Cesa-Bianchi, Nicolo and Freund, Yoav and Schapire, Robert E},
  booktitle={Proceedings of IEEE 36th annual foundations of computer science},
  pages={322--331},
  year={1995},
  month ={Aug.},
  organization={IEEE}
}

@InProceedings{pmlr-v75-wei18a,
  title = 	 {More Adaptive Algorithms for Adversarial Bandits},
  author =       {Wei, Chen-Yu and Luo, Haipeng},
  booktitle = 	 {Proceedings of the 31st  Conference On Learning Theory},
  pages = 	 {1263--1291},
  year = 	 {2018},
  month = {Jul.},
  volume = 	 {75},
}

@ARTICLE{Huang2024,
  author={Huang, Jiatai and Golubchik, Leana and Huang, Longbo},
  journal={IEEE/ACM Transactions on Networking}, 
  title={When \uppercase{L}yapunov Drift Based Queue Scheduling Meets Adversarial Bandit Learning}, 
  year={2024},
  volume={},
  number={},
  pages={1-11},
  doi={10.1109/TNET.2024.3374755}
}

@inproceedings{Krishnasamy2016,
author = {Krishnasamy, Subhashini and Sen, Rajat and Johari, Ramesh and Shakkottai, Sanjay},
title = {Regret of queueing bandits},
year = {2016},
isbn = {9781510838819},
publisher = {Curran Associates Inc.},
address = {Red Hook, NY, USA},
booktitle = {Proceedings of the 30th International Conference on Neural Information Processing Systems},
pages = {1677–1685},
numpages = {9},
location = {Barcelona, Spain},
series = {NIPS'16}
}

@book{Cesa-Bianchi_Lugosi_2006, place={Cambridge}, title={Prediction, Learning, and Games}, publisher={Cambridge University Press}, author={Cesa-Bianchi, Nicolo and Lugosi, Gabor}, year={2006}}

@book{Bubeck2012,
  author={Bubeck, Sébastien and Nicolò, Cesa-Bianchi},
  title={Regret Analysis of Stochastic and Nonstochastic Multi-armed Bandit Problems},
  year={2012},
  volume={},
  number={},
  pages={},
  doi={10.1561/2200000024}}

@ARTICLE{Kong2021,
  author={Kong, Xiangqi and Lu, Ning and Li, Bin},
  journal={IEEE Transactions on Mobile Computing}, 
  title={Optimal Scheduling for Unmanned Aerial Vehicle Networks With Flow-Level Dynamics}, 
  year={2021},
  volume={20},
  number={3},
  month={Mar.},
  pages={1186-1197},
  doi={10.1109/TMC.2019.2952848}
}

@article{Maguluri2012StochasticMO,
  title={Stochastic models of load balancing and scheduling in cloud computing clusters},
  author={Siva Theja Maguluri and Rayadurgam Srikant and Lei Ying},
  journal={2012 Proceedings IEEE INFOCOM},
  year={2012},
  pages={702-710},
  month={Mar.},
}

@INPROCEEDINGS{Cruz2003,
  author={Cruz, R.L. and Santhanam, A.V.},
  booktitle={IEEE INFOCOM 2003.}, 
  title={Optimal routing, link scheduling and power control in multihop wireless networks}, 
  year={2003},
  month={Apr.},
  volume={1},
  number={},
  pages={702-711},
  doi={10.1109/INFCOM.2003.1208720}
}

@ARTICLE{Palomar2006,
  author={Palomar, D.P. and Mung Chiang},
  journal={IEEE Journal on Selected Areas in Communications}, 
  title={A tutorial on decomposition methods for network utility maximization}, 
  year={2006},
  month={Jul.},
  volume={24},
  number={8},
  pages={1439-1451},
  doi={10.1109/JSAC.2006.879350}
}

@INPROCEEDINGS{Jung2007,
  author={Jung, Kyomin and Shah, Devavrat},
  booktitle={2007 IEEE International Symposium on Information Theory}, 
  title={Low Delay Scheduling in Wireless Network}, 
  year={2007},
  volume={},
  month={Jun.},
  number={},
  pages={1396-1400},
  doi={10.1109/ISIT.2007.4557418}}

@ARTICLE{Songwu1999,
  author={Songwu Lu and Bharghavan, V. and Srikant, R.},
  journal={IEEE/ACM Transactions on Networking}, 
  title={Fair scheduling in wireless packet networks}, 
  year={1999},
  volume={7},
  month={Aug.},
  number={4},
  pages={473-489},
  doi={10.1109/90.793003}
}

@ARTICLE{Cai2024,
  author={Cai, Xuelian and Fan, Yixin and Yue, Wenwei and Fu, Yuchuan and Li, Changle},
  journal={IEEE Transactions on Vehicular Technology}, 
  title={Dependency-Aware Task Scheduling for Vehicular Networks Enhanced by the Integration of Sensing, Communication and Computing}, 
  year={2024},
  volume={},
  month={Apr.},
  number={},
  pages={1-16},
  doi={10.1109/TVT.2024.3389951}
}

@ARTICLE{Fu2023,
  author={Fu, Xinzhe and Modiano, Eytan},
  journal={IEEE/ACM Transactions on Networking}, 
  title={Optimal Routing to Parallel Servers With Unknown Utilities—Multi-Armed Bandit With Queues}, 
  year={2023},
  volume={31},
  number={5},
  pages={1997-2012},
  doi={10.1109/TNET.2022.3227136}}

@article{Krishnasamy2020,
author = {Krishnasamy, Subhashini and Sen, Rajat and Johari, Ramesh and Shakkottai, Sanjay},
year = {2020},
month = {09},
pages = {},
title = {Learning Unknown Service Rates in Queues: A Multiarmed Bandit Approach},
volume = {69},
journal = {Operations Research},
doi = {10.1287/opre.2020.1995}
}

@article{Agrawal1995,
author = {Agrawal, Rajeev},
title = {The Continuum-Armed Bandit Problem},
journal = {SIAM Journal on Control and Optimization},
volume = {33},
number = {6},
pages = {1926-1951},
year = {1995},
doi = {10.1137/S0363012992237273},

URL = { 
    
        https://doi.org/10.1137/S0363012992237273
    
    

},
eprint = { 
    
        https://doi.org/10.1137/S0363012992237273
    
    

}
}

@inproceedings{Kleinberg2008,
author = {Kleinberg, Robert and Slivkins, Aleksandrs and Upfal, Eli},
title = {Multi-armed bandits in metric spaces},
year = {2008},
isbn = {9781605580470},
publisher = {Association for Computing Machinery},
address = {New York, NY, USA},
url = {https://doi.org/10.1145/1374376.1374475},
doi = {10.1145/1374376.1374475},
booktitle = {Proceedings of the Fortieth Annual ACM Symposium on Theory of Computing},
pages = {681–690},
numpages = {10},
location = {Victoria, British Columbia, Canada},
series = {STOC '08}
}

@InProceedings{pmlr-v134-podimata21a,
  title = 	 {Adaptive Discretization for Adversarial Lipschitz Bandits},
  author =       {Podimata, Chara and Slivkins, Alex},
  booktitle = 	 {Proceedings of Thirty Fourth Conference on Learning Theory},
  pages = 	 {3788--3805},
  year = 	 {2021},
  editor = 	 {Belkin, Mikhail and Kpotufe, Samory},
  volume = 	 {134},
  series = 	 {Proceedings of Machine Learning Research},
  month = 	 {15--19 Aug},
  publisher =    {PMLR},
  url = 	 {https://proceedings.mlr.press/v134/podimata21a.html}
}

@article{Chichportich2023,
author = {Chichportich, Jeremy and Kharroubi, Idris},
year = {2023},
month = {10},
pages = {},
title = {Discrete-Time Mean-Field Stochastic Control with Partial Observations},
volume = {88},
journal = {Applied Mathematics \& Optimization},
doi = {10.1007/s00245-023-10068-4}
}

@article{KAELBLING199899,
title = {Planning and acting in partially observable stochastic domains},
journal = {Artificial Intelligence},
volume = {101},
number = {1},
pages = {99-134},
year = {1998},
issn = {0004-3702},
doi = {https://doi.org/10.1016/S0004-3702(98)00023-X},
url = {https://www.sciencedirect.com/science/article/pii/S000437029800023X},
author = {Leslie Pack Kaelbling and Michael L. Littman and Anthony R. Cassandra}
}

@inproceedings{zuo2021combinatorial,
  title={Combinatorial multi-armed bandits for resource allocation},
  author={Zuo, Jinhang and Joe-Wong, Carlee},
  booktitle={2021 55th Annual Conference on Information Sciences and Systems (CISS)},
  pages={1--4},
  year={2021},
  organization={IEEE}
}

@article{Shardin02012017,
author = {Anton A. Shardin and Ralf Wunderlich},
title = {Partially observable stochastic optimal control problems for an energy storage},
journal = {Stochastics},
volume = {89},
number = {1},
pages = {280--310},
year = {2017},
publisher = {Taylor \& Francis},
doi = {10.1080/17442508.2016.1166506},


URL = { 
    
        https://doi.org/10.1080/17442508.2016.1166506
    
    

},
eprint = { 
    
        https://doi.org/10.1080/17442508.2016.1166506
    
    

}

}

@ARTICLE{Sun2021,
  author={Sun, Ke and Schlotfeldt, Brent and Pappas, George J. and Kumar, Vijay},
  journal={IEEE Transactions on Robotics}, 
  title={Stochastic Motion Planning Under Partial Observability for Mobile Robots With Continuous Range Measurements}, 
  year={2021},
  volume={37},
  number={3},
  pages={979-995},
  doi={10.1109/TRO.2020.3042129}}

@article {Velentzas117598,
	author = {Velentzas, George and Tzafestas, Costas and Khamassi, Mehdi},
	title = {Bridging Computational Neuroscience and Machine Learning on Non-Stationary Multi-Armed Bandits},
	elocation-id = {117598},
	year = {2017},
	doi = {10.1101/117598},
	publisher = {Cold Spring Harbor Laboratory},
	URL = {https://www.biorxiv.org/content/early/2017/05/23/117598},
	eprint = {https://www.biorxiv.org/content/early/2017/05/23/117598.full.pdf},
	journal = {bioRxiv}
}

@article{lyu2017optimal,
  title={Optimal schedule of mobile edge computing for Internet of Things using partial information},
  author={Lyu, Xinchen and Ni, Wei and Tian, Hui and Liu, Ren Ping and Wang, Xin and Giannakis, Georgios B and Paulraj, Arogyaswami},
  journal={IEEE Journal on Selected Areas in Communications},
  volume={35},
  number={11},
  pages={2606--2615},
  year={2017},
  publisher={IEEE}
}

@article{li2020multi,
  title={Multi-armed-bandit-based spectrum scheduling algorithms in wireless networks: A survey},
  author={Li, Feng and Yu, Dongxiao and Yang, Huan and Yu, Jiguo and Karl, Holger and Cheng, Xiuzhen},
  journal={IEEE Wireless Communications},
  volume={27},
  number={1},
  pages={24--30},
  year={2020},
  publisher={IEEE}
}

@inproceedings{Nguyen,
author = {Nguyen, Quang Minh and Modiano, Eytan},
title = {Learning to Schedule in Non-Stationary Wireless Networks With Unknown Statistics},
year = {2023},
isbn = {9781450399265},
publisher = {Association for Computing Machinery},
address = {New York, NY, USA},
url = {https://doi.org/10.1145/3565287.3610258},
doi = {10.1145/3565287.3610258},
booktitle = {Proceedings of the Twenty-Fourth International Symposium on Theory, Algorithmic Foundations, and Protocol Design for Mobile Networks and Mobile Computing},
pages = {181–190},
numpages = {10},
location = {Washington, DC, USA},
series = {MobiHoc '23}
}

@misc{tang_joint_2021,
	title = {Joint {Link} {Rate} {Selection} and {Channel} {State} {Change} {Detection} in {Block}-{Fading} {Channels}},
	url = {http://arxiv.org/abs/2108.09728},
	doi = {10.48550/arXiv.2108.09728},
	urldate = {2025-02-18},
	publisher = {arXiv},
	author = {Tang, Haoyue and Hou, Xinyu and Wang, Jintao and Song, Jian},
	month = aug,
	year = {2021},
	note = {arXiv:2108.09728 [cs]},
}

@InProceedings{Yang2023,
  title = 	 {Learning While Scheduling in Multi-Server Systems With Unknown Statistics: MaxWeight with Discounted UCB},
  author =       {Yang, Zixian and Srikant, R. and Ying, Lei},
  booktitle = 	 {Proceedings of The 26th International Conference on Artificial Intelligence and Statistics},
  pages = 	 {4275--4312},
  year = 	 {2023},
  editor = 	 {Ruiz, Francisco and Dy, Jennifer and van de Meent, Jan-Willem},
  volume = 	 {206},
  series = 	 {Proceedings of Machine Learning Research},
  month = 	 {25--27 Apr},
  publisher =    {PMLR},
  url = 	 {https://proceedings.mlr.press/v206/yang23d.html}
}

@INPROCEEDINGS{Tarzjani2025,
  author={Tarzjani, Faezeh Dehghan and Krishnamachari, Bhaskar},
  booktitle={2025 34th International Conference on Computer Communications and Networks (ICCCN)}, 
  title={Computing the Saturation Throughput for Heterogeneous p-CSMA in a General Wireless Network}, 
  year={2025},
  volume={},
  number={},
  pages={1-7},
  doi={10.1109/ICCCN65249.2025.11133870}}

@article{combes_optimal_2019,
	title = {Optimal {Rate} {Sampling} in 802.11 {Systems}: {Theory}, {Design}, and {Implementation}},
	volume = {18},
	copyright = {https://ieeexplore.ieee.org/Xplorehelp/downloads/license-information/IEEE.html},
	issn = {1536-1233, 1558-0660, 2161-9875},
	shorttitle = {Optimal {Rate} {Sampling} in 802.11 {Systems}},
	url = {https://ieeexplore.ieee.org/document/8409322/},
	doi = {10.1109/TMC.2018.2854758},
	language = {en},
	number = {5},
	urldate = {2025-02-18},
	journal = {IEEE Transactions on Mobile Computing},
	author = {Combes, Richard and Ok, Jungseul and Proutiere, Alexandre and Yun, Donggyu and Yi, Yung},
	month = may,
	year = {2019},
	pages = {1145--1158},
}

@article{cho_use_2023,
	title = {Use of {Logarithmic} {Rates} in {Multi}-{Armed} {Bandit}-{Based} {Transmission} {Rate} {Control} {Embracing} {Frame} {Aggregations} in {Wireless} {Networks}},
	volume = {13},
	copyright = {https://creativecommons.org/licenses/by/4.0/},
	issn = {2076-3417},
	url = {https://www.mdpi.com/2076-3417/13/14/8485},
	doi = {10.3390/app13148485},
	language = {en},
	number = {14},
	urldate = {2025-02-17},
	journal = {Applied Sciences},
	author = {Cho, Soohyun},
	month = jul,
	year = {2023},
	pages = {8485},
}

@inproceedings{tong_rate_2023,
	title = {Rate {Adaptation} with {Correlated} {Multi}-{Armed} {Bandits} in 802.11 {Systems}},
	url = {https://ieeexplore.ieee.org/document/10233472/?arnumber=10233472},
	doi = {10.1109/ICCC57788.2023.10233472},
	urldate = {2025-02-17},
	booktitle = {2023 {IEEE}/{CIC} {International} {Conference} on {Communications} in {China} ({ICCC})},
	author = {Tong, Yuzhou and Fan, Jiakun and Cai, Xuhong and Chen, Yi},
	month = aug,
	year = {2023},
	note = {ISSN: 2377-8644},
	pages = {1--6},
}

@inproceedings{le_multi-armed_2024,
	title = {Multi-{Armed} {Bandits} in {IEEE} 802.11ac: {Efficient} {Algorithms} and {Testbed} {Experiments}},
	shorttitle = {Multi-{Armed} {Bandits} in {IEEE} 802.11ac},
	url = {https://ieeexplore.ieee.org/document/10705888/?arnumber=10705888},
	doi = {10.1109/CQR62340.2024.10705888},
	urldate = {2025-02-17},
	booktitle = {2024 {IEEE} {International} {Workshop} {Technical} {Committee} on {Communications} {Quality} and {Reliability} ({CQR})},
	author = {Le, Martin and Peng, Bile and Pawar, Sankalp Prakash and Auer, Peter and Thapa, Prashiddha Dhoj and Hühn, Thomas and Jorswieck, Eduard A.},
	month = sep,
	year = {2024},
	note = {ISSN: 2640-060X},
	pages = {7--12},
}

@misc{combes_optimal_2013,
	title = {Optimal {Rate} {Sampling} in 802.11 {Systems}},
	url = {http://arxiv.org/abs/1307.7309},
	doi = {10.48550/arXiv.1307.7309},
	urldate = {2025-02-16},
	publisher = {arXiv},
	author = {Combes, Richard and Proutiere, Alexandre and Yun, Donggyu and Ok, Jungseul and Yi, Yung},
	month = sep,
	year = {2013},
	note = {arXiv:1307.7309 [cs]},
}

\appendix
\section{Sequence of Independent Random variables}
\begin{lemma}\label{lemma:hoeff_with_dep}
    Consider a sequence of independent, zero mean $c$-sub Gaussian random variables $X_1,X_2,\dots$. Also consider a positive integer valued random variable $G$ which is possible dependent on the sequence $X_1,X_2,\dots$. For any $\delta \in (0,1)$, we have that
    \begin{align}
        \mathbb{P}\left\{\frac{1}{G}\sum_{g=1}^G X_g \geq \sqrt{\frac{2c^2\log\left( \frac{G(G+1)}{\delta}\right)}{G}}\right\} \leq \delta \nonumber
    \end{align}
    \begin{proof}
    Let us define $\bar{X} = \frac{1}{G}\sum_{g=1}^G X_g$. Notice that
        \begin{align}
            &\mathbb{P}\left\{\frac{1}{G}\sum_{g=1}^G X_g\geq \sqrt{\frac{2c^2\log{\frac{G(G+1)}{\delta}}}{G}}\right\} = \sum_{g=1}^{\infty}\mathbb{P}\left\{\frac{1}{G}\sum_{g=1}^G X_g \geq \sqrt{\frac{2c^2\log{\frac{G(G+1)}{\delta}}}{G}},G= g\right\}  \nonumber\\& =  \sum_{g=1}^{\infty}\mathbb{P}\left\{\frac{1}{g}\sum_{t=1}^g X_t \geq \sqrt{\frac{2c^2\log{\frac{g(g+1)}{\delta}}}{g}},G= g\right\} \leq_{(a)}   \sum_{g=1}^{\infty}\mathbb{P}\left\{\frac{1}{g}\sum_{t=1}^g X_t \geq \sqrt{\frac{2c^2\log{\frac{g(g+1)}{\delta}}}{g}}\right\} \nonumber\\& \leq_{(b)} \sum_{g=1}^{\infty}e^{-\frac{\left(\sqrt{2g\log{\frac{g(g+1)}{\delta}}}\right)^2}{2g}}= \sum_{g=1}^{\infty}\frac{\delta}{g(g+1)} = \sum_{g=1}^{\infty}\left(\frac{\delta}{g} - \frac{\delta}{g+1}\right) = \delta \nonumber
        \end{align}
         where (a) follows since for any two events $A,B$, $P(A,B) \leq P(A)$, and (b) follows from the standard Hoeffding inequality.
        \end{proof}
\end{lemma}
\section{KL Divergence Between Bernoulli Random Variables}
For $x,y \in [0,1]$, we use the notation $D_{\text{KL}}(x\lVert y)$ to denote the KL divergence between two Bernoulli$(x)$, and Bernoulli$(y)$ random variables. We have the following lemma.
\begin{lemma}\label{lemma:bern_KL}
     We have the following.
     \begin{enumerate}
         \item Fix $c \in (0,1)$. Then $D_{\text{KL}}(x\lVert c)$ is nonincreasing in $x$ in the interval $[0,c]$, and nondecreasing in $x$ in the interval $[c,1]$.
         \item Given $a,b \in (0,1)$, we have that
        \begin{align}
            D_{\text{KL}}(a\lVert b) \leq \frac{(a-b)^2}{b(1-b)} \nonumber
        \end{align}
        \begin{proof}
            We only prove 2, since 1 is a simple calculus exercise. Note that
            \begin{align}
                &D_{\text{KL}}(a\lVert b)\nonumber\\ &= a\ln(a/b) + (1-a)\ln((1-a)/(1-b)) \leq_{(a)} a\left(\frac{a}{b}-1\right)+(1-a)\left(\frac{1-a}{1-b}-1\right) \nonumber\\&=\frac{(a-b)^2}{b(1-b)}  \nonumber
            \end{align}
            where for (a) we have used $\ln(x) \leq x - 1$ for all $x>0$.
        \end{proof}
     \end{enumerate}
\end{lemma}
\section{Proofs and Lemmas for Section~\ref{sec:unknown}}
\subsection{Proof of Corollary~\ref{thm:main_thm}}\label{app:main_thm}
We use the following lemma.
\begin{lemma}\label{lemma:max_expabc}
    Consider positive real numbers $a,b,c$ and the function $f(x) = \log^a(bx)/x^c$. The maximum value of $f$ in $[1,\infty)$ is $\max\left\{\log^a(b)/c,b^c \left(\frac{a}{c\exp(1)}\right)^a\right\}$. 
\end{lemma}

Substituting $\delta = 1/6$, \eqref{eqn:llate_stage_2} translates to
\begin{align}\label{eqn:inequ_vital_simp}
    &\frac{1}{H}\sum_{t = 1}^H \mathbb{E}\{Q(t)\} \nonumber\\&\leq  \frac{65\times2^{\frac{2}{p-1}}\gamma }{(\gamma-1) \varepsilon}+\frac{\left(2^{\frac{p+1}{p-1}}+2\right)\gamma^{\frac{2}{1-2\delta}}}{\varepsilon^{\frac{2}{1-2\delta}}C^{^{\frac{2}{1-2\delta}}}}+1 +\frac{2^{\frac{5q}{2}-\delta q+3}C^q\gamma^{2q}(7-2\delta)^q\log^{q+2}(2H)H^{1-\frac{q}{2}-\delta q}}{(\gamma-1)^{2q}\varepsilon^{2q}(1-(q/2))^2}  \nonumber\\&\ \ \ \ + \frac{2^{2q +3}\gamma^{2q}(7-2\delta)^q\log^{q+2}(2H)H^{1-q}}{(\gamma-1)^{2q}\varepsilon^{2q}(1-(q/2))^2} \nonumber\\& \leq \frac{65\times2^{\frac{2}{p-1}}\gamma }{(\gamma-1) \varepsilon}+\frac{\left(2^{\frac{p+1}{p-1}}+2\right)\gamma^{3}}{\varepsilon^{3}C^{3}}+1 +\frac{2^{\frac{7q}{3}+3}C^q\gamma^{2q}(20/3)^q\log^{q+2}(2H)}{(\gamma-1)^{2q}\varepsilon^{2q}(1-(q/2))^2H^{\frac{2q-3}{3}}}  \nonumber\\&\ \ \ \ + \frac{2^{2q +3}\gamma^{2q}(20/3)^q\log^{q+2}(2H)}{(\gamma-1)^{2q}\varepsilon^{2q}(1-(q/2))^2H^{q-1}}
\end{align}
holds for all $C \in (0,1),\gamma>1 , q \in (1,2)$.

Using $q = 3/2$, \eqref{eqn:inequ_vital_simp} translates to
\begin{align}\label{eqn:bound_2}
 \frac{1}{H}\sum_{t = 1}^H \mathbb{E}\{Q(t)\} &\leq \frac{130\gamma }{(\gamma-1)\varepsilon}+\frac{6\gamma^{3}}{\varepsilon^{3}C^{3}}+1 + \frac{24928C^{1.5}\gamma^{3}}{(\gamma-1)^{3}\varepsilon^{3}}\log^{3.5}(2H) +\frac{17627\gamma^{3}}{(\gamma-1)^{3}\varepsilon^{3}}\frac{\log^{3.5}(2H)}{H^{0.5}}.
\end{align}
holds for all $C \in (0,1),\gamma>1$.

\begin{align}\label{eqn:T_def}
n_{\text{stage}} = \frac{e^{\alpha}}{2\varepsilon^{6}}
\end{align}
where $\alpha$ is a constant to be determined later.

\noindent{\textbf{Case 1} $H \leq n_{\text{stage}}$: We have from \eqref{eqn:bound_2}, for $H \leq n_{\text{stage}}$,
\begin{align}\label{eqn:bound_1}
     &\frac{1}{H}\sum_{t = 1}^H \mathbb{E}\{Q(t)\nonumber\\&\leq \frac{130\gamma }{(\gamma-1)\varepsilon}+\frac{6\gamma^{3}}{\varepsilon^{3}C^{3}}+1 + \frac{24928C^{1.5}\gamma^{3}}{(\gamma-1)^{3}\varepsilon^{3}}\log^{3.5}(2H) +\frac{17627\gamma^{3}}{(\gamma-1)^{3}\varepsilon^{3}}\frac{\log^{3.5}(2H)}{H^{0.5}} \nonumber \\&\leq_{(a)} \frac{130\gamma }{(\gamma-1)\varepsilon}+\frac{6\gamma^{3}}{\varepsilon^{3}C^{3}}+1 + \frac{24928C^{1.5}\gamma^{3}}{(\gamma-1)^{3}\varepsilon^{3}}\log^{3.5}(2H) +\frac{687453\gamma^{3}}{(\gamma-1)^{3}\varepsilon^{3}}\nonumber \\&\leq_{(b)} \frac{130\gamma }{(\gamma-1)\varepsilon}+\frac{6\gamma^{3}}{\varepsilon^{3}C^{3}}+1 + \frac{24928C^{1.5}\gamma^{3}}{(\gamma-1)^{3}\varepsilon^{3}}\log^{3.5}\left(\frac{e^{\alpha}}{\varepsilon^{6}}\right) +\frac{687453\gamma^{3}}{(\gamma-1)^{3}\varepsilon^{3}} \\&\leq \frac{130\gamma }{(\gamma-1)\varepsilon}+\frac{6\gamma^{3}}{\varepsilon^{3}C^{3}}+1 + \frac{24928C^{1.5}\gamma^{3}}{(\gamma-1)^{3}\varepsilon^{3}}\left(\alpha+\log\left(\frac{1}{\varepsilon}\right)\right)^{3.5} +\frac{687453\gamma^{3}}{(\gamma-1)^{3}\varepsilon^{3}}\nonumber \\&\leq 1+\frac{130\gamma }{(\gamma-1)\varepsilon}+\frac{1}{\varepsilon^3}\left[\frac{6\gamma^{3}}{C^{3}} + \frac{141015 \alpha^{2.5}C^{1.5}\gamma^{3}}{(\gamma-1)^{3}}+\frac{687453\gamma^{3}}{(\gamma-1)^{3}}\right]+\frac{\log^{3.5}\left(\frac{1}{\varepsilon}\right)}{\varepsilon^3}\frac{141015 C^{1.5}\gamma^{3}}{(\gamma-1)^{3}}\nonumber
\end{align}
where (a) follows using Lemma~\ref{lemma:max_expabc} with $(a,b,c)= (3.5,2,0.5)$, (b) follows from $H \leq n_{\text{stage}}$ and the definition of $n_{\text{stage}}$ in \eqref{eqn:T_def}, and (c) follows from $(a+b)^d \leq 2^{d-1}(a^d+b^d)$ for $d \geq 1$.

\noindent
\textbf{Case 2} $H \geq n_{\text{stage}}$: Notice that for $q \in (3/2,1)$, \eqref{eqn:inequ_vital_simp} reduces to
\begin{align}\label{eqn:bound_3}
    &\frac{1}{H}\sum_{t = 1}^H \mathbb{E}\{Q(t)\} \nonumber\\&\leq \frac{65\times2^{\frac{2}{p-1}}\gamma }{(\gamma-1) \varepsilon}+\frac{\left(2^{\frac{p+1}{p-1}}+2\right)\gamma^{3}}{\varepsilon^{3}C^{3}}+1 +\frac{2^{\frac{7q}{3}+3}C^q\gamma^{2q}(20/3)^q\log^{q+2}(2H)}{(\gamma-1)^{2q}\varepsilon^{2q}(1-(q/2))^2H^{\frac{2q-3}{3}}}  \nonumber\\&\ \ \ \ + \frac{2^{2q +3}\gamma^{2q}(20/3)^q\log^{q+2}(2H)}{(\gamma-1)^{2q}\varepsilon^{2q}(1-(q/2))^2H^{q-1}} \nonumber\\& \leq  \frac{65\times2^{\frac{2}{p-1}}\gamma }{(\gamma-1) \varepsilon}+\frac{\left(2^{\frac{p+1}{p-1}}+2\right)\gamma^{3}}{\varepsilon^{3}C^{3}}+1 +\frac{2^{\frac{7q}{3}+3}C^q\gamma^{2q}(20/3)^q}{(\gamma-1)^{2q}\varepsilon^{2q}(1-(q/2))^2H^{\frac{2q-3}{6}}}\frac{\log^{q+2}(2H)}{H^{\frac{2q-3}{6}}}  \nonumber\\&\ \ \ \ + \frac{2^{2q +3}\gamma^{2q}(20/3)^q}{(\gamma-1)^{2q}\varepsilon^{2q}(1-(q/2))^2H^{\frac{2q-3}{6}}}\frac{\log^{q+2}(2H)}{H^{\frac{4q-3}{6}}}\nonumber\\& \leq_{(a)}  \frac{65\times2^{\frac{2}{p-1}}\gamma }{(\gamma-1) \varepsilon}+\frac{\left(2^{\frac{p+1}{p-1}}+2\right)\gamma^{3}}{\varepsilon^{3}C^{3}}+1 +\frac{2^{\frac{7q}{3}+3+\frac{2q-3}{6}}C^q\gamma^{2q}(20/3)^q}{e^{\frac{\alpha(2q-3)}{6}}(\gamma-1)^{2q}\varepsilon^{3}(1-(q/2))^2}\frac{\log^{q+2}(2H)}{H^{\frac{2q-3}{6}}}  \nonumber\\&\ \ \ \ + \frac{2^{2q +3+\frac{2q-3}{6}}\gamma^{2q}(20/3)^q}{e^{\frac{\alpha(2q-3)}{6}}(\gamma-1)^{2q}\varepsilon^{3}(1-(q/2))^2}\frac{\log^{q+2}(2H)}{H^{\frac{4q-3}{6}}}\nonumber\\&\leq_{(b)}  \frac{65\times2^{\frac{2}{p-1}}\gamma }{(\gamma-1) \varepsilon}+\frac{\left(2^{\frac{p+1}{p-1}}+2\right)\gamma^{3}}{\varepsilon^{3}C^{3}}+1 +\frac{2^{\frac{7q}{3}+3+\frac{2q-3}{6}+\frac{2q-3}{6}}C^q\gamma^{2q}(20/3)^q}{e^{\frac{\alpha(2q-3)}{6}}(\gamma-1)^{2q}\varepsilon^{3}(1-(q/2))^2}\left(\frac{6(q+2)}{2q-3}\right)^{q+2}  \nonumber\\&\ \ \ \ + \frac{2^{2q +3+\frac{2q-3}{6}+\frac{4q-3}{6}}\gamma^{2q}(20/3)^q}{e^{\frac{\alpha(2q-3)}{6}}(\gamma-1)^{2q}\varepsilon^{3}(1-(q/2))^2}\left(\frac{q+2}{e (4q-3)}\right)^{q+2}\nonumber\\&= 1+\frac{65\times2^{\frac{2}{p-1}}\gamma }{(\gamma-1) \varepsilon}
  + \frac{1}{\varepsilon^3}\left[
    \begin{aligned}[t]
      &\frac{\left(2^{\frac{p+1}{p-1}}+2\right)\gamma^{3}}{C^{3}}
      + \frac{2^{3q+2}C^q\gamma^{2q}(20/3)^q}{e^{\frac{\alpha(2q-3)}{6}}(\gamma-1)^{2q}\!\left(1-\frac{q}{2}\right)^{\!2}}
        \left(\frac{6(q+2)}{2q-3}\right)^{\!q+2}\\
      &\quad
      + \frac{2^{3q+2}\gamma^{2q}(20/3)^q}{e^{\frac{\alpha(2q-3)}{6}}(\gamma-1)^{2q}\!\left(1-\frac{q}{2}\right)^{\!2}}
        \left(\frac{6(q+2)}{e(4q-3)}\right)^{\!q+2}
    \end{aligned}
  \right]
\end{align}
where (a) follows since $H \geq n_{\text{stage}}$, (b) follows by applying Lemma~\ref{lemma:max_expabc} with $(a,b,c)= (q+2,2,(2q-3)/6)$, and $(a,b,c)= (q+2,2,(4q-3)/6)$

Now, let $C = 0.04$, $\gamma = 4$, $q = 1.81$, $\alpha = 26$,  \eqref{eqn:bound_1} reduces to
\begin{align}
    \frac{1}{H}\sum_{t = 1}^H \mathbb{E}\{Q(t)\} \leq 1 + \frac{174}{\varepsilon} + \frac{16846843}{\varepsilon^3} + \frac{2675\log^{3.5}\left(\frac{1}{\varepsilon}\right)}{\varepsilon^3}
\end{align}
for all $H \leq n_{\text{stage}}$, and \eqref{eqn:bound_3} reduces to
\begin{align}
    \frac{1}{H}\sum_{t = 1}^H \mathbb{E}\{Q(t)\} \leq 1 + \frac{267}{\varepsilon} + \frac{16651943}{\varepsilon^3} 
\end{align}
for all $H\geq n_{\text{stage}}$. Combining the two bounds, we get
\begin{align}
    \frac{1}{H}\sum_{t = 1}^H \mathbb{E}\{Q(t)\} \leq 1 + \frac{267}{\varepsilon} +\frac{16846843}{\varepsilon^3} + \frac{2675\log^{3.5}\left(\frac{1}{\varepsilon}\right)}{\varepsilon^3}
\end{align}
for all $H \in \mathbb{N}$ as desired.

\subsection{Proof of Lemma~\ref{lemma:init_lemma}}\label{app:init_lemma}
 
We have for $l \geq 2$, $T^{\text{sum}}_{l} = \sum_{\tau = 1}^{l-1}T_{\tau} = \sum_{\tau = 1}^{l-1}2^{\tau+2} \geq 2^{l+1}$. Also, from the definition of $a(t)$ in \eqref{eqn:a_t_def} and the definition of $T^{\text{sum}}_l$ in \eqref{eqn:t_l_sum}, we have $t \geq T^{\text{sum}}_{a(t)}$ for all $t \in \mathbb{N}$. Furthermore, since $I \geq 9$, we have $a(I) \geq 2$. Hence,
\begin{align}
    I \geq T^{\text{sum}}_{a(I)} \geq 2^{a(I)+1} = T_{a(I)-1}  = \frac{T_{a(I)}}{2} \nonumber
\end{align}
This establishes 1. The above also establishes 2 since $2^{a(I)+1} \leq I$. 

To prove 3, notice that for all $l \geq 1$
\begin{align}
    T_l^{\text{sum}} = \sum_{\tau = 1}^{l-1} T_{\tau} = \sum_{\tau = 1}^{l-1}2^{\tau+2} \leq  2^{l+2} = T_l
\end{align}
Hence,
\begin{align}
    \sum_{n=1}^{a(I)} T^{\text{sum}}_n \leq \sum_{n=1}^{a(I)} T_n = \sum_{n=1}^{a(I)} 2^{n+2} \leq 2^{a(I)+3}\leq 4I\nonumber
\end{align}
where the last inequality follows from $I \geq 2^{a(I)+1}$.

\subsection{Proof of Lemma~\ref{lemma:good_event_lemma_app}}\label{app:good_event}
The first part follows trivially, since the decisions in phase $l$ do not depend on the queue length or the feedback received in past phases. For the probability bound, notice that if $u = 0$, we have
\begin{align}
    \mathcal{G}_l(0) = \left\{ \mu_{l,k} \in \left[-\sqrt{\frac{(7-2\delta)\log(T_l)}{4}},\sqrt{\frac{(7-2\delta)\log(T_l)}{4}}\right]\right\}.
\end{align}
We have $\mathbb{P}\{\mathcal{G}_l(0)\} = 1$, since $\mu_{l,k} \in [0,1]$ and
\begin{align}
    \sqrt{\frac{(7-2\delta)\log(T_l)}{4}} \geq_{(a)} \sqrt{\frac{3}{2}\log(T_l)} \geq_{(b)} \sqrt{\frac{3}{2}\log(8)} \geq 1
\end{align}
where (a) follows since $\delta \in (0,1/2)$, and (b) follows since $T_l = 2^{l+2} \geq 8$. Hence, we assume $u > 1$. To prove the bound for $u>1$, we begin with the following lemma.
\begin{lemma}\label{lemma:good_event_lemma}
    Consider a phase $l \geq b$, $u \in [1:T_l]$, and a \emph{rate level} $k \in \mathcal{K}_l$. For any $d \geq 2$ and $M \geq 1$, define \begin{align}
        U_{l,k}(u) = \bar{\mu}_{l,k}(u) + \sqrt{\frac{d\log(M(u+1))}{2(1 \vee N_{l,k}(u))}},  \nonumber
    \end{align}
    We have that
    \begin{align}
        &\mathbb{P}\{ \mu_{l,k} > U_{l,k}(u)\} \leq \frac{1}{M^d(u+1)^{d-2}}  \nonumber
    \end{align}
    and
    \begin{align}
        \mathbb{P}\left\{ \mu_{l,k} < U_{l,k}(u)-2\sqrt{\frac{d\log(M(u+1))}{2(1 \vee N_{l,k}(u))}}\right\} \leq \frac{1}{M^d(u+1)^{d-2}}. \nonumber
    \end{align}
    \begin{proof}
  Define the random variable $G = (1 \vee N_{l,k}(u))$. Also, let $\tilde{\mu}_{l,k}(s)$ denote the empirical mean of the $k$-th arm in the $l$-th episode when it is chosen for the $s$-th time (If it is chosen less than $s$ times, for the remainder consider random variables independently sampled from the same distribution). Let us denote $p_0 = \mathbb{P}\{N_{l,k}(u) = 0\}$, and $p_1 = 1-p_0$

     From Lemma~\ref{lemma:hoeff_with_dep} using $\delta = \frac{1}{M^d(u+1)^{d-2}}$, we have that
     \begin{align}\label{eqn:1111}
         \frac{1}{M^d(u+1)^{d-2}} &\geq \mathbb{P}\left\{  \mu_{l,k}-\tilde{\mu}_{l,k}(G) \geq \sqrt{\frac{\log(M^dG(G+1)(u+1)^{d-2})}{2G}}\right\} \nonumber\\& \geq_{(a)} \mathbb{P}\left\{  \mu_{l,k}-\tilde{\mu}_{l,k}(G) \geq \sqrt{\frac{d\log(M(u+1))}{2G}}\right\} \nonumber\\& =_{(b)}  \mathbb{P}\left\{  \mu_{l,k}-\tilde{\mu}_{l,k}(G) \geq \sqrt{\frac{d\log(M(u+1))}{2G}} \Bigg{|} N_{l,k}(u) > 0\right\}p_1 \nonumber\\&\ \ \ \  + \mathbb{P}\left\{  \mu_{l,k}-\tilde{\mu}_{l,k}(G) \geq \sqrt{\frac{d\log(M(u+1))}{2G}} \Bigg{|} N_{l,k}(u) = 0\right\}p_0\nonumber\\& =  \mathbb{P}\left\{  \mu_{l,k}-\bar{\mu}_{l,k}(u)  \geq\sqrt{\frac{d\log(M(u+1))}{2G}} \Bigg{|} N_{l,k}(u) > 0\right\}p_1  \nonumber\\&\ \ \ \ + \mathbb{P}\left\{  \mu_{l,k}-\tilde{\mu}_{l,k}(G) \geq \sqrt{\frac{d\log(M(u+1))}{2G}} \Bigg{|} N_{l,k}(u) = 0\right\}p_0\nonumber\\& \geq  \mathbb{P}\left\{  \mu_{l,k}-\bar{\mu}_{l,k}(u)  \geq\sqrt{\frac{d\log(M(u+1))}{2G}} \Bigg{|} N_{l,k}(u) > 0\right\}p_1
     \end{align}
     where (a) follows since $G \leq u$, (b) follows since given $N_{l,k}(u) > 0$, we have $\tilde{\mu}_{l,k}(G)= \tilde{\mu}_{l,k}(N_{l,k}(u)) = \bar{\mu}_{l,k}(u)$. Now, notice that
     \begin{align}
         &\mathbb{P}\left\{  \mu_{l,k}-\bar{\mu}_{l,k}(u)  \geq\sqrt{\frac{d\log(M(u+1))}{2G}} \Bigg{|} N_{l,k}(u) = 0\right\}  \nonumber\\& = \mathbb{P}\left\{  \mu_{l,k}  \geq\sqrt{d\log(M(u+1))} \Bigg{|} N_{l,k}(u) = 0\right\} \nonumber\\&= 0, \nonumber
     \end{align}
     where the last equality is true since
     \begin{align}
         \mu_{l,k} \leq 1 < \sqrt{2\log(2)} \leq \sqrt{d\log(M(u+1))}. \nonumber
     \end{align}
     Using the above in \eqref{eqn:1111}, we have,
     \begin{align}
         \frac{1}{M^d(u+1)^{d-2}}&\geq\mathbb{P}\left\{  \mu_{l,k}-\bar{\mu}_{l,k}(u)  \geq\sqrt{\frac{d\log(M(u+1))}{2G}} \Bigg{|} N_{l,k}(u) > 0\right\}p_1\nonumber\\&\ \ \ \  + \mathbb{P}\left\{  \mu_{l,k}-\bar{\mu}_{l,k}(u)  \geq\sqrt{\frac{d\log(M(u+1))}{2G}} \Bigg{|} N_{l,k}(u) = 0\right\}p_0  \nonumber\\& = \mathbb{P}\left\{  \mu_{l,k}-\bar{\mu}_{l,k}(u)  \geq\sqrt{\frac{d\log(M(u+1))}{2G}}\right\} \nonumber
     \end{align}
     as desired. The second inequality follows by repeating a similar argument.
    \end{proof}
\end{lemma}
Now we move onto the main proof. Using $d = 2$ and $M = T_l^{3/4-\delta/2}$ in Lemma~\ref{lemma:good_event_lemma}, we have that
\begin{align}
     \mathbb{P}\left\{ \bar{\mu}_{l,k}(u)-\sqrt{\frac{\log(T_l^{3/4-\delta/2}(u+1))}{(1 \vee N_{l,k}(u))}} < \mu_{l,k} < \bar{\mu}_{l,k}(u)+\sqrt{\frac{\log(T_l^{3/4-\delta/2}(u+1))}{(1 \vee N_{l,k}(u))}}\right\} \leq \frac{2}{T_l^{3/2-\delta}} \nonumber
\end{align}
for each $u \in [1:T_l-1]$ and $k \in \mathcal{K}_l$. Using $u \in [1:T_l-1]$
\begin{align}
     \mathbb{P}\left\{ \bar{\mu}_{l,k}(u)-\sqrt{\frac{\log(T_l^{7/4-\delta/2})}{(1 \vee N_{l,k}(u))}} < \mu_{l,k} < \bar{\mu}_{l,k}(u)+\sqrt{\frac{\log(T_l^{7/4-\delta/2})}{(1 \vee N_{l,k}(u))}}\right\} \leq \frac{2}{T_l^{3/2-\delta}}  \nonumber
\end{align}
Using a union bound over all $k \in \mathcal{K}_l$, and noticing that $d_l \leq T_l^{\frac{1}{2}-\delta}+1 \leq 2T_l^{\frac{1}{2}-\delta}$, we have the result,
\subsection{Proof of Lemma~\ref{lemma:temp_1}}\label{app:temp_1}

Define $k^* = \arg\max_{k \in \mathcal{K}_l} \mu_{l,k}$. Notice that
\begin{align}
  &\mathbb{E}\{Q(u+T^{\text{sum}}_l)[\lambda - \mu_{l,K_l(u)}]|\mathcal{G}_l(u-1)\}\mathbb{P}\{\mathcal{G}_l(u-1)\} \nonumber\\
    & \leq_{(a)} \mathbb{E}\left\{Q(u+T^{\text{sum}}_l)\left[\lambda - \text{UCB}_{l,K_l(u)}(u-1) +2\sqrt{\frac{(7-2\delta)\log\left(T_l\right)}{4(1\vee N_{l,K_l(u)}(u-1))}}\right]\Bigg{|}\mathcal{G}_l(u-1) \right\}\nonumber\\&\quad\quad\cdot\mathbb{P}\{\mathcal{G}_l(u-1)\} \nonumber
    \\&\leq_{(b)} 
    \mathbb{E}\left\{ Q(u+T^{\text{sum}}_l)\left[\lambda - \text{UCB}_{l,k^*}(u-1) +2\sqrt{\frac{(7-2\delta)\log\left(T_l\right)}{4(1 \vee N_{l,K_l(u)}(u-1))}}\right]\Bigg{|}\mathcal{G}_l(u-1)\right\}\nonumber\\&\quad\quad\cdot\mathbb{P}\{\mathcal{G}_l(u-1)\}\nonumber\\&\leq_{(c)} \mathbb{E}\left\{Q(T^{\text{sum}}_l+u)\left[\lambda - \mu_{l,k^*}+2\sqrt{\frac{(7-2\delta)\log\left(T_l\right)}{4(1 \vee N_{l,K_l(u)}(u-1))}}\right]\Bigg{|}\mathcal{G}_l(u-1)\right\} \mathbb{P}\{\mathcal{G}_l(u-1)\} \nonumber
    \\&\leq_{(d)} 
    -\frac{(\gamma-1)\varepsilon}{\gamma}\mathbb{E}\{Q(T^{\text{sum}}_l+u)|\mathcal{G}_l(u-1)\}\mathbb{P}\{\mathcal{G}_l(u-1)\} \nonumber\\&\ \ \ \ + \sqrt{7-2\delta}\mathbb{E}\left\{Q(T^{\text{sum}}_l+u)\sqrt{\frac{\log\left(T_l\right)}{(1 \vee N_{l,K_l(u)}(u-1))}}\Bigg{|}\mathcal{G}_l(u-1)\right\} \mathbb{P}\{\mathcal{G}_l(u-1)\}\nonumber
    \\&\leq  
    -\frac{(\gamma-1)\varepsilon}{\gamma}\mathbb{E}\{Q(u+T^{\text{sum}}_l)\}+\frac{(\gamma-1)\varepsilon}{\gamma}\mathbb{E}\{Q(T^{\text{sum}}_l+u)|\mathcal{G}_l^c(u-1)\}\mathbb{P}\{\mathcal{G}_l^c(u-1)\} \nonumber\\&\ \ \ \ +  \sqrt{7-2\delta}\mathbb{E}\left\{Q(T^{\text{sum}}_l+u)\sqrt{\frac{\log\left(T_l\right)}{(1 \vee N_{l,K_l(u)}(u-1))}}\right\}\nonumber
    \\&\leq_{(e)}  
    -\frac{(\gamma-1)\varepsilon}{\gamma}\mathbb{E}\{Q(u+T^{\text{sum}}_l)\}+\frac{4T^{\text{sum}}_{l+1}}{T_l} +  \sqrt{7-2\delta}\mathbb{E}\left\{Q(T^{\text{sum}}_l+u)\sqrt{\frac{\log\left(T_l\right)}{(1 \vee N_{l,K_l(u)}(u-1))}}\right\} \nonumber
\end{align}  
where (a) and (c) follow from the definition of the \emph{good event} $\mathcal{G}_l(u-1)$ in \eqref{eqn:good_event}, (b) follows from $\text{UCB}_{l,K_l(u)}(u-1) \leq \text{UCB}_{l,k}(u-1)$
for any $k \in \mathcal{K}_l$ due to the decision in \eqref{eqn:decision}, (d) follows from Corollary~\ref{corr:main_corr} since $l \geq b$, and (e) follows by Lemma~\ref{lemma:term_2} and $\varepsilon \leq 1$.
\subsection{Proof of Lemma~\ref{lemma:proof_1_to_til_H}}\label{app:proof_1_to_til_H}
First, notice that
\begin{align}\label{eqn:1_to_tbsum}
    \sum_{t=1}^{ T^{\text{sum}}_b}\mathbb{E}\{Q(t)\} \leq\sum_{t=1}^{ T^{\text{sum}}_b}(t-1) =\frac{T^{\text{sum}}_b(T^{\text{sum}}_b-1)}{2} \leq \frac{(T^{\text{sum}}_b)^2}{2},
\end{align}
where the first inequality follows from Lemma~\ref{eqn:det_queue_bound}.

Summing \eqref{eqn:basic_eqn} within a phase and then again over the phases $b,b+1,\dots,a(I)$, we have
\begin{align}
&\sum_{l=b}^{a( I)}\sum_{u=1}^{\tilde{T}_l}\left[\mathbb{E}\{\frac{1}{2}Q^2(u+T^{\text{sum}}_l+1)\} - \frac{1}{2}\mathbb{E}\{Q^2(u+T^{\text{sum}}_l)\}\right]\nonumber\\& \leq \sum_{l=b}^{a( I)}\sum_{u=1}^{\tilde{T}_l} \Bigg{(}\frac{1}{2} -\frac{(\gamma -1)\varepsilon}{\gamma}\mathbb{E}\{Q(u+T^{\text{sum}}_l)\}\nonumber\\&\ \ \ \ \quad\quad\quad\quad +  \sqrt{7-2\delta}\mathbb{E}\left\{Q(T^{\text{sum}}_l+u)\sqrt{\frac{\log\left(T_l\right)}{(1 \vee N_{l,K_l(u)}(u-1))}}\right\} +\frac{8T^{\text{sum}}_{l+1}}{T_l}\Bigg{)} \nonumber\\& \leq \left(\sum_{l=b}^{a( I)}\sum_{u=1}^{\tilde{T}_l} \frac{1}{2}\right) - \frac{(\gamma -1)\varepsilon}{\gamma}\left( \sum_{l=b}^{a( I)}\sum_{u=1}^{\tilde{T}_l} \mathbb{E}\{Q(u+T^{\text{sum}}_l)\}\right)\nonumber\\&\ \ \ \   + \sqrt{7-2\delta} \mathbb{E}\left\{\sum_{l=b}^{a( I)}\sum_{u=1}^{\tilde{T}_l} Q(T^{\text{sum}}_l+u)\sqrt{\frac{\log\left(T_l\right)}{(1 \vee N_{l,K_l(u)}(u-1))}}\right\}  + \sum_{l=b}^{a( I)}8T^{\text{sum}}_{l+1}\nonumber\\& =_{(a)} \frac{( I-T^{\text{sum}}_b)}{2} - \frac{(\gamma -1)\varepsilon}{\gamma}\left( \sum_{t=T^{\text{sum}}_b+1}^{ I} \mathbb{E}\{Q(t)\}\right) \nonumber\\&\ \ \ \ + \sqrt{7-2\delta} \mathbb{E}\left\{\sum_{l=b}^{a( I)}\sum_{u=1}^{\tilde{T}_l} Q(T^{\text{sum}}_l+u)\sqrt{\frac{\log\left(T_l\right)}{(1 \vee N_{l,K_l(u)}(u-1))}}\right\} +8\sum_{l=b}^{a( I)}T^{\text{sum}}_{l+1} \nonumber
\end{align}
where (a) follows since $\sum_{l=b}^{a( I)}\sum_{u=1}^{\tilde{T}_l} f(u+T^{\text{sum}}_l) = \sum_{t= T^{\text{sum}}_b+1}^{I} f(t)$ (first summing inside a phase and then summing over phases vs. summing over the horizon).
In addition, the first line of the above inequality is equal to $\frac{1}{2}\mathbb{E}\{Q^2( I+1)\} - \frac{1}{2}\mathbb{E}\{Q^2(T^{\text{sum}}_b+1)\}$ due to the same reason. Hence, we have
\begin{align}
     &\frac{1}{2}\mathbb{E}\{Q^2( I+1)\} - \frac{1}{2}\mathbb{E}\{Q^2(T^{\text{sum}}_b+1)\}  \nonumber\\&\leq  \frac{( I-T^{\text{sum}}_b)}{2} - \frac{(\gamma -1)\varepsilon}{\gamma}\left( \sum_{t=T^{\text{sum}}_b+1}^{I} \mathbb{E}\{Q(t)\}\right)  \nonumber\\& \ \ \ \ + \sqrt{7-2\delta} \mathbb{E}\left\{\sum_{l=b}^{a( I)}\sum_{u=1}^{\tilde{T}_l} Q(T^{\text{sum}}_l+u)\sqrt{\frac{\log\left(T_l\right)}{(1 \vee N_{l,K_l(u)}(u-1))}}\right\}  + 8\sum_{l=b}^{a( I)}T^{\text{sum}}_{l+1} \nonumber
\end{align}
Rearranging, we have
\begin{align}
&\sum_{t=T^{\text{sum}}_b+1}^{I}\mathbb{E}\{Q(t)\} \nonumber\\&\leq \frac{\gamma (I-T^{\text{sum}}_b)}{2(\gamma-1)\varepsilon}+ \frac{\gamma\sqrt{7-2\delta}}{(\gamma-1)\varepsilon} \mathbb{E}\left\{\sum_{l=b}^{a( I)}\sum_{u=1}^{\tilde{T}_l} Q(T^{\text{sum}}_l+u)\sqrt{\frac{\log\left(T_l\right)}{(1 \vee N_{l,K_l(u)}(u-1))}}\right\} \nonumber\\&\ \ \ \ + \frac{8\gamma}{(\gamma-1)\varepsilon}\sum_{l=b}^{a( I)}T^{\text{sum}}_{l+1}  + \frac{\gamma}{2(\gamma-1)\varepsilon}\mathbb{E}\{Q^2(T^{\text{sum}}_b+1)\}-\frac{\gamma}{2(\gamma-1)\varepsilon}\mathbb{E}\{Q^2( I+1)\}\nonumber\\& \leq \frac{\gamma I}{2(\gamma-1)\varepsilon}+  \frac{\gamma\sqrt{7-2\delta}}{(\gamma-1)\varepsilon} \mathbb{E}\left\{\sum_{l=b}^{a( I)}\sum_{u=1}^{\tilde{T}_l} Q(T^{\text{sum}}_l+u)\sqrt{\frac{\log\left(T_l\right)}{(1 \vee N_{l,K_l(u)}(u-1))}}\right\} \nonumber\\&\ \ \ \ + \frac{8\gamma}{(\gamma-1)\varepsilon}\sum_{l=b}^{a( I)}T^{\text{sum}}_{l+1}   + \frac{\gamma}{2(\gamma-1)\varepsilon}\left[\mathbb{E}\{Q^2(T^{\text{sum}}_b+1)\}-\mathbb{E}\{Q^2( I+1)\}\right]_+
\end{align}
where the last inequality follows since $x \leq [x]_+$.
Adding the above inequality with \eqref{eqn:1_to_tbsum}, we have the result.
\subsection{Proof of Lemma~\ref{lemma:q_pow_bnd}}\label{app:q_pow_bnd}
Define a permutation $\phi: [1:n] \to [1:n]$ such that $x_{\phi(t)} \leq x_{\phi(t+1)}$ for all $t \in [1:n-1]$. Let $y_t = x_{\phi(t)}$. Notice that $\phi$ can be chosen in such a way that for all $t \in [1:n-1]$,
\begin{align}\label{eqn:cond}
    x_{\phi(t)} = x_{\phi(t+1)} \implies \phi(t) < \phi(t+1).
\end{align}
We first establish $|y_t - y_{t+1}| \leq 1$ for all $t \in [1:n-1]$. 

\noindent
\textbf{Case 1} $\phi(t) < \phi(t+1)$: Define $k = \min\{i : \phi(t) < i \leq \phi(t+1), x_i \geq y_t\}$. Notice that the definition is valid since $ \phi(t+1) \in \{i : \phi(t) < i \leq \phi(t+1), x_i \geq y_t\}$. 

\textbf{Claim 1} $x_{k-1}  \leq  x_{\phi(t)}$:
To prove this, notice that from the definition of $k$, we have $k -1 \geq \phi(t)$. If $k-1 = \phi(t)$, we are done. If $k-1 > \phi(t)$ we have $\phi(t+1) \geq k> k-1 > \phi(t)$, which gives $x_{k-1} < x_{\phi(t)}$ from the definition of $k$. Hence, we are done with the proof of claim 1.

\textbf{Claim 2} $x_{\phi(t+1)}  \leq x_k$: Notice that there exists $t^{'} \in [1:n]$ such that $k = \phi(t^{'})$. We are done if we establish $t^{'} >t$. To prove this, notice that from the definition of $k$, we have $x_{\phi(t^{'})} \geq  x_{\phi(t)}$. If $x_{\phi(t^{'})} >  x_{\phi(t)}$, we directly have $t^{'} > t$. Hence, it remains to consider $x_{\phi(t^{'})} =  x_{\phi(t)}$. Notice that from the definition of $k$, we have $\phi(t^{'}) > \phi(t)$. Combining $x_{\phi(t^{'})} =  x_{\phi(t)}$ with $\phi(t^{'}) > \phi(t)$, the result follows from \eqref{eqn:cond}.

Now, combining the two claims, we have $x_{k-1}  \leq  x_{\phi(t)} \leq x_{\phi(t+1)}  \leq x_k$. Hence, $|y_t - y_{t+1}| = |x_{\phi(t)}-x_{\phi(t+1)}| \leq |x_k - x_{k-1}|\leq  1$, proving case 1.

\noindent
\textbf{Case 2}  $\phi(t) > \phi(t+1)$: Notice that due to \eqref{eqn:cond}, we have $x_{\phi(t+1)} > x_{\phi(t)}$. Define $k = \max\{i : \phi(t) > i \geq  \phi(t+1), x_i > y_{t}\}$. The definition is valid since $\phi(t+1) \in \{i : \phi(t) > i \geq  \phi(t+1), x_i > y_{t}\}$.

\textbf{Claim 1} $x_{k+1}  \leq  x_{\phi(t)}$:
Notice that from the definition of $k$, we have $k+1 \leq \phi(t)$. If $k+1 = \phi(t)$, we are done. If $k+1 < \phi(t)$ we have $\phi(t) > k+1> k \geq \phi(t+1)$, which gives $x_{k+1} \leq x_{\phi(t)}$ from the definition of $k$. Hence, we are done with the proof of claim 1.

\textbf{Claim 2} $x_k \geq x_{\phi(t+1)}$:  Take $t^{'}$ such that $k = \phi(t^{'})$. We prove that $t^{'} >t$ which establishes the result. To prove this, notice that from the definition of $k$, we have $x_{\phi(t^{'})} >  x_{\phi(t)}$. Hence, we have $t^{'} > t$ as desired.

Now, combining the two claims, we have $x_{k}  \geq  x_{\phi(t+1)} \geq x_{\phi(t)}  \geq x_{k+1}$. Hence, $|y_t - y_{t+1}| = |x_{\phi(t)}-x_{\phi(t+1)}| \leq |x_k - x_{k-1}|\leq  1$, proving case 2.

Now, we prove the lemma. Let $s = \lceil y_n - y_1 \rceil$. Notice that $s \in [0:n-1]$. We have
\begin{align}
S &= \sum_{t=1}^n y_t \geq \sum_{t=0}^{s-1} y_{n-t} \geq \sum_{t=0}^{s-1} \left(y_{n}-t\right)  = sy_n - \frac{s(s-1)}{2} = \frac{s(2y_n-s+1)}{2} \geq \frac{(y_n-y_1)(y_n+y_1)}{2} \nonumber\\&= \frac{y_n^2 - y_1^2}{2} = \frac{y_n^2}{2} \nonumber
\end{align}
where the last inequality follows from $y_n -y_1 \leq s \leq y_n -y_1 +1$. Hence, 
\begin{align}
    D^p = \sum_{t=1}^n x_t^p \leq y_n^{p-1} S  \leq (2S)^{\frac{p-1}{2}}S \leq 2^{\frac{p-1}{2}}S^{\frac{p+1}{2}} \nonumber
\end{align}
\subsection{Proof of Lemma~\ref{lemma:bounding_ucb_slack_holder}}\label{app:bounding_ucb_slack_holder}
Notice that,
\begin{align}
&\sum_{u=1}^{\tilde{T}_l}\left(\frac{\log\left(T_l\right)}{(1 \vee N_{l,K(u)}(u-1))}\right)^{\frac{q}{2}} = \log^{\frac{q}{2}}(T_l)\sum_{k \in \mathcal{K}_l} \sum_{\substack{u=1\\K_l(u) = k}}^{\tilde{T}_l}\frac{1}{(1 \vee N_{l,K(u)}(u-1))^{q/2}}  \nonumber\\&\leq \log^{\frac{q}{2}}(T_l)\sum_{k \in \mathcal{K}_l} \left(1+\sum_{\tau = 1}^{N_{l,k}(\tilde{T}_l)-1}\frac{1}{\tau^{q/2}}\right) \leq_{(a)} \log^{\frac{q}{2}}(T_l)\sum_{k \in \mathcal{K}_l} \left(1+\frac{[N_{l,k}(\tilde{T}_l)-1]^{1-\frac{q}{2}}-\frac{q}{2})}{1-(q/2)}\right) \nonumber\\&\leq \frac{(1-q)\log^{\frac{q}{2}}(T_l)d_l}{1-\frac{q}{2}} +\log^{\frac{q}{2}}(T_l)\sum_{k \in \mathcal{K}_l}\frac{[N_{l,k}(\tilde{T}_l)-1]^{1-\frac{q}{2}}}{1-(q/2)}
\nonumber\\& \leq_{(b)} \log^{\frac{q}{2}}(T_l)d_l^{\frac{q}{2}}\frac{\left[\sum_{k \in \mathcal{K}_l}(N_{l,k}(\tilde{T}_l)-1)\right]^{1-\frac{q}{2}}}{1-(q/2)} \leq \log^{\frac{q}{2}}(T_l)\frac{d_l^{\frac{q}{2}}T_l^{1-(q/2)}}{1-\frac{q}{2}} \nonumber
\end{align}
where (a) follows from $\sum_{\tau = 1}^y \frac{1}{\tau^x} \leq \frac{y^{1-x}-x}{1-x}$ for any $x \leq 1$, and (b) follows from $q \in (1,2)$, and $\left(\sum_{i=1}^n a_i^x\right) \leq \left(\sum_{i=1}^n a_i\right)^xn^{1-x}$
for $a_i \geq 0$ and $x \in (0,1]$.
\subsection{Proof of Lemma~\ref{lemma:abX_lemma}}\label{app:abX_lemma}
We prove the first inequality. The second inequality follows from a simple application of Jensen's inequality to the convex function $f(x) = x^d$.

Assume the contrary that $X >  a^{\frac{1}{d}} + b$. Hence we have $X > a^{1/d}$ and $X-b > a^{1/d}$. Hence, 
\begin{align}
  X^{d-1}(X-b) > a^{\frac{d-1}{d}} a^{\frac{1}{d}} = a, \nonumber
\end{align}
which contradicts the original condition. 

\section{Proofs for Section~\ref{sec:converse}}
\subsection{Proof of Lemma~\ref{lemma:ref_def_lemma}}\label{app:ref_def_lemma}
Recall Claim 1 and Claim 2 of Section~\ref{sec:prel_cons}.
The first part follows from Claim 2 since
       \begin{align}
          [x_k,x_{k+1}] = [x_k,x_k+|\mathcal{I}_k|]\subseteq_{(a)} [7/12, 2/3 +3\varepsilon] \subset (7/12,1), \nonumber
       \end{align}
       where (a) follows combining Claim 2 with $7/12 = x_1 \leq x_k \leq 2/3$ for all $k \in [1:K]$, and the last inclusion follows since we assumed $\varepsilon \leq 1/144$. 
       
       For the first inequality of the second part, note that the intervals $\mathcal{I}_1,\mathcal{I}_2,\dots,\mathcal{I}_K$, are disjoint, and cover $(7/12,2/3]$. Claim 2 ensures each of these intervals have size at most $3\varepsilon$, so $3\varepsilon K \geq \frac{2}{3} - \frac{7}{12}$, which yields $K \geq 1/(36\varepsilon)$. For the second inequality of the second part, notice that
       \begin{align}
           \frac{2}{3} \leq x_{K+1} = \frac{7}{12} \left(1+ \frac{2\varepsilon}{\frac{1}{2}-\varepsilon}\right)^{K} \leq \frac{7}{12}\left(1+ \frac{2(1/144)}{\frac{1}{2}-(1/144)}\right)^{K} \nonumber 
       \end{align}
       where the last inequality follows since $x/(0.5-x)$ is nondecreasing in $[0,0.5)$. This gives $K \geq 5$.
\subsection{Proof of Theorem~\ref{thm:lower bound}   }\label{app:lower_bound}

Fix $T \in \mathbb{N}$. For $t \in [1:T]$, recall that $V(t) \in[ 0,1]$ is the rate chosen in time slot $t$, and $C(t) \in [0,1]$ is the channel capacity in time slot $t$. Define 
\begin{align}
    B(t) = \mathbbm{1}\{V(t) \leq C(t) \}. \nonumber
\end{align}

Also, let us denote by $\mathcal{H}(t)$, the history up to time $t$, that is,
\begin{align}
    \mathcal{H}(t) = \{A(1), \dots, A(t),B(1),\dots, B(t)\} \nonumber
\end{align} 
Notice that $\mathcal{H}(t) \in \mathcal{B}_t$ where $\mathcal{B}_t = \{0,1\}^{2t}$. A deterministic policy for selecting rates can be denoted by a sequence of functions $f^{1}, f^2,\dots,f^T$, where $f^{\tau}:\mathcal{B}_{\tau-1} \times \{0,1\} \to [0,1]$ and given $ A(\tau) = a,\mathcal{H}(\tau-1) = \vec{h}$, we have $V(\tau) = f^{\tau}(\vec{h},a)$. We prove the theorem for deterministic policies. From Fubini's theorem, the result extends to randomized algorithms. Recall the definition of the interval $\mathcal{I}_k$ in \eqref{eqn:int_x_k}. In particular,
\begin{align}
    \mathcal{I}_k = (x_k,x_{k+1}],
\end{align}
where $x_k$ is defined by
\begin{align}
    &x_1 = \frac{7}{12}, \text{ and }x_{k+1} = x_k\left(1+ \frac{2\varepsilon}{\frac{1}{2}-\varepsilon}\right).
\end{align}
For $t \in [0:T]$, and $k \in [1:K]$, let $N_k(t)$ denote the number of times the policy chooses an action in $\mathcal{I}_k$ in the first $t$ time slots. In particular,
\begin{align}
    N_k(t) = \sum_{\tau = 1}^t \mathbbm{1}\{V(\tau) \in \mathcal{I}_k\}. \nonumber
\end{align}

Let us also define an additional environment (Environment 0) with Bernoulli$(1/2)$ arrivals and $C(t)$ sampled from $X_0$ with
\begin{align}
    F_{X_0}(x) = \begin{cases}
        0 & \text{ if } x \leq \frac{1}{2} - \varepsilon\\
        1 - \frac{\frac{1}{2}-\varepsilon}{x} & \text{ if }\frac{1}{2} - \varepsilon < x < 1\\
         1 & \text{ if }x \geq 1.\\
    \end{cases} \nonumber
\end{align}
It is easy to see that the function $g_0 :[0,1] \to [0,1]$ given by $g_0(x) = x\mathbb{P}\{X_0 \geq x\}$ satisfies
\begin{align}
    g_0(x) = \begin{cases}
        x & \text{ for } x \in \left[0,\frac{1}{2}-\varepsilon\right]\\
        \frac{1}{2} - \varepsilon & \text{ for } x \in \left[\frac{1}{2}-\varepsilon,1\right]
    \end{cases} \nonumber
\end{align}
Note that we cannot stabilize the queue in the environment since $\max_{x \in [0,1]} g_0(x) = \frac{1}{2} - \varepsilon < \frac{1}{2} = \lambda$. 

Figure~\ref{fig_sim_11} denotes the plots of the above CDF and the function $g_0$.
\begin{figure}
\centering
{\includegraphics[width=0.48\linewidth]{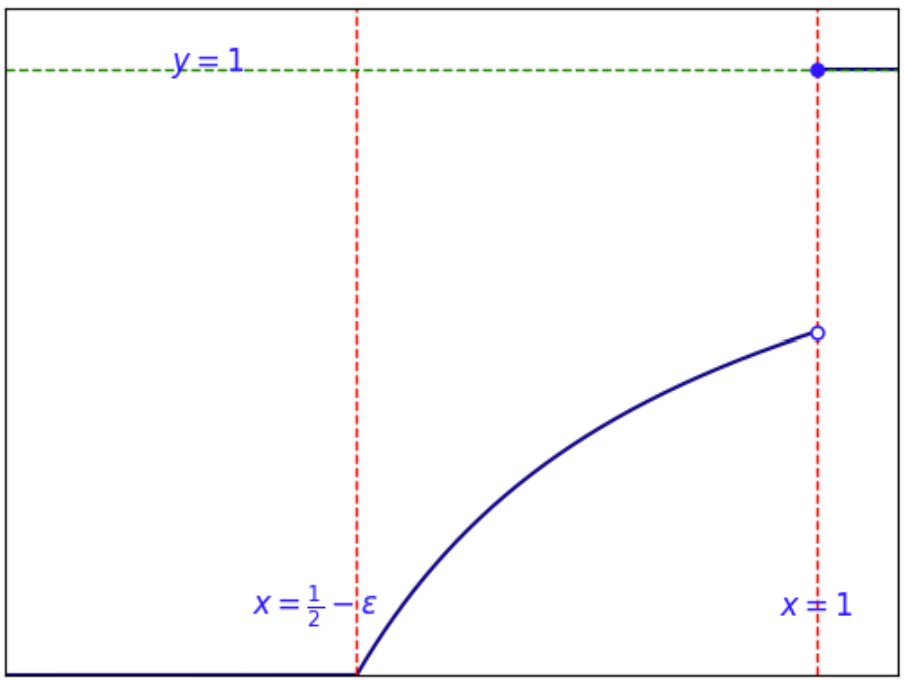}%
\label{Fig:51}}
\hfil
\centering
{\includegraphics[width=0.48\linewidth]{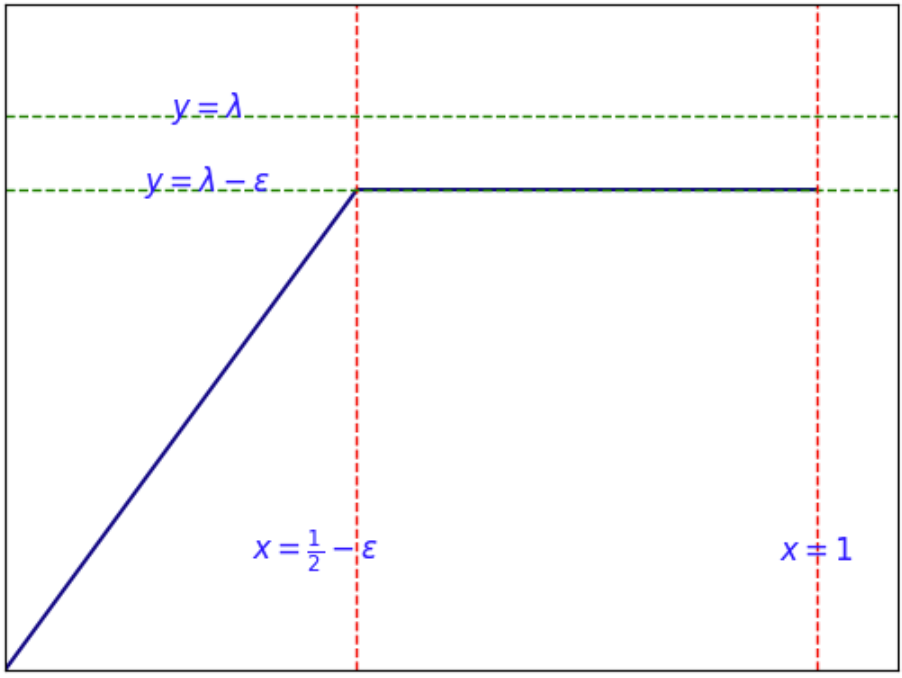}%
\label{Fig:57}}
\caption{Plot of the CDF, $F_{X_0}$, and the corresponding function $g_0$ \textbf{Left:} Plot of $F_{X_0}$. \textbf{Right:} Plot of $g_0$. }
\label{fig_sim_11}
\end{figure}

For $i \in [0:K]$, the Environment $i$ interacts with the rate selection policy and gives rise to a probability measure $\mathbb{P}^i$ over $\mathcal{H}(T)$. Let $\mathbb{E}^i$ denote the corresponding expectation. We have the following lemma. 

\begin{lemma}\label{lemma:l_b_lem1}
For each $k \in [1:K]$ and $t \in [0:T]$, we have that
\begin{align}
    \mathbb{E}^k\{N_k(t) \} \leq \mathbb{E}^0\{N_k(t) \}  + 4\sqrt{7}\varepsilon t  \sqrt{\mathbb{E}^0\{N_k(T) \}} \nonumber
\end{align}
\begin{proof}
 For two distributions in  $G_1$ and $G_2$ supported in $\mathcal{B}_T$, let $D_{\text{TV}}(G_1 \lVert G_2)$ denote their total variation distance. The result is trivial for $t = 0$, since $N_i(0) = 0$ for all $i \in [0:K]$. Hence, let us assume $t >0$. Since we assumed that the policy for selecting rates is deterministic,  $N_k(t)$ is $\mathcal{H}(T)$ measurable. Hence,
\begin{align}
     \mathbb{E}^k\{N_k(t)\}-\mathbb{E}^0\{N_k(t)\} & \leq \sum_{\vec{h} \in \mathcal{B}_T} N_k(t)(\vec{h})\left[\mathbb{P}^k(\vec{h})-\mathbb{P}^0(\vec{h})\right] \leq \sum_{\vec{h} \in \mathcal{B}_T}  N_k(t)(\vec{h})\left|\mathbb{P}^k(\vec{h})-\mathbb{P}^0(\vec{h})\right| \nonumber\\&\leq t\sum_{\vec{h} \in \mathcal{B}_T} \left|\mathbb{P}^k(\vec{h})-\mathbb{P}^0(\vec{h})\right| = 2tD_{\text{TV}}(\mathbb{P}^0\lVert \mathbb{P}^k)\leq t \sqrt{2D_{\text{KL}}(\mathbb{P}^0\lVert \mathbb{P}^k)} \nonumber
\end{align}
where the last inequality follows by Pinsker's inequality. 

Now, we prove that
\begin{align}\label{eqn:to_est}
    D_{\text{KL}}(\mathbb{P}^0\lVert \mathbb{P}^k) \leq 56\varepsilon^2\mathbb{E}^0\{N_k(T)\}
\end{align}
which establishes the result. 
Notice that
\begin{align}\label{eqn:orig_eq}
    D_{\text{KL}}(\mathbb{P}^0\lVert \mathbb{P}^k) &= \sum_{\tau =1 }^T D_{\text{KL}}(\mathbb{P}^0(\mathcal{H}(\tau)|\mathcal{H}(\tau-1))\lVert \mathbb{P}^k(\mathcal{H}(\tau)|\mathcal{H}(\tau-1)))\nonumber\\& = \sum_{\tau = 1}^T  D_{\text{KL}}(\mathbb{P}^0(B(\tau), A(\tau)|\mathcal{H}(\tau-1))\lVert \mathbb{P}^k(B(\tau),A(\tau)|\mathcal{H}(\tau-1)))
\end{align}
where the first equality follows by applying chain rule of KL divergence. For each $i \in [0:K]$, we  define the function $\tilde{F}^i(x) = \mathbb{P}\{X_k \geq x\}$. Hence, we have that
\begin{align}
    \tilde{F}^i(x) = \begin{cases}
        1 & \text{ if } x \leq \frac{1}{2} - \varepsilon\\
        \frac{\frac{1}{2}-\varepsilon}{x} & \text{ if }x \in \left(\frac{1}{2} - \varepsilon ,1\right]\setminus \mathcal{I}_i\\
         \frac{\frac{1}{2}-\varepsilon}{x_i}& \text{ if }x \in \mathcal{I}_i\\
         0 & \text{ if } x > 1
    \end{cases} \nonumber
\end{align}
for $i \in [1:K]$, and
\begin{align}
    \tilde{F}^0(x) = \begin{cases}
        1 & \text{ if } x \leq \frac{1}{2} - \varepsilon\\
        \frac{\frac{1}{2}-\varepsilon}{x} & \text{ if }x \in \left(\frac{1}{2} - \varepsilon ,1\right]\\
         0 & \text{ if } x > 1.
    \end{cases} \nonumber
\end{align}
Also, for $\vec{h} \in \mathcal{B}_{\tau-1}$, $a \in \{0,1\}$, $P^{i,\tau}_{a,\vec{h}}$ denotes the PMF of a Bernoulli$(\tilde{F}^i(f^{\tau}(\vec{h},a)))$ distribution (recall that $f^{1},f^2,\dots, f^{T}$ denotes the rate selection policy). First, notice that for each $\tau \in [1:T]$, $i \in [0:K]$, $a,b \in \{0,1\}$, and $\vec{h} \in \mathcal{B}_{\tau-1}$,
\begin{align}\label{eqn:imp_equation}
    \mathbb{P}^i(B(\tau) = b,A(\tau)& = a|\mathcal{H}(\tau-1) = \vec{h}) \nonumber\\& =\mathbb{P}^i(B(\tau) = b|A(\tau) = a,\mathcal{H}(\tau-1)= \vec{h})\mathbb{P}^i(A(\tau) = a|\mathcal{H}(\tau-1) = \vec{h})\nonumber\\& = P^{i,\tau}_{a,\vec{h}}(b)\mathbb{P}^0(A(\tau) = a), 
\end{align}
where the last equality follows since $A(\tau)$ is independent of $\mathcal{H}(\tau-1)$, and given $A(\tau) = a,\mathcal{H}(\tau-1)= \vec{h}$, the chosen rate is $f^{\tau}(\vec{h},a)$.

For $x,y \in [0,1]$, we use the notation $D_{\text{KL}}(x\lVert y)$ to denote the KL divergence between two Bernoulli$(x)$, and Bernoulli$(y)$ random variables. Hence, for each $\tau \in [1:T]$, we have
\begin{align}
     &D_{\text{KL}}(\mathbb{P}^0(B(\tau), A(\tau)|\mathcal{H}(\tau-1))\lVert \mathbb{P}^k(B(\tau), A(\tau)|\mathcal{H}(\tau-1))) \nonumber\\&  =  \sum_{\vec{h} \in \mathcal{B}_{\tau-1}}\mathbb{P}^0(\vec{h})\sum_{(a,b) \in \{0,1\}^2}\mathbb{P}^0(B(\tau) = b, A(\tau) = a| \mathcal{H}(\tau-1) = \vec{h})  \nonumber\\&\ \ \ \ \cdot\ln\left(\frac{\mathbb{P}^0(B(\tau) = b, A(\tau) = a| \mathcal{H}(\tau-1) = \vec{h})}{\mathbb{P}^k(B(\tau) = b, A(\tau) = a| \mathcal{H}(\tau-1) = \vec{h})}\right)\nonumber\\& =_{(a)}  
     \sum_{\vec{h} \in \mathcal{B}_{\tau-1}} \mathbb{P}^0(\vec{h})\sum_{(a,b) \in \{0,1\}^2} P^{0,\tau}_{a,\vec{h}}(b)\mathbb{P}^0(A(\tau) = a) \ln\left(\frac{P^{0,\tau}_{a,\vec{h}}(b)}{P^{k,\tau}_{a,\vec{h}}(b)}\right)\nonumber\\&=\sum_{\vec{h} \in \mathcal{B}_{\tau-1}} \mathbb{P}^0(\vec{h})\sum_{(a,b) \in \{0,1\}^2}\mathbbm{1}\{f^{\tau}(\vec{h},a) \in \mathcal{I}_k\}P^{0,\tau}_{a,\vec{h}}(b)\mathbb{P}^0(A(\tau) = a) \ln\left(\frac{P^{0,\tau}_{a,\vec{h}}(b)}{P^{k,\tau}_{a,\vec{h}}(b)}\right)\nonumber\\&\ \ \ \   + \sum_{\vec{h} \in \mathcal{B}_{\tau-1}} \mathbb{P}^0(\vec{h})\sum_{(a,b) \in \{0,1\}^2}\mathbbm{1}\{f^{\tau}(\vec{h},a) \not\in \mathcal{I}_k\}P^{0,\tau}_{a,\vec{h}}(b)\mathbb{P}^0(A(\tau) = a) \ln\left(\frac{P^{0,\tau}_{a,\vec{h}}(b)}{P^{k,\tau}_{a,\vec{h}}(b)}\right)\nonumber\\&=_{(b)}\sum_{\vec{h} \in \mathcal{B}_{\tau-1}} \mathbb{P}^0(\vec{h})\sum_{(a,b) \in \{0,1\}^2} \mathbbm{1}\{f^{\tau}(\vec{h},a) \in \mathcal{I}_k\}P^{0,\tau}_{a,\vec{h}}(b)\mathbb{P}^0(A(\tau) = a) \ln\left(\frac{P^{0,\tau}_{a,\vec{h}}(b)}{P^{k,\tau}_{a,\vec{h}}(b)}\right)\nonumber\\&=\sum_{\vec{h} \in \mathcal{B}_{\tau-1}} \mathbb{P}^0(\vec{h})\sum_{a \in \{0,1\}} \mathbbm{1}\{f^{\tau}(\vec{h},a) \in \mathcal{I}_k\}\mathbb{P}^0(A(\tau) = a)\sum_{b \in \{0,1\}}P^{0,\tau}_{a,\vec{h}}(b) \ln\left(\frac{P^{0,\tau}_{a,\vec{h}}(b)}{P^{k,\tau}_{a,\vec{h}}(b)}\right)\nonumber\\&=\sum_{\vec{h} \in \mathcal{B}_{\tau-1}} \mathbb{P}^0(\vec{h})\sum_{a \in \{0,1\}} \mathbbm{1}\{f^{\tau}(\vec{h},a) \in \mathcal{I}_k\}\mathbb{P}^0(A(\tau) = a)D_{\text{KL}}(P^{0,\tau}_{a,\vec{h}}\lVert P^{k,\tau}_{a,\vec{h}}) 
     \nonumber\\& \leq_{(c)} \sum_{\vec{h} \in \mathcal{B}_{\tau-1}}\mathbb{P}^0(\vec{h})\sum_{a \in \{0,1\}}\mathbbm{1}\{f^{\tau}(\vec{h},a) \in \mathcal{I}_k\}\mathbb{P}^0(A(\tau) = a)D_{\text{KL}}\left(\frac{1/2 - \varepsilon}{x_{k+1}}\Bigg{\lVert} \frac{1/2 -\varepsilon}{x_k}\right)  \nonumber\\&= \mathbb{P}^0(V(\tau) \in  \mathcal{I}_k)D_{\text{KL}}\left(\frac{1/2 - \varepsilon}{x_{k+1}}\Bigg{\lVert} \frac{1/2 -\varepsilon}{x_k}\right)  \nonumber
\end{align}
where (a) follows from \eqref{eqn:imp_equation}, (b) follows since for each $b \in \{0,1\}$,  $P^{i,\tau}_{a,\vec{h}}(b) = P^{0,\tau}_{a,\vec{h}}(b)$ whenever $f^{\tau}(\vec{h},a) \not\in \mathcal{I}_k$ (since $\tilde{F}^0$ and $\tilde{F}^k$ are the same outside of $\mathcal{I}_k$), and (c) follows from Lemma~\ref{lemma:bern_KL}-1, since given $f^{\tau}(\vec{h},a) \in \mathcal{I}_k$, we have
\begin{align}
     \tilde{F}^0(f^{\tau}(\vec{h},a)) \geq \frac{1/2 - \varepsilon}{x_{k+1}} \text{ and } \tilde{F}^k(f^{\tau}(\vec{h},a)) = \frac{1/2 - \varepsilon}{x_{k}}. \nonumber
\end{align}

Plugging the above back in \eqref{eqn:orig_eq}, we have that
\begin{align}
    &D_{\text{KL}}(\mathbb{P}^0\lVert \mathbb{P}^k) \leq \sum_{\tau = 1}^T \mathbb{P}^0(V(\tau) \in  \mathcal{I}_k)D_{\text{KL}}\left(\frac{1/2 - \varepsilon}{x_{k+1}}\Bigg{\lVert} \frac{1/2 -\varepsilon}{x_k}\right)  \nonumber\\&= \mathbb{E}^0\{N_k(T)\}D_{\text{KL}}\left(\frac{1/2 - \varepsilon}{x_{k+1}}\Bigg{\lVert} \frac{1/2 -\varepsilon}{x_{k}}\right) =_{(a)} \mathbb{E}^0\{N_k(T)\}D_{\text{KL}}\left(\frac{1/2 - \varepsilon}{x_{k+1}}\Bigg{\lVert} \frac{1/2 +\varepsilon}{x_{k+1}}\right)\nonumber\\&\leq _{(b)} \frac{4\varepsilon^2}{ \frac{1/2 +\varepsilon}{x_{k+1}}\left(1- \frac{1/2 +\varepsilon}{x_{k+1}}\right)}\mathbb{E}^0\{N_k(T)\}=\frac{4\varepsilon^2}{ \frac{1/2 +\varepsilon}{x_{k+1}}\left(1- \frac{1/2 -\varepsilon}{x_{k}}\right)}\mathbb{E}^0\{N_k(T)\} \nonumber\\&\leq_{(c)} \frac{4\varepsilon^2}{ (1/2 +\varepsilon)\left(1- \frac{1/2- \varepsilon}{7/12}\right)}\mathbb{E}^0\{N_k(T)\} \leq \frac{4\varepsilon^2}{ \frac{1}{2}\left(1- \frac{1/2 }{7/12}\right)}\mathbb{E}^0\{N_k(T)\} \nonumber\\&= 56 \varepsilon^2\mathbb{E}^0\{N_k(T)\} \nonumber
\end{align}
where (a) follows from the definition of $x_{k+1}$ in \eqref{eqn:x_def},  (b) follows from Lemma~\ref{lemma:bern_KL}-2, and (c) follows since $x_k,x_{k+1} \in (7/12,1)$ (Lemma~\ref{lemma:ref_def_lemma}-1) which establishes \eqref{eqn:to_est} as desired. Hence, we are done with the proof of the lemma.
\end{proof}
\end{lemma}
Next, we have the following lemma.
\begin{lemma}\label{lemma:service_bound}
Fix $k \in [1:K]$ and $t \in [1:T]$. We have
\begin{align}
    \sum_{\tau = 1}^t \mathbb{E}^k\{g_k(V(\tau))\} \leq \left(\frac{1}{2}-\varepsilon \right)t + 2\varepsilon \mathbb{E}^k\{N_k(t)\} \nonumber
\end{align}
where function $g_k$ is defined in \eqref{eqn:g_k_def}.
\begin{proof}
Notice that
\begin{align}\label{eqn:0101010}
    &\sum_{\tau = 1}^t \mathbb{E}^k\{g_k(V(\tau))\} \nonumber\\&= \sum_{\tau = 1}^t \left[\mathbb{E}^k\{g_k(V(\tau))|V(\tau) \in \mathcal{I}_k\}\mathbb{P}^k\{V(\tau) \in \mathcal{I}_k\}+\mathbb{E}^k\{g_k(V(\tau))|V(\tau) \not\in \mathcal{I}_k\}\mathbb{P}^k\{V(\tau) \not\in \mathcal{I}_k\} \right]\nonumber\\& \leq \sum_{\tau = 1}^t\left[ \left(\frac{1}{2}+\varepsilon \right)\mathbb{P}^k\{V(\tau) \in \mathcal{I}_k\}+\left(\frac{1}{2}-\varepsilon \right)\mathbb{P}^k\{V(\tau) \not\in \mathcal{I}_k\} \right]\nonumber\\& = \sum_{\tau = 1}^t\left[ \left(\frac{1}{2}+\varepsilon \right)\mathbb{P}^k\{V(\tau) \in \mathcal{I}_k\}+\left(\frac{1}{2}-\varepsilon \right)\mathbb{P}^k\{V(\tau) \not\in \mathcal{I}_k\} \right] \nonumber\\& = \left(\frac{1}{2}+\varepsilon \right)\mathbb{E}\{N_k(t)\}+\left(\frac{1}{2}-\varepsilon \right)\mathbb{E}\{t-N_k(t)\} = \left(\frac{1}{2}-\varepsilon \right)t + 2\varepsilon \mathbb{E}^k\{N_k(t)\}
\end{align}
Hence, we are done.
\end{proof}
\end{lemma}

Fix $t \in [1:T]$. Consider $\tau \in [2:t]$. From the queueing equation,
\begin{align}
    Q(\tau) \geq Q(\tau-1) + A(\tau-1) - V(\tau - 1)B(\tau-1). \nonumber
\end{align}
Consider $k \in [1:K]$. Taking expectations in Environment $k$ and summing the above for $\tau \in [2:t]$, we have
\begin{align}
   & \mathbb{E}^k\{Q(t) \} \nonumber\\&\geq \frac{t-1}{2} - \sum_{\tau = 1}^{t-1}\mathbb{E}^k\{g_k(V(\tau))\}  \geq_{(a)} \frac{t-1}{2} - \left(\frac{1}{2}-\varepsilon \right)(t-1) - 2\varepsilon \mathbb{E}^k\{N_k(t-1)\} \nonumber\\&= \varepsilon (t-1)-2\varepsilon \mathbb{E}^k\{N_k(t-1)\}\geq_{(b)} \varepsilon (t-1)-2\varepsilon \mathbb{E}^0\{N_k(t-1) \}  - 8\sqrt{7}\varepsilon^2 (t-1)  \sqrt{\mathbb{E}^0\{N_k(T) \}} \nonumber
\end{align}
where for (a) we have used Lemma~\ref{lemma:service_bound}, for (b) we have used Lemma~\ref{lemma:l_b_lem1}. Now, we sum the above over $[1:K]$ to get,
\begin{align}
   &\sum_{k \in [1:K]} \mathbb{E}^k\{Q(t) \} \geq  \varepsilon (t-1)K-2\varepsilon  \sum_{k \in [1:K]}\mathbb{E}^0\{N_k(t-1) \}  - 8\sqrt{7}\varepsilon^2 (t-1)  \sum_{k \in [1:K]} \sqrt{\mathbb{E}^0\{N_k(t-1) \}} \nonumber\\&\geq_{(a)}  \varepsilon (t-1)K-2\varepsilon (t-1) - 8\sqrt{7}\varepsilon^2 (t-1)  \sum_{k \in [1:K]} \sqrt{\mathbb{E}^0\{N_k(T) \}} \nonumber\\&\geq_{(b)}  \varepsilon (t-1)K-2\varepsilon (t-1) - 8\sqrt{7}\varepsilon^2 (t-1)   \sqrt{K\sum_{k \in [1:K]}\mathbb{E}^0\{N_k(T) \}} \nonumber\\&\geq  \varepsilon (t-1)K-2\varepsilon (t-1) - 8\sqrt{7}\varepsilon^2 (t-1)   \sqrt{KT} \nonumber
\end{align}
where (a) follows since $\sum_{k \in [1:K]} N_k(t-1) \leq t-1$, (b)  follows from Cauchy-Schwarz inequality, and the last inequality follows since  $\sum_{k \in [1:K]} N_k(T) \leq T$. Summing the above for $t\in [1:T]$, 
\begin{align}
  &\sum_{t = 1}^T \sum_{k \in [1:K]} \mathbb{E}^k\{Q(t) \} \nonumber\\&\geq   \frac{\varepsilon (T-1)TK}{2}-\varepsilon T(T-1) - 4\sqrt{7}\varepsilon^2 (T-1)   \sqrt{KT^3}  \nonumber\\&\geq_{(a)}  \frac{\varepsilon (T-1)TK}{2}-\frac{K\varepsilon T(T-1)}{5} - 4\sqrt{7}\varepsilon^2 T  \sqrt{KT^3}\geq  \frac{3\varepsilon KT(T-1) }{10}- 4\sqrt{7}\varepsilon^2 T  \sqrt{KT^3}. \nonumber
\end{align}
where (a) follows by $K \geq 5$ (Lemma~\ref{lemma:ref_def_lemma}-2).
Hence, notice that
\begin{align}\label{eqn:one_bef_last}
\frac{1}{K} \sum_{k \in [1:K]} \frac{1}{T} \sum_{t = 1}^T  \mathbb{E}^k\{Q(t) \} & \geq  \frac{3\varepsilon (T-1)}{10} - 4\varepsilon^2\sqrt{\frac{7T^3}{K}}  \geq_{(a)} \frac{3\varepsilon (T-1)}{10} - 24\sqrt{7}\varepsilon^{2.5}T^{1.5} 
\end{align}
where (a) follows since $K \geq 1/(36\varepsilon)$ (Lemma~\ref{lemma:ref_def_lemma}-2). 

Now we set $T$. In particular, let
\begin{align}\label{eqn:keyyy}
    T = \left\lceil \left(\frac{1}{160\sqrt{7}\varepsilon^{1.5}} \right)^2 + 1\right\rceil
\end{align}
Notice that
\begin{align}
    T \leq  \left(\frac{1}{160\sqrt{7}\varepsilon^{1.5}} \right)^2 + 2\leq  \left(1+\frac{1}{8}\right)\left(\frac{1}{160\sqrt{7}\varepsilon^{1.5}} \right)^2 \nonumber
\end{align}
where the second inequality follows since
\begin{align}
  \frac{1}{8}\left(\frac{1}{160\sqrt{7}\varepsilon^{1.5}} \right)^2 \geq \frac{1}{8}\left(\frac{144^{1.5}}{160\sqrt{7}} \right)^2 \geq 2 \nonumber
\end{align}
Hence,
\begin{align}
     24\sqrt{7}\varepsilon^{2.5}T^{1.5}  \leq 24\sqrt{7}\varepsilon^{2.5}\left( 1+\frac{1}{8}\right)^{1.5}\left(\frac{1}{160\sqrt{7}\varepsilon^{1.5}} \right)^3 = \frac{24}{160}\left( 1+\frac{1}{8}\right)^{1.5}\left(\frac{1}{160\sqrt{7} \varepsilon}\right)^2 \nonumber
\end{align}
Similarly, from \eqref{eqn:keyyy}, we have
\begin{align}
   \frac{3\varepsilon (T-1)}{10} \geq \frac{3\varepsilon}{10}\left(\frac{1}{160\sqrt{7} \varepsilon^{1.5}}\right)^2 = \frac{3}{10}\left(\frac{1}{160\sqrt{7} \varepsilon}\right)^2 \nonumber
\end{align}
Using the above in \eqref{eqn:one_bef_last}, we have
\begin{align}
\frac{1}{K} \sum_{k \in [1:K]} \frac{1}{T} \sum_{t = 1}^T  \mathbb{E}^k\{Q(t) \} & \geq\left( \frac{3}{10} -  \frac{24}{160}\left(1+ \frac{1}{8}\right)^{1.5}\right)\left(\frac{1}{160\sqrt{7} \varepsilon}\right)^2 \geq \frac{6 \times 10^{-7}}{\varepsilon^2} \nonumber
\end{align}
Hence, for at least one of the environments $k^{'}$ in $[1:K]$, we have
\begin{align}
 \frac{1}{T} \sum_{t = 1}^T  \mathbb{E}^{k^{'}}\{Q(t) \} &  \geq \frac{6 \times 10^{-7}}{\varepsilon^2} \nonumber
\end{align}
as desired.

% \noindent
% Figure~\ref{Fig:22} illustrates the performance of Algorithm~\ref{algo:UCB_mesh} in the first environment described above for $\varepsilon = 1/144$.
% \begin{figure}
% \centering
% \includegraphics[width=0.6\linewidth]{Figures/plot_of_F1.png}%
% \caption{Plot of $\frac{1}{t}\sum_{\tau=1}^t Q(\tau)$ vs. $t$ for first environment}\label{Fig:22}
% \Description{Plot shows a simulations of the time-average expected queue size vs. the time-horizon.}
% \end{figure}
\section{Proofs and Lemmas for Section~\ref{sec:known}}

\subsection{Proof of Lemma~\ref{thm:log(H)}: Stage I}
\label{sec:stageI}
Consider the auxiliary queue \(\widetilde{Q}(t)\) with \(\widetilde{Q}(1)=0\) evolving as
\[
\widetilde{Q}(t+1)
= \bigl[\,\widetilde{Q}(t) + A(t) - \widetilde{S}_{K(t)}(t)\,\bigr]_+.
\]
We define the auxiliary service as
\[
\widetilde{S}_k(t) \;\triangleq\; S_k(t) - \mu_k + \widetilde{\mu},
\]
where \(\widetilde{\mu}\in[\lambda,\mu^*]\) will be specified later. Note that \(\mathbb{E}[\widetilde{S}_k(t)] = \widetilde{\mu}\) for all \(k\), and \(\widetilde{S}_k(t)\in[-1,2]\) in general.

\begin{definition}
Define for each arm \(k\in\mathcal{K}\) the sub-optimality gap
\[
\Delta_k = \mu^* - \mu_k,
\]
and
\[
\tilde{\Delta} = \mu^* - \tilde{\mu}.
\]
\end{definition}

\begin{lemma}\label{lem:Q2sum}
For any \(t\ge2\),
\[
Q(t) - \widetilde{Q}(t)
\;\le\;
\sum_{\tau=1}^{t-1} \bigl[\widetilde{\mu} - \mu_{K(\tau)}\bigr]_+.
\]
\end{lemma}

\begin{proof}[Proof of Lemma~\ref{lem:Q2sum}]
For any \(t \in \{1,2,\ldots\}\), define
\[
\mathrm{Emp}(t) = \max\{\tau \le t : Q(\tau) = 0\}.
\]
Since \(Q(1)=0\), the definition of \(\mathrm{Emp}(t)\) is valid for all \(t\ge1\). We consider two cases.

\textbf{Case 1:} If \(Q(t)=0\), then \(\widetilde{Q}(t)\ge 0\) implies
\[
Q(t)-\widetilde{Q}(t)\le 0 \le \sum_{\tau=1}^{t-1} \left[\widetilde{\mu} - \mu_{K(\tau)}\right]_+.
\]

\textbf{Case 2:} If \(Q(t)>0\), then by definition \(\mathrm{Emp}(t) < t\) and
\[
Q(\tau+1) > 0 \quad \text{for every } \tau\in\{\mathrm{Emp}(t),\ldots,t-1\}.
\]
Under this condition, the queue evolution
simplifies to
\[
Q(\tau+1) = Q(\tau) - S_{K(\tau)}(\tau) + A(\tau).
\]
Subtracting the auxiliary queue evolution gives
\begin{align*}
Q(\tau+1)-\widetilde{Q}(\tau+1)
&=\left(Q(\tau)+A(\tau)-S_{K(\tau)}(\tau)\right) - \left[\widetilde{Q}(\tau)+A(\tau)-\widetilde{S}_{K(\tau)}(\tau)\right]_+\\[1mm]
&\le Q(\tau)+A(\tau)-S_{K(\tau)}(\tau) - \left(\widetilde{Q}(\tau)+A(\tau)-\widetilde{S}_{K(\tau)}(\tau)\right)\\[1mm]
&=\left(Q(\tau)-\widetilde{Q}(\tau)\right) + \left(\widetilde{S}_{K(\tau)}(\tau)-S_{K(\tau)}(\tau)\right)\\[1mm]
&=\left(Q(\tau)-\widetilde{Q}(\tau)\right) + \left(\widetilde{\mu}-\mu_{K(\tau)}\right).
\end{align*}
Summing from \(\tau=\mathrm{Emp}(t)\) to \(\tau=t-1\) yields
\[
\left(Q(t)-Q(\mathrm{Emp}(t))\right)-\left(\widetilde{Q}(t)-\widetilde{Q}(\mathrm{Emp}(t))\right)
\le \sum_{\tau=\mathrm{Emp}(t)}^{t-1}\left(\widetilde{\mu}-\mu_{K(\tau)}\right).
\]
Since \(Q(\mathrm{Emp}(t))=0\) and \(\widetilde{Q}(\mathrm{Emp}(t))\ge0\), we conclude
\[
Q(t)-\widetilde{Q}(t)
\le \sum_{\tau=\mathrm{Emp}(t)}^{t-1}\left[\widetilde{\mu}-\mu_{K(\tau)}\right]_+
\le \sum_{\tau=1}^{t-1}\left[\widetilde{\mu}-\mu_{K(\tau)}\right]_+.
\]
\end{proof}
Define $N_k(t)$ as the number of times arm $k$ is chosen up to (and including) time step $t$. In particular,
\begin{align}\label{eqn:n_k_t_known}
    N_k(t) = \sum_{\tau = 1}^t \mathbbm{1}\{K(\tau) = k\}.
\end{align}

\begin{lemma}\label{lem:sum2Delta}
For any  \(t\ge 1\), we have
\[
\sum_{\tau=1}^t \left[ \tilde{\mu} - \mu_{K(\tau)}  \right]_+ 
\le
\sum_{\substack{k\in\mathcal{K}:\\ \Delta_k \ge \tilde{\Delta}}} \Delta_k\, N_k(t).
\]
\end{lemma}
\begin{proof}[Proof of Lemma~\ref{lem:sum2Delta}]
Decompose the sum over time into a sum over arms:
\begin{align*}
\sum_{\tau=1}^t \left[ \tilde{\mu} - \mu_{K(\tau)}  \right]_+ 
&=\sum_{\tau=1}^t \sum_{k\in\mathcal{K}} \left[\tilde{\mu} - \mu_k \right]_+ \, \mathbbm{1}\{K(\tau)=k\} \\
&=\sum_{k\in\mathcal{K}} \left[ \tilde{\mu} - \mu_k \right]_+ \, N_k(t).
\end{align*}
If \(\widetilde{\mu}<\mu_k\), the term vanishes.  Otherwise \([\widetilde{\mu}-\mu_k]_+ = \widetilde{\mu}-\mu_k \le \mu^*-\mu_k=\Delta_k\), and \(\widetilde{\mu}\ge\mu_k\) implies \(\Delta_k\ge\tilde\Delta\).  The claimed bound follows.
\end{proof}

\begin{lemma}\label{lem:bounding-Q-tilde}
For any \(H \ge 1\),
\[
\frac{1}{H} \sum_{t=1}^{H} \EE\bigl[\widetilde{Q}(t)\bigr]
\;\le\;
\frac{2}{\vedisc - \tilde{\Delta}}.
\]
\end{lemma}
\begin{proof}[Proof of Lemma~\ref{lem:bounding-Q-tilde}]
Using the standard Lyapunov drift bound for the auxiliary queue,
\begin{align*}
\EE\bigl[\widetilde{Q}(t+1)^2 - \widetilde{Q}(t)^2\bigr]
&\le
\EE\Bigl[(A(t) - \widetilde{S}_{K(t)}(t))^2\Bigr]
+2\,\EE\Bigl[(A(t)-\widetilde{S}_{K(t)}(t))\,\widetilde{Q}(t)\Bigr]\\
&\le
4 + 2\,\EE\bigl[A(t)-\widetilde{S}_{K(t)}(t)\bigr]\,\EE\bigl[\widetilde{Q}(t)\bigr]\\
&=
4 + 2\bigl(\lambda - \widetilde{\mu}\bigr)\,\EE\bigl[\widetilde{Q}(t)\bigr]\\
&=
4 - 2(\vedisc - \tilde{\Delta})\,\EE\bigl[\widetilde{Q}(t)\bigr],
\end{align*}
where we used \(\EE[(A(t)-\widetilde{S}_{K(t)}(t))^2]\le4\) and \(\lambda-\widetilde{\mu}=-\bigl(\vedisc-\tilde\Delta\bigr)\).

Summing over \(t=1,\dots,H\) and noting \(\widetilde{Q}(1)=0\) gives
\[
\EE\bigl[\widetilde{Q}(H+1)^2\bigr]
\;\le\;
4n - 2(\vedisc - \tilde{\Delta})\sum_{t=1}^{H}\EE\bigl[\widetilde{Q}(t)\bigr].
\]
Since the left side is nonnegative,
\[
2(\vedisc - \tilde{\Delta})\sum_{t=1}^{H}\EE\bigl[\widetilde{Q}(t)\bigr]
\;\le\;
4n,
\]
and dividing by \(2n(\vedisc - \tilde{\Delta})\) yields the stated bound.
\end{proof}

\begin{lemma}\label{Lemma:initialbound-main-lemma}
For any integer \(H\ge2\),
\[
\frac{1}{H}\sum_{t=1}^H \mathbb{E}\bigl[Q(t)\bigr]
\;\le\;
\frac{2}{\vedisc - \tilde{\Delta}}
\;+\;
\frac{1}{H}
\sum_{\substack{k\in\mathcal{K}\\\Delta_k \ge \tilde{\Delta}}}
\Delta_k
\sum_{t=1}^{H-1}\mathbb{E}\bigl[N_k(t)\bigr].
\]
\end{lemma}

\begin{proof}[Proof of Lemma~\ref{Lemma:initialbound-main-lemma}]
By Lemmas~\ref{lem:Q2sum} and \ref{lem:sum2Delta}, for each \(t\ge2\),
\[
Q(t) - \widetilde{Q}(t)
\;\le\;
\sum_{\tau=1}^{t-1}\bigl[\widetilde{\mu}-\mu_{K(\tau)}\bigr]_+
\;\le\;
\sum_{\substack{k\in\mathcal{K}\\\Delta_k \ge \tilde{\Delta}}}
\Delta_k\,N_k(t-1).
\]
Summing over \(t=2,\dots,H\) (and using \(Q(1)=\widetilde{Q}(1)=0\)) gives
\[
\sum_{t=1}^H\bigl[Q(t)-\widetilde{Q}(t)\bigr]
\;\le\;
\sum_{\substack{k\in\mathcal{K}\\\Delta_k \ge \tilde{\Delta}}}
\Delta_k
\sum_{t=2}^H N_k(t-1)
=
\sum_{\substack{k\in\mathcal{K}\\\Delta_k \ge \tilde{\Delta}}}
\Delta_k
\sum_{t=1}^{H-1} N_k(t).
\]
Dividing by \(H\) and taking expectations yields
\[
\frac{1}{H}\sum_{t=1}^H \mathbb{E}\bigl[Q(t)-\widetilde{Q}(t)\bigr]
\;\le\;
\frac{1}{H}
\sum_{\substack{k\in\mathcal{K}\\\Delta_k \ge \tilde{\Delta}}}
\Delta_k
\sum_{t=1}^{H-1}\mathbb{E}\bigl[N_k(t)\bigr].
\]
Combining this with Lemma~\ref{lem:bounding-Q-tilde} gives the desired result:
\[
\frac{1}{H}\sum_{t=1}^H \mathbb{E}\bigl[Q(t)\bigr]
\;\le\;
\frac{2}{\vedisc - \tilde{\Delta}}
\;+\;
\frac{1}{H}
\sum_{\substack{k\in\mathcal{K}\\\Delta_k \ge \tilde{\Delta}}}
\Delta_k
\sum_{t=1}^{H-1}\mathbb{E}\bigl[N_k(t)\bigr].
\]
\end{proof}

Lemma~\ref{Lemma:initialbound-main-lemma} applies to any stochastic bandit algorithm.  Here, similar to \cite{Kleinberg2003} we specialize to a UCB1  from \cite{Auer2002FinitetimeAO}.  Define
At each slot \(t\ge{d}+1\), the controller selects
\[
K(t)
\;=\;
\arg\max_{k\in\mathcal{K}}
\Bigl\{
\bar\mu_k(t-1)
\;+\;
\sqrt{\frac{2\log t}{N_k(t-1)}}
\Bigr\},
\]
where \(\bar\mu_k(t-1)\) is the empirical mean reward of arm \(k\) up to (and including)
time step \(t-1\).

\begin{algorithm}[ht]
\caption{UCB1}
\label{alg:ucb-queue}
\SetKwInOut{Input}{Input}
\Input{Number of arms ${d}$}
\For{$t \gets 1$ \KwTo ${d}$}{
  Pull arm $t$: $K(t) \gets t$\;
}
\For{$t \gets {d}+1$ \KwTo $\infty$}{
  \For{$k \gets 1$ \KwTo ${d}$}{
    $U_k(t-1) \gets \bar{\mu}_k(t-1) + \sqrt{\frac{2\log t}{N_k(t-1)}}$\;
  }
  $K(t) \gets \arg\max_{k\in\mathcal{K}} U_k(t-1)$\;
}
\end{algorithm}

\begin{lemma}
\label{lem:UCB1}[Lemma~1 from \cite{10.1007/978-3-540-72927-3_33} see also Theorem~1 from \cite{Auer2002FinitetimeAO}]
For any ${d} > 1$, consider running UCB1 on a set of ${d}$ actions with arbitrary unknown reward distributions supported on $[0,1]$.  For any suboptimal arm \(k\in\mathcal{K}\) with gap \(\Delta_k=\mu^*-\mu_k\), and any \(H\ge1\), 
\[
\EE\bigl[N_k(H)\bigr] \leq 
8  \frac{\log H}{\Delta_k^2}
+ \Bigg(1 + \frac{\pi^2}{3}\Bigg).
\]
\end{lemma}

\begin{remark}
The model in~\cite{Auer2002FinitetimeAO} formally assumes that rewards across different actions are independent. 
This assumption does not hold in our setting. 
However, as mention in \cite{Kleinberg2003} since their proof of Lemma~\ref{lem:UCB1} does not actually rely on independence, 
we are still justified in applying the result.
\end{remark}
\begin{proof}[Proof of Lemma~\ref{thm:log(H)}]
By Lemma~\ref{lem:UCB1}, for each arm \(k\) with gap \(\Delta_k\ge\tilde\Delta\),
\begin{align*}
\Delta_k\sum_{t=1}^{H-1}\EE[N_k(t)]
&\le
\Delta_k\sum_{t=1}^{H-1}\Bigl(8  \frac{\log t}{\Delta_k^2}
+ 1 + \frac{\pi^2}{3}\Bigg)\Bigr)
=
\sum_{t=1}^{H-1}\Bigl(8  \frac{\log t}{\Delta_k}
+ (1 + \frac{\pi^2}{3})\Delta_k\Bigg)\\
&\le
\sum_{t=1}^{H-1}\Bigl((1 + \frac{\pi^2}{3}) + 8\frac{\log t}{\tilde\Delta}\Bigr)
\;\le\;
\Bigl[(H-1)(1 + \frac{\pi^2}{3}) + \tfrac{8}{\tilde\Delta}\sum_{t=1}^{H-1}\log t\Bigr]\\
&\le
\Bigl[(H-1)(1 + \frac{\pi^2}{3}) + \tfrac{8}{\tilde\Delta}(H(\log H - 1) + 1)\Bigr]\\
&=
\Bigl[(H-1)(1 + \frac{\pi^2}{3} - \frac{8}{\tilde\Delta}) + \tfrac{8}{\tilde\Delta}H\log H\Bigr]\\
&\leq
\tfrac{8}{\tilde\Delta}H\log H,
\end{align*}
where
\[
\sum_{t=1}^{H-1}\log t
\;\le\;
\int_{1}^{H}\log x\,dx
=
H\log H - H + 1
=
H(\log H - 1) + 1,
\]
we get
\[
\frac{1}{H}\sum_{\Delta_k\ge\tilde\Delta}\Delta_k\sum_{t=1}^{H-1}\EE[N_k(t)]
\;\le\;
\tfrac{8\,{d}\,\log H}{\tilde\Delta}.
\]
Substituting into Lemma~\ref{Lemma:initialbound-main-lemma} yields
\[
\frac{1}{H}\sum_{t=1}^H\EE\bigl[Q(t)\bigr]
\;\le\;
\frac{2}{\vedisc-\tilde\Delta}
\;+\;
\frac{8\,{d}\,\log H}{\tilde\Delta}.
\]
Taking the infimum over \(\tilde\Delta\in(0,\vedisc)\) gives the first claim.  Finally, setting
\[
\tilde\Delta_{\mathrm{opt}}
=
\vedisc\,
\frac{\sqrt{8\,{d}\,\log H}}{\sqrt{2} + \sqrt{8\,{d}\,\log H}}
=
\vedisc\,
\frac{2\sqrt{{d}\,\log H}}{1 + 2\sqrt{{d}\,\log H}}
\]
and plugging back yields
\[
\frac{1}{H}\sum_{t=1}^H\EE\bigl[Q(t)\bigr]
\;\le\;
\frac{2}{\vedisc}\Bigl(1+2\sqrt{\,{d}\,\log H}\Bigr)^2
\]
completing the proof.

\end{proof}
\subsection{Proof of Lemma~\ref{thm:stage2}: Stage II}
\label{sec:stageII}

\begin{lemma}\label{lem:drift-t<kappa}
For each slot \(t\) with \(1 \le t \le {d}\), the drift satisfies
\[
\frac{1}{2}\EE\bigl[Q(t+1)^2\bigr]
- \frac{1}{2}\EE\bigl[Q(t)^2\bigr]
\;\le\;
\frac{1}{2}
\;-\;
\vedisc\,\EE\bigl[Q(t)\bigr]
\;+\;
\EE\bigl[Q(t)\bigl(\mu^*-\mu_t\bigr)\bigr].
\]
\end{lemma}

\begin{proof}[Proof of Lemma~\ref{lem:drift-t<kappa}]
By the queue update,
\begin{align*}
Q(t+1)^2
&\le Q(t)^2 + \left(A(t)-S_{K(t)}(t)\right)^2
+ 2\,Q(t)\,\left(A(t)-S_{K(t)}(t)\right)\\
&\le Q(t)^2 + 1 + 2\,Q(t)\,\left(A(t)-S_{K(t)}(t)\right),
\end{align*}
since \(\bigl(A(t)-S_{K(t)}(t)\bigr)^2\le1\).  Taking expectations and dividing by 2 gives
\[
\frac12\EE\bigl[Q(t+1)^2\bigr]
-\frac12\EE\bigl[Q(t)^2\bigr]
\le \frac12 + \EE\bigl[Q(t)\,(\lambda-\mu_{K(t)})\bigr].
\]
For \(t\le{d}\), the algorithm sets \(K(t)=t\), and \(\lambda=\mu^*-\vedisc\), so
\[
\EE\bigl[Q(t)\,(\lambda-\mu_t)\bigr]
= \EE\bigl[Q(t)\,(\mu^*-\vedisc-\mu_t)\bigr]
= -\vedisc\,\EE[Q(t)] + \EE\bigl[Q(t)\,(\mu^*-\mu_t)\bigr].
\]
Substituting yields the claimed bound.
\end{proof}

\begin{lemma}\label{lem:drift-t>kappa}
For every slot \(t\ge{d}+1\), the drift satisfies
\[
\frac{1}{2}\EE\bigl[Q(t+1)^2\bigr]
- \frac{1}{2}\EE\bigl[Q(t)^2\bigr]
\;\le\;
\frac{1}{2}
\;-\;
\vedisc\,\EE\bigl[Q(t)\bigr]
\;+\;
4{d} t^{-2}
\;+\;
2\,\EE\!\Bigl[\,Q(t)\,\sqrt{\tfrac{2\log t}{\,N_{K(t)}(t-1)\,}}\Bigr].
\]
\end{lemma}

\begin{proof}[Proof of Lemma~\ref{lem:drift-t>kappa}]
From the queue dynamics, for any \(t\ge1\),
\begin{equation}\label{eq:queue_dynamics}
\begin{aligned}
Q(t+1)^2
&\le Q(t)^2 + \left(A(t)-S_{K(t)}(t)\right)^2
+ 2\,Q(t)\,\left(A(t)-S_{K(t)}(t)\right)\\
&\le Q(t)^2 + 1 + 2\,Q(t)\,\left(A(t)-S_{K(t)}(t)\right),
\end{aligned}
\end{equation}
since \(\left(A(t)-S_{K(t)}(t)\right)^2\le1\).  Taking expectations for \(t\ge{d}+1\) gives
\begin{equation}\label{eq:drift}
\begin{aligned}
&\frac{1}{2}\,\mathbb{E}\left[Q(t+1)^2\right]
- \frac{1}{2}\,\mathbb{E}\left[Q(t)^2\right]
\le \frac{1}{2}
+ \mathbb{E}\left[\,Q(t)\,(\lambda-\mu_{K(t)})\right]\\
&= \frac{1}{2}
+ \mathbb{E}\left[Q(t)\,(\lambda-\mu_{K(t)})\,\mathbbm{1}\{\mathcal{G}(t-1)\}\right]
+ \mathbb{E}\left[Q(t)\,(\lambda-\mu_{K(t)})\,\mathbbm{1}\{\mathcal{G}(t-1)^c\}\right],
\end{aligned}
\end{equation}
where the “good” event is
\[
\mathcal{G}(t-1)
=\left\{\mu_k \in [\,U_k(t-1)-2\sqrt{\tfrac{2\log\!t}{N_k(t-1)}},\,U_k(t-1)\,]\ \forall\,k\in\mathcal{K}\right\}.
\]
On \(\mathcal{G}(t-1)\) we have
\(\mu_{K(t)}\ge U_{K(t)}(t-1)-2\sqrt{\tfrac{2\log\!t}{N_{K(t)}(t-1)}}\),
so
\[
\lambda-\mu_{K(t)}
\le \left(\lambda - U_{K(t)}(t-1)\right)
+2\sqrt{\frac{2\log\!t}{N_{K(t)}(t-1)}}.
\]
By the selection rule,
\(U_{K(t)}(t-1)\ge U_{k^*}(t-1)\) and by the \emph{good event},
\( U_{k^*}(t-1)\ge\mu_{k^*}\)
and \(\vedisc=\mu_{k^*}-\lambda\).  Hence
\[
\mathbb{E}\left[Q(t)\,(\lambda-\mu_{K(t)})\,\mathbbm{1}\{\mathcal{G}(t-1)\}\right]
\le \mathbb{E}\!\left[Q(t)\left(-\vedisc+2\sqrt{\tfrac{2\log\!t}{N_{K(t)}(t-1)}}\right)\mathbbm{1}\{\mathcal{G}(t-1)\}\right].
\]
Expanding gives
\[
\mathbb{E}\left[Q(t)\,(\lambda-\mu_{K(t)})\,\mathbbm{1}\{\mathcal{G}(t-1)\}\right]
=-\vedisc\,\mathbb{E}[Q(t)]
+\vedisc\,\mathbb{E}\!\left[Q(t)\,\mathbbm{1}\{\mathcal{G}(t-1)^c\}\right]
+2\,\mathbb{E}\!\left[Q(t)\sqrt{\tfrac{2\log\!t}{N_{K(t)}(t-1)}}\right].
\]
Substituting into \eqref{eq:drift} and using
\(\lambda-\mu_{K(t)}\le1\) yields
\begin{align}
\frac{1}{2}\,\mathbb{E}\left[Q(t+1)^2\right]
- \frac{1}{2}\,\mathbb{E}\left[Q(t)^2\right]
\le &\frac{1}{2}
-\vedisc\,\mathbb{E}[Q(t)]
+(1+\vedisc)\,\mathbb{E}\!\left[Q(t)\,\mathbbm{1}\{\mathcal{G}(t-1)^c\}\right]
\nonumber\\&+2\,\mathbb{E}\!\left[Q(t)\sqrt{\tfrac{2\log\!t}{N_{K(t)}(t-1)}}\right]. \nonumber
\end{align}
Using Lemma~\ref{eqn:det_queue_bound} and \(1+\vedisc\le 2\) gives
\[
\frac{1}{2}\,\mathbb{E}\left[Q(t+1)^2\right]
- \frac{1}{2}\,\mathbb{E}\left[Q(t)^2\right]
\le \frac{1}{2}
-\vedisc\,\mathbb{E}[Q(t)]
+2 t\,\mathbb{E}\!\left[\mathbbm{1}\{\mathcal{G}(t-1)^c\}\right]
+2\,\mathbb{E}\!\left[Q(t)\sqrt{\tfrac{2\log\!t}{N_{K(t)}(t-1)}}\right].
\]
% To Bound $\mathbb{E}\!\left[\mathbbm{1}\{\mathcal{G}(t-1)\}\right] = \mathbb{P}\{\mathcal{G}(t-1)^c\}$ we first have for a fixed $k\in\mathcal{K}$
% \begin{align*}
%     \mathbb{P}\{|\mu_k-\bar\mu_k|\geq\sqrt{\frac{2\log t}{N_k(t-1)}}\}
%     =
%     \sum_{n=1}^{t-1}
%     \mathbb{P}\{|\mu_k-\bar\mu_k|\geq\sqrt{\frac{2\log t}{n}} \mid N_k(t-1) = n\} 
%     \mathbb{P}\{N_k(t-1) = n\}
%     \leq
%     \sum_{n=1}^{t-1}
%     \mathbb{P}\{|\mu_k-\bar\mu_k|\geq\sqrt{\frac{2\log t}{n}} \mid N_k(t-1) = n\} 
% \end{align*}
% using Hoeffding's inequality we get
% \begin{align*}
%      \mathbb{P}\{|\mu_k-\bar\mu_k|\geq\sqrt{\frac{2\log t}{n}} \mid N_k(t-1) = n\} 
%      \leq 
%      2 t^{-4}
% \end{align*}
% Thus pluging back
% \begin{align*}
%     \mathbb{P}\{|\mu_k-\bar\mu_k|\geq\sqrt{\frac{2\log t}{N_k(t-1)}}\}
%     \leq
%     \sum_{n=1}^{t-1}
%     (2 t^{-4})
%     \leq
%     2 t^{-3}
% \end{align*}
% Using union bound over all $k\in\mathcal{K}$ we get
% $\mathbb{P}\{\mathcal{G}(t-1)^c\} \leq 2 d t^{-3}$.

To bound $\EE\!\left[\mathbbm{1}\{\mathcal{G}(t-1)^c\}\right] = \mathbb{P}\{\mathcal{G}(t-1)^c\}$, we first fix $k \in \mathcal{K}$ and note
\begin{align*}
    \mathbb{P}\!\left\{ \big|\mu_k - \bar\mu_k\big| \;\geq\; \sqrt{\tfrac{2\log t}{N_k(t-1)}} \right\}
    &=
    \sum_{n=1}^{t-1}
    \mathbb{P}\!\left\{ \big|\mu_k - \bar\mu_k\big| \;\geq\; \sqrt{\tfrac{2\log t}{n}} \,\middle|\, N_k(t-1) = n \right\}
    \mathbb{P}\{N_k(t-1) = n\}  \\
    &\leq
    \sum_{n=1}^{t-1}
    \mathbb{P}\!\left\{ \big|\mu_k - \bar\mu_k\big| \;\geq\; \sqrt{\tfrac{2\log t}{n}} \,\middle|\, N_k(t-1) = n \right\}.
\end{align*}
By Hoeffding's inequality,
\begin{align*}
     \mathbb{P}\!\left\{ \big|\mu_k - \bar\mu_k\big| \;\geq\; \sqrt{\tfrac{2\log t}{n}} \,\middle|\, N_k(t-1) = n \right\}
     \;\leq\; 2 t^{-4}.
\end{align*}
Plugging back, we obtain
\begin{align*}
    \mathbb{P}\!\left\{ \big|\mu_k - \bar\mu_k\big| \;\geq\; \sqrt{\tfrac{2\log t}{N_k(t-1)}} \right\}
    &\leq
    \sum_{n=1}^{t-1} 2 t^{-4}
    \;\leq\; 2 t^{-3}.
\end{align*}
Finally, applying a union bound over all $k \in \mathcal{K}$ yields
\[
    \mathbb{P}\{\mathcal{G}(t-1)^c\} \;\leq\; 2 d t^{-3}.
\]

\end{proof}

\begin{lemma}\label{lem:H2-sum}
Define
\[
h(t) = 
\begin{cases}
\mu^* - \mu_t, 
&1 \le t \le {d},\\[6pt]
2\sqrt{\dfrac{2\log t}{\,N_{K(t)}(t-1)\,}}, 
&t \ge {d}+1.
\end{cases}
\]
Then for any \(H\ge1\),
\[
\sum_{t=1}^H h(t)^2
\;\le\;
{d}\Bigl(\tfrac{1}{3} + 8\bigl(1+\log H\bigr)\log H\Bigr).
\]
\end{lemma}
\begin{proof}[Proof of Lemma~\ref{lem:H2-sum}]
Split the sum into three parts:
\[
\sum_{t=1}^H h(t)^2
=
\sum_{t=1}^{k^*} h(t)^2
+
\sum_{t=k^*+1}^{{d}} h(t)^2
+
\sum_{t={d}+1}^H h(t)^2.
\]

\noindent\textbf{Case 1: \(1\le t\le k^*\).}  By Lemma~\ref{lem:lip}, \(\mu^*-\mu_t\le r_{k^*}-r_t =  (k^*-t)/{d}\).  Hence
\[
\sum_{t=1}^{k^*}h(t)^2
=\sum_{t=1}^{k^*}(\mu^*-\mu_t)^2
\;\le\;
\frac{1}{{d}^2}\sum_{t=1}^{k^*}(k^*-t)^2
=\frac{(k^*-1)k^*(2k^*-1)}{6{d}^2}
\;\le\;
\frac{{k^*}^3}{3{d}^2}.
\]

\noindent\textbf{Case 2: \(k^*+1\le t\le {d}\).}  Since each \(\mu_t\ge0\) and \(\mu^*\le k^*/{d}\),
\[
\sum_{t=k^*+1}^{{d}}h(t)^2
=\sum_{t=k^*+1}^{{d}}(\mu^*-\mu_t)^2
\;\le\;
({d}-k^*)\,(\mu^*)^2
\;\le\;
({d}-k^*)\,\frac{{k^*}^2}{{d}^2}.
\]
Combining Cases 1 and 2,
\[
\sum_{t=1}^{{d}}h(t)^2
\;\le\;
\frac{{k^*}^3}{3{d}^2}
\;+\;
({d}-k^*)\frac{{k^*}^2}{{d}^2}
=
\frac{{k^*}^2}{{d}}
-\frac{2{k^*}^3}{3{d}^2}
\;\le\;
\sup_{x\in[1,{d}]}\Bigl(\tfrac{x^2}{{d}}-\tfrac{2x^3}{3{d}^2}\Bigr)
\;=\;
\frac{{d}}{3}.
\]

\noindent\textbf{Case 3: \(t\ge{d}+1\).}  
\[
\sum_{t={d}+1}^nH(t)^2
=4\sum_{t={d}+1}^H\frac{2\log t}{N_{K(t)}(t-1)}
\;\le\;
8\log\ H
\sum_{t={d}+1}^H\frac{1}{N_{K(t)}(t-1)}.
\]
Further we simply have
\[
\sum_{t={d}+1}^H\frac{1}{N_{K(t)}(t-1)}
\;\le\;
{d}\sum_{m=1}^{H-1} \frac{1}{m}
\;\le\;
{d}\left(1 + \log H\right).
\]
Hence
\[
\sum_{t={d}+1}^nH(t)^2
\;\le\;
8{d}\,\bigl(1 + \log H\bigr)\log H.
\]

Putting the three parts together gives the claimed bound.
\end{proof}

\begin{proof}[Proof of Lemma \ref{thm:stage2}]
Summing the bound from Lemma~\ref{lem:drift-t<kappa} over \(t=1,\dots,{d}\) gives
\[
\frac{1}{2}\EE\bigl[Q({d}+1)^2\bigr]-\frac{1}{2}\EE\bigl[Q(1)^2\bigr]
\;\le\;
\frac{{d}}{2}
\;-\;
\vedisc\,\EE\Bigl[\sum_{t=1}^{{d}}Q(t)\Bigr]
\;+\;
\EE\Bigl[\sum_{t=1}^{{d}}Q(t)\,(\mu^*-\mu_t)\Bigr].
\]
Similarly, summing Lemma~\ref{lem:drift-t>kappa} for \(t={d}+1,\dots,H\) yields
\[
\begin{aligned}
\frac{1}{2}\EE\bigl[Q(H+1)^2\bigr]
- \frac{1}{2}\EE\bigl[Q({d}+1)^2\bigr]
&\le
\frac{H-{d}}{2}
- \vedisc\,\EE\Bigl[\sum_{t={d}+1}^nQ(t)\Bigr]
+4{d} \sum_{t={d}+1}^H t^{-2}\\
&\quad
+2\,\EE\!\Bigl[\sum_{t={d}+1}^nQ(t)\sqrt{\tfrac{2\log t}{N_{K(t)}(t-1)}}\Bigr].
\end{aligned}
\]
Adding these two inequalities, noting \(Q(1)=0\) and \(Q(H+1)^2\ge0\), gives
\[
0
\;\le\;
\frac{H}{2}
- \vedisc\,\EE\Bigl[\sum_{t=1}^nQ(t)\Bigr]
+4{d} \sum_{t={d}+1}^H t^{-2}
+2\,\EE\!\Bigl[\sum_{t={d}+1}^nQ(t)\sqrt{\tfrac{2\log t}{N_{K(t)}(t-1)}}\Bigr]
+\EE\Bigl[\sum_{t=1}^{d} Q(t)(\mu^*-\mu_t)\Bigr].
\]
Rearranging and dividing by \(\vedisc\), and using the definition of \(h(t)\) from Lemma~\ref{lem:H2-sum}, yields
\[
\EE\Bigl[\sum_{t=1}^nQ(t)\Bigr]
\;\le\;
\frac{H}{2\vedisc}
+4\,\frac{{d} }{\vedisc}\sum_{t={d}+1}^H t^{-2}
+\frac{1}{\vedisc}\,\EE\Bigl[\sum_{t=1}^nQ(t)\,h(t)\Bigr].
\]
By Cauchy–Schwarz,
\[
\sum_{t=1}^nQ(t)\,h(t)
\;\le\;
\sqrt{\sum_{t=1}^nQ(t)^2}\;\sqrt{\sum_{t=1}^nH(t)^2}
\;\le\;
2^{1/4}\Bigl(\sum_{t=1}^nQ(t)\Bigr)^{3/4}
\sqrt{\sum_{t=1}^nH(t)^2},
\]
where the last step uses Lemma~\ref{lemma:q_pow_bnd}.  Invoking Lemma~\ref{lem:H2-sum},
\[
\sum_{t=1}^nQ(t)\,h(t)
\;\le\;
2^{1/4}\Bigl(\sum_{t=1}^nQ(t)\Bigr)^{3/4}
\sqrt{{d}\Bigl(\tfrac13+8\bigl(1+\log H\bigr)\log H\Bigr)}.
\]
Since \(x\mapsto x^{-2}\) is decreasing on \([\,{d},\infty)\), we can bound the sum by the corresponding integral:
\[
\sum_{t={d}+1}^H \frac{1}{t^2}
\;\le\;
\int_{{d}}^{H} x^{-2}\,\mathrm{d}x
=
\bigl[-\,x^{-1}\bigr]_{x={d}}^{x=H}
=
\frac{1}{d} - \frac{1}{H}
\;\le\;
\frac{1}{d}.
\]  
Substituting these bounds and dividing both sides by \(H\) yields
\[
\frac{1}{H}\EE\Bigl[\sum_{t=1}^nQ(t)\Bigr]
\;\le\;
\frac{1}{2\vedisc}
+\frac{4}{H\vedisc}
+2^{1/4}\frac{\sqrt{{d}}}{H^{1/4}\vedisc}
\sqrt{\tfrac13+8\bigl(1+\log H\bigr)\log H}
\Bigl(\frac{1}{H}\EE\bigl[\sum_{t=1}^nQ(t)\bigr]\Bigr)^{3/4}.
\]
Set \(x^4=\tfrac{1}{H}\EE[\sum_{t=1}^nQ(t)]\).  Then this implies
\begin{align*}
    x^4- b x^{3} -a
    \leq
    0
\end{align*}
where
\[
a = \frac{1}{2\vedisc}+\frac{4}{H\vedisc},
\quad
b = 2^{1/4}\frac{\sqrt{{d}}}{H^{1/4}\vedisc}
\sqrt{\tfrac{1}{3}+8\bigl(1+\log H\bigr)\log H}.
\]
Applying Lemma~\ref{lemma:abX_lemma} with \(d=4\) gives
\[
\frac{1}{H}\EE\Bigl[\sum_{t=1}^nQ(t)\Bigr]
\;\le\;
8\,a + 8\,b^4
=
\frac{4}{\vedisc}
+\frac{32}{H\vedisc}
+16\,\frac{{d}^2}{H\vedisc^4}
\bigl(\tfrac13+8\bigl(1+\log H\bigr)\log H\bigr)^2.
\]
    
\end{proof}
\subsection{Proof of Theorem~\ref{thm:main_known}}\label{app:Proof-thm:main_known}
Since Lemma~\ref{thm:log(H)} and \ref{thm:stage2} both apply, we take the tighter of the two bounds: use Lemma~\ref{thm:log(H)} for \(2 \le H \le n_{\mathrm{stage}}\) and Lemma~\ref{thm:stage2} for \(H \ge n_{\mathrm{stage}}\). This yields a bound that holds uniformly for all \(H \ge 2\). Set
\[
  n_{\mathrm{stage}} \;=\; e^{15.5}\,\velb^{-4.5}.
\]

\noindent\textbf{Case 1: \(2 \le H \le n_{\mathrm{stage}}\).}
By Theorem~\ref{thm:log(H)},
\begin{align*}
    \frac{1}{H}\sum_{t=1}^H \EE\!\bigl[Q(t)\bigr]
    &\le  
    \frac{4}{\vedisc}\Bigl(1+4 d\,\log H\Bigr)
    \\
    &\stackrel{(a)}{\le}
    \frac{6}{\velb}
    + \frac{96\bigl(15.5 + 4.5\log(1/\velb)\bigr)}{\velb^{2}}
    \\
    &\le
    \frac{6}{\velb}
    + \frac{1488}{\velb^{2}}
    + \frac{432\,\log(1/\velb)}{\velb^{2}}.
\end{align*}
For \((a)\), we use: (i) \(\vedisc \ge \tfrac{2}{3}\,\velb\), hence \(\tfrac{4}{\vedisc} \le \tfrac{6}{\velb}\); (ii) \(d = \lceil 3/\velb\rceil \le 4/\velb\); and (iii) \(H \le n_{\mathrm{stage}} = e^{15.5}\velb^{-4.5}\), so \(\log H \le 15.5 + 4.5\log(1/\velb)\).

\noindent\textbf{Case 2: \(H \ge n_{\mathrm{stage}}\).}
We first state a useful lemma.

\begin{lemma}\label{lem:decreasing-for-H-large}
For all real \(x\ge e^{8}\), define
\[
f(x)\;=\;\frac{1}{x}\Bigl(\tfrac13+8\,(1+\log x)\,\log x\Bigr)^2.
\]
Then
\[
f(x)\;\le\;\frac{82(\log x)^4}{x},
\]
and the function \(x\mapsto 82(\log x)^4/x\) is strictly decreasing on \([e^{8},\infty)\).
\end{lemma}

\begin{proof}[Proof of Lemma~\ref{lem:decreasing-for-H-large}]
Since \(x\ge e^{8}\) we have \(\tfrac{1}{8}\log x\ge 1 \), so
\[
1+\log x\;\le\;\tfrac{9}{8}\log x
\quad\text{and}\quad
\tfrac13\;\le\;\tfrac{1}{192}(\log x)^2.
\]
Therefore
\[
\tfrac13 +8(1+\log x)\,\log x
\;\le\;(9+\tfrac{1}{192})\,(\log x)^2,
\]
and thus
\[
f(x)\;\le\;
\frac{82(\log x)^4}{x}.
\]

It remains to check that \(g(x)=82(\log x)^4/x\) is decreasing for \(x\ge e^{8}\).  A direct derivative gives
\[
g'(x)
=82\left(\frac{\,4(\log x)^3\cdot\frac1x - (\log x)^4\cdot\frac1x\,}{x}\right)
=\frac{82(\log x)^3}{x^2}\bigl(4-\log x\bigr).
\]
For \(x\ge e^{8}> e^4\) we have \(\log x>4\), so \(4-\log x<0\), and hence \(g'(x)<0\).  This shows \(g\) is strictly decreasing on \([e^{8},\infty)\), completing the proof.

\end{proof}

 Combining Lemma~\ref{thm:stage2} with Lemma~\ref{lem:decreasing-for-H-large} yields
\begin{align*}
    \frac{1}{H}\sum_{t=1}^H \EE\bigl[Q(t)\bigr]
    &\le
    \frac{4}{\vedisc}
    + \frac{32}{H\,\vedisc}
    + 16\,\frac{d^2}{H\,\vedisc^4}
    \Bigl(\tfrac13 + 8\,(1+\log H)\,\log H\Bigr)^2
    \\&\le_{(a)}
    \frac{4}{\vedisc}
    + \frac{32}{H\,\vedisc}
    + 16\,\frac{d^2}{\vedisc^4}\,82\,\frac{(\log H)^4}{H}
    \\&\le_{(b)}
    \frac{6}{\velb}
    +\frac{48}{e^{15.5}}
    + 16\,\frac{d^2}{\vedisc^4}\,82\,\frac{(\log H)^4}{H}
    \\&\le_{(c)}
    \frac{6}{\velb}
    +\frac{48}{e^{15.5}}
    +\frac{106272}{e^{15.5}}\,
      \frac{\bigl(15.5 + 4.5\log(1/\velb)\bigr)^4}{\velb^{1.5}}
    \\&\le
    \frac{6}{\velb}
    +\frac{8\cdot 106273}{e^{15.5}}\,
      \frac{{15.5}^4 + 5^4(\log(1/\velb))^4}{\velb^{1.5}}
    \\&\le
    \frac{6}{\velb}
    +\frac{9105}{\velb^{1.5}}
    +\frac{99\,(\log(1/\velb))^4}{\velb^{1.5}}.
\end{align*}
Here, \((a)\) uses Lemma~\ref{lem:decreasing-for-H-large}; \((b)\) uses \(\vedisc \ge  \tfrac{2}{3}\,\velb\) and \(H \ge n_{\text{stage}}\); and \((c)\) uses \(H \ge n_{\text{stage}} \ge e^8\) so that, by the second part of Lemma~\ref{lem:decreasing-for-H-large}, \(\tfrac{(\log H)^4}{H} \le \tfrac{(\log n_{\text{stage}} )^4}{n_{\text{stage}} }\), together with \(\vedisc \ge \tfrac{2}{3}\,\velb\) and \(d \le 4/\velb\).

Combining the two cases gives
\begin{align*}
    \frac{1}{H}\sum_{t=1}^H \EE\bigl[Q(t)\bigr]
    &\le
    \max\!\Biggl\{
    \frac{6}{\velb}
    + \frac{1488}{\velb^2}
    + \frac{432\,\log(1/\velb)}{\velb^2}
    \;,\;
    \frac{6}{\velb}
    +\frac{9105}{\velb^{1.5}}
    +\frac{99\,(\log(1/\velb))^4}{\velb^{1.5}}
    \Biggr\}.
\end{align*}
To streamline the analysis, we treat the large- and small-\(\velb\) regimes separately and begin with the following lemma.

\begin{lemma}\label{lem:log4}
For all \(x \in (0,1]\), we have
\[
  \frac{\bigl(\log(1/x)\bigr)^{4}}{x^{3/2}}
  \;\le\;
  \frac{11\,\log(1/x)}{x^{2}}.
\]
\end{lemma}

\begin{proof}[Proof of Lemma~\ref{lem:log4}]
simplification gives
  \[
    \frac{\bigl(\log(1/x)\bigr)^{4}}{x^{1.5}}
    \;\le\;
    \frac{11\,\log(1/x)}{x^{2}}
    \Leftrightarrow
    x^{\tfrac{1}{6}} \log\tfrac{1}{x} \leq 11^{\tfrac{1}{3}}
  \]
  Write \(t = \tfrac{1}{6}\log(1/x)\), so \(t \ge 0\) and \(x = e^{-6t}\).  The desired inequality
  is equivalent to
  \[
    f(t) \;\stackrel{\text{def}}{=}\;t\,e^{-t}\;\le\;\frac{11^{1/3}}{6}
    \quad\text{for all }t\ge0.
  \]
  Now
  \[
    f'(t)
    = {e^{-t}}\,(1 - t),
  \]
  so \(f'(t)=0\) only at \(t=1\).  Checking these and the zero and the limit at infinity:
  \[
    f(0)=0,
    \quad
    f(1)=e^{-3} < \frac{11^{1/3}}{6},
    \quad
    \lim_{t\to\infty} f(t)=0.
  \]
  Hence \(f(t)\le f(1)<\frac{11^{1/3}}{6}\) for all \(t\ge0\), which proves the lemma.

\end{proof}

\noindent\textbf{Small \(\velb\) (i.e., \(\velb \le e^{-3}\)).}
In this regime,
\[
  \frac{1}{\velb}\ge e^{3}, \qquad \log\!\frac{1}{\velb}\ge 3,
\]
and hence
\[
  \frac{1}{\velb}\log\!\frac{1}{\velb}\ge 3e^{3},
  \qquad
  \frac{1}{\sqrt{\velb}}\log\!\frac{1}{\velb}\ge 3e^{3/2}.
\]
Consequently,
\begin{align}\label{eqn:2222}
  \frac{1}{H}\sum_{t=1}^{H}\EE\bigl[Q(t)\bigr]
  &\le_{(a)}
  \frac{\log(1/\velb)}{\velb^{2}}
  \cdot
  \max\!\Biggl\{
    \frac{6}{3e^{3}}
    + \frac{1488}{3}
    + 432
    \;,\;
    \frac{6}{3e^{3}}
    + \frac{9105}{3e^{3/2}}
    + 99\cdot 11
  \Biggr\}
  \nonumber\\&\le
  1767\,\frac{\log(1/\velb)}{\velb^{2}}.
\end{align}
Here, \((a)\) invokes Lemma~\ref{lem:log4}.

\noindent\textbf{Large \(\velb\) (i.e., \(e^{-3}\le \velb \le 1\)).}
Using \(\log(1/\velb)\le 3\), we obtain
\begin{align}\label{eqn:11111}
    \frac{1}{H}\sum_{t=1}^H \EE\bigl[Q(t)\bigr]
    \le
    \max\!\Biggl\{
    \frac{6}{\velb^2}
    + \frac{1488}{\velb^2}
    + \frac{432\cdot 3}{\velb^2}
    \;,\;
    \frac{6}{\velb^2}
    + \frac{9105}{\velb^{2}}
    + \frac{99\cdot 11 \cdot 3}{\velb^{2}}
    \Biggr\}
    \;=\; \frac{12378}{\velb^2}.
\end{align}

 The proof is an immediate consequence of combining \eqref{eqn:2222} and \eqref{eqn:11111}.

\end{document}